\documentclass[12pt]{article}

\usepackage[T1]{fontenc}
\usepackage{lmodern,microtype}
\usepackage[margin=1in]{geometry}
\usepackage{setspace}
\usepackage{amsmath,amsfonts,amssymb,amsthm}
\usepackage{graphicx}
\usepackage[authoryear,round,comma]{natbib}

\theoremstyle{plain}

\newtheorem{theorem}{Theorem}
\newtheorem{prop}{Proposition}

\newtheorem{lemma}[theorem]{Lemma}
\newtheorem{cor}{Corollary}

\theoremstyle{definition}
\newtheorem{definition}{Definition}
\newtheorem{ass}{Assumption}
\newtheorem{example}{Example}
\newtheorem{remark}{Remark}

\usepackage{tikz} 
\usepackage{pgfplots}
\pgfplotsset{compat=1.5}
\usepgfplotslibrary{dateplot}
\usepackage{tikz-cd}
\usetikzlibrary{patterns, shapes, calc}
\usetikzlibrary{positioning,arrows} 
\usetikzlibrary{decorations.pathreplacing} 
\usepackage{tikzsymbols} 
\usepackage{blindtext} 
\usepackage{dsfont}
\usepackage{enumitem,booktabs}
\usepackage{algorithm}
\usepackage{algorithmicx}
\usepackage{algpseudocode}
\usepackage{multirow}

\usepackage{subcaption}
\renewcommand{\P}{\mathbb{P}}
\newcommand{\E}{\mathbb{E}}
\newcommand{\Var}{\text{Var}}
\newcommand{\Cov}{\text{Cov}}

\newcommand{\rk}{\text{rk}}

\usepackage[
  colorlinks,citecolor=blue,linkcolor=blue,
  urlcolor=blue,pagebackref
]{hyperref}

\title{The Linear Programming Approach to Partially Identified Models}

\author{Andrei Voronin\thanks{I thank three anonymous referees, Andres Santos, Denis Chetverikov, Rosa Matzkin, Jinyong Hahn, Bulat Gafarov, Tim Armstrong, Kirill Ponomarev, Zhipeng Liao, Manu Navjeevan, Bohdan Salahub, Niu He, Richard Grigorian, Shuyang Sheng, as well as all the participants of the 2024 California Econometrics Conference and the 2024 European Winter Meeting of the Econometric Society for the valuable discussions and criticisms.}\\
Department of Economics, UCLA\\
\texttt{avoronin@ucla.edu}}
\date{}
\begin{document}
\maketitle

%
\begin{abstract}
In many partially identified models, sharp bounds on structural or causal parameters are values of linear programs (LPs) with estimated coefficients. We introduce a debiased penalty estimator of the LP value. It is pointwise $\sqrt n$-consistent whenever the population feasible set is nonempty and bounded. Although no estimator is uniformly consistent over the unrestricted collection of distributions, ours retains a one-sided guarantee: the extent to which it can overestimate a sharp lower bound vanishes uniformly. Under a condition ruling out weak identification of the bound, it is also uniformly consistent. We develop two inference procedures exploiting LP structure: a simple pointwise procedure with exact asymptotic coverage, and a uniformly valid procedure. In simulations calibrated to \citet{holmes2011diffusion}, our estimator avoids the anticonservative bias of unregularized benchmark estimators and the conservativeness of the set expansion estimator. Both inference procedures deliver reliable coverage, small excess length, and orders-of-magnitude computational gains over test-inversion benchmarks.
\end{abstract}


\clearpage
\section{Introduction}
In many partially identified models, sharp bounds on structural or causal parameters equal the values of linear programs (LPs) whose coefficients are identified functionals of the underlying distribution. Examples include nonparametric discrete choice models (\citet{tebalditorgovitskyyang2023}), generalized IV models (\citet{santos}), revealed preference restrictions (\cite{klinetartari}), intersection bounds (\citet{honoreadriana2006}), dynamic discrete choice panels (\citet{honoretamer2006}), and shape restrictions (\citet{MP2000}).

In these models, the sharp identified set for a feature $\theta$ of the underlying distribution is a polytope \( \Theta_I = \{\theta \in \mathbb{R}^d:M\theta \geq c\}\), and the parameter of interest is the sharp lower bound on a linear function $p'\theta$,
\begin{align}\label{lp_intro}
    V(\mu_0) = \underset{M\theta \geq c}{\min} ~ p'\theta,  
\end{align}  
where \( \mu_0 = (p,M,c)\) is potentially unknown but estimated at $1/\sqrt{n}$ rate. This optimization problem is non-regular when the model is rich enough that some linear functions of \( \theta \) are nearly or exactly point-identified. In such cases, existing estimators of $V(\mu_0)$ are either inconsistent or rate-conservative, creating an undesirable tradeoff in empirical work: richer models provide tighter bounds but can complicate estimation.

We approach this problem by exploiting the local linear structure of an LP optimum. To illustrate, suppose that the minimum in \eqref{lp_intro} is attained at $\theta^*$ and that $d$ linearly independent inequalities indexed by $A$ bind there. The system $M_A\theta^*=c_A$ then determines $\theta^*$ and hence the bound $p'\theta^*$. Rather than estimating the entire identified set $\Theta_I$, we therefore seek to select a subset of inequalities that bind jointly at some population optimum and are sufficient to identify $V(\mu_0)$. We then use the resulting linear system to estimate the bound. 

Learning the identity of these inequalities directly from the sample LP is difficult due to irregularity of \eqref{lp_intro}. We first regularize the estimation problem by relaxing it to an unconstrained problem whose objective penalizes constraint violations. The resulting penalized value is consistent for $V(\mu_0)$, but may be rate-conservative. More importantly, however, the penalized problem's solution selects `nice' binding restrictions. Our \textit{debiased penalty estimator} retains this selected structure but removes the penalization bias from the reported value, achieving $\sqrt n$-consistency, pointwise in the underlying distribution, whenever $\Theta_I$ is nonempty and bounded. As in post-Lasso estimation (\citet{bc2013}), regularization is used to select relevant structure, after which a second step exploits that structure without retaining the regularization bias in the reported estimator. By contrast, the set expansion estimator in the spirit of \citet{CHT} retains its regularization bias in the reported endpoint and is rate-conservative; our simulations in Section \ref{mc_section} show that this can also translate into substantial finite-sample conservative bias. 

Uniformly, our estimator retains a useful one-sided validity property even over the unrestricted collection of DGPs, on which no estimator is uniformly consistent for $V(\mu_0)$. It may be conservative along some sequences of DGPs, but any overestimation of the sharp lower bound is uniformly at most of order $w_n/\sqrt n$ for a slowly diverging sequence $w_n$.

For any estimator to be uniformly consistent for $V(\mu_0)$, weak identification of the bound must be ruled out. To illustrate, note that linear IV models are nested in \eqref{lp_intro}. When $\Theta_I = \{\theta^*\}$, some sequences of DGPs with increasingly weak instruments preclude uniformly consistent estimation of $\theta^*$  (\citet{staiger1997instrumental}), and hence of the LP value $p'\theta^*$. We introduce the `$\delta$-condition' that rules out analogous sequences in general LPs. It requires that, at each vertex of $\Theta_I$, all square matrices formed by linearly independent binding constraints have their smallest singular values bounded below by $\delta>0$. In exactly identified linear IV models with instruments $Z$ and regressors $X$, this reduces to requiring $\Cov(X,Z)$ to remain bounded away from singularity. Importantly, the condition permits point identification of linear functions of $\theta$, as well as overidentifying restrictions and solution multiplicity at the optima of \eqref{lp_intro}. Moreover, smaller values of $\delta$ allow increasingly ill-conditioned problems, and the corresponding collections of DGPs exhaust the unrestricted collection as $\delta\downarrow0$. Under the $\delta$-condition, the debiased penalty estimator understates $V(\mu_0)$ by at most order $1/\sqrt n$ uniformly and is therefore uniformly consistent for $V(\mu_0)$. 

The selected set of inequalities can also be used for inference. Combining it with coefficient estimates from an independent fold in an orthogonal estimating function, we obtain an asymptotically normal version of our estimator and construct confidence regions with exact asymptotic coverage at any fixed probability measure. In contrast, existing LP-based procedures either require additional regularity conditions (\citet{gafarov2024simple}) or deliver asymptotically conservative inference (\citet{chorussel2023}). Our pointwise procedure is computationally simple even in larger problems. Its reliance on a single selected set of constraints, however, creates a uniformity problem when some population systems $M_A\theta=c_A$ induce solutions that fail feasibility or optimality in \eqref{lp_intro} by margins comparable to sampling error. To address this, we supplement the selected system with well-conditioned systems whose solutions cannot be rejected as optimal. The lower envelope of their confidence bounds yields uniformly valid inference under the $\delta$-condition, and the resulting endpoint remains pointwise $\sqrt n$-consistent for $V(\mu_0)$. Applying our procedures to $p$ and $-p$ also yields two-sided confidence regions for $p'\Theta_I$ and tests of the null that $p'\theta$ is point-identified.

Finally, we conduct Monte Carlo simulations based on \citet{holmes2011diffusion}, treating moment inequalities estimated from his replication data as the population problem and constructing irregular variants around it. In regular designs, our estimator performs similarly to the plug-in estimator; as the geometry becomes irregular, it avoids the anticonservative bias of the minimum-violation estimator used by Holmes and the conservativeness of the set-expansion estimator. Both inference procedures perform well in terms of coverage, with excess length comparable to, and often smaller than, that of existing methods, while remaining two to four orders of magnitude faster than test-inversion benchmarks.

\subsection*{Relationship to literature}
There is a rich literature on estimation and inference for the value of an LP with estimated coefficients. \citet{bhattacharya2009inferring}, \citet{freyberger2015identification}, \citet{fang2023inference}, \citet{andrews2023} and \citet{semenova2023adaptive} study inference in LPs with a known or a `fixed-design' matrix $M$. These methods do not extend to our setting: while $V(\mu_0)$ is Lipschitz in feasibility-preserving $c$ perturbations, it need not be continuous in $M$ perturbations (see \citet{robinson1975stability}). \citet{santos} propose a consistent set-expansion estimator. \citet{gafarov2024simple} develops uniform inference for LP value under uniform Linear Independence Constraint Qualification (LICQ) and Slater’s condition (SC). Section \ref{limitations_sec} presents examples where LICQ and SC fail, making this approach inapplicable. \citet{chorussel2023} introduce non-vanishing random perturbations to obtain uniformly valid conservative inference for $V(\mu_0)$. Also related is the literature on testing LP feasibility (\citealp{syrgkanis2017inference}; \citealp{bai2022testing}; \citealp{deb2023revealed}). 

There are several general approaches to inference in partially identified models that are amenable to our setup. These include procedures that provide confidence regions for the entire identified set $\Theta_I$ (\citet{CHT}, \citet{romano2010inference}) and, when $M\theta \geq c$ arises from moment inequalities, test-inversion procedures for inference on a partially identified parameter $\theta_0\in\Theta_I$ (\citet{andrews2010inference}, \citet{chernozhukov2019inference}) or its subvector $p'\theta_0$ (\citet{romanoshaikh2008}, \citet{bugni2017inference}, \citet{kaido2019confidence}). The LP approach targets the endpoints of $p'\Theta_I$ directly. A lower confidence bound for $V(\mu_0)$ yields a one-sided confidence region for $p'\Theta_I$, while applying our procedures to $p$ and $-p$ yields two-sided confidence regions, see Remark \ref{remark_cht_ci}. Both are automatically valid for any $p'\theta_0\in p'\Theta_I$. By avoiding projection of confidence regions for $\theta_0$ or $\Theta_I$, our procedure avoids the associated projection conservativeness and, by exploiting the LP formulation \eqref{lp_intro}, remains computationally tractable in larger problems, see Remark \ref{remark_comput}. 

Our primary focus, however, is point estimation of $V(\mu_0)$, which is a distinct---and often central---objective. This is particularly clear when the economic question concerning $p'\theta$ is one-sided. In \citet{holmes2011diffusion}, for example, a lower bound on the density-economies that exceeds direct trucking savings points to additional logistical benefits, such as faster responses to demand shocks. In the setting of \citet{santos}, an estimated lower bound on a policy-relevant treatment effect can be judged against an assessment of implementation costs. In optimal-assignment problems such as \citet{bhattacharya2009inferring}, $-V(\mu_0)$ equals the maximum attainable welfare, which can be compared with the status quo. In such applications, it is natural to follow the empirical practice common in point-identified settings by reporting both an estimate of the economically relevant magnitude and a confidence region quantifying its sampling uncertainty. Accurate estimation and valid inference therefore remain distinct and complementary objectives under partial identification.

The remainder of the paper is organized as follows. Section \ref{section_lp} formalizes the setup and motivates our approach by illustrating the limitations of standard regularity conditions with examples. Section \ref{section_consistency} develops the estimation and inference procedures, and Section \ref{section_uniformity} presents uniform asymptotic theory. Section \ref{mc_section} presents Monte Carlo simulations. 

\section*{Notation}
\small
All vectors are column vectors. $\text{vec}(B)$ stacks the columns of $B$, and $\text{vec}(B_1,\ldots,B_k)$ stacks vectorizations in the given order. $M'$ denotes the transpose of $M$ if $M$ is a matrix, and its complement if $M$ is a set. $2^A$ is the powerset of $A$, and $[n]=\{j\}^n_{j=1}$. For $r, u \in \mathbb{R}^d$, $r^+ \equiv (\max \{r_j, 0\})_{j \in [d]}$, while $r > u$ and $r \geq u$ mean $r_i > u_i ~ \forall i \in [d]$ and $r_i \geq u_i ~ \forall i \in [d]$ respectively. $\otimes$ is the Kronecker product. For $M \in \mathbb{R}^{q\times d}$ and $A \subseteq [q]$, $M_A \in \mathbb{R}^{|A| \times d}$ is the submatrix of the rows of $M$ with indices in $A$; $\mathcal{R}(M)$ is the range of $M$, with $\rk{(M)} \equiv\text{dim}(\mathcal{R}(M))$, and $\mathrm{sv}_k(M)$ is its $k-$th largest singular value; $\lVert M\rVert \equiv \mathrm{sv}_1(M)$, and $M^{-\prime} \equiv (M^{-1})'$. $\iota_d \in \mathbb{R}^d$ is the vector of ones, $I_d \in \mathbb{R}^{d\times d}$ is the identity matrix, and the subscript is dropped occasionally. $\text{Int}(S)$ is the interior of $S \subseteq \mathbb{R}^d$. $w_n \to \infty$ w.p.a.1 means $\forall C > 0,$ $ \lim _{n\to\infty} \mathbb{P}[w_n > C] = 1$. We set $\inf \emptyset \equiv +\infty$, $\sup \emptyset \equiv - \infty$. 

\section{Setup}
\label{section_lp}
\normalsize

This section formalizes the setup and motivates our approach. The LP value function is
\begin{align}\label{lp_init}
    V(\mu) \equiv \underset{\tilde{M} \theta \geq \tilde{c}}{\inf} \tilde{p}'\theta, 
\end{align}
where $\tilde{p} \in \mathbb{R}^d$, $\tilde{M} \in \mathbb{R}^{q\times d}$ and $\tilde{c} \in \mathbb{R}^q$ are generic LP coefficients\footnote{Note that \eqref{lp_init} does not rule out equality constraints, as $R\theta = b$ if and only if $R\theta\geq b \text{ and } -R\theta \geq -b$.}, collected in $\mu = \text{vec}(\tilde{p},\tilde{M},\tilde{c}) \in \mathbb{R}^m$, with $m=d+qd+q$. Their population values are collected in $\mu_0=\text{vec}(p,M,c)$, and are estimated via $\hat{\mu}_n = \text{vec}(\hat{p}_n, \hat{M}_n, \hat{c}_n)$. The value of interest is $V(\mu_0)$. 

We denote the constraint set of \eqref{lp_init} by $\Theta_I(\mu)$ and omit the argument when $\mu_0$ is concerned. In the context of Section \ref{mc_section} and other applications (e.g. \citet{santos}), the set $\Theta_I$ is the identified set for an unobserved feature $\theta$ of the true distribution. 
\begin{ass}[Pointwise setup]
    At the fixed $\mu_0$: i) $\Theta_I \ne \emptyset$; ii) $\Theta_I \subseteq \Theta$ for a known compact $\Theta \subseteq \mathbb{R}^d$; iii) The estimator satisfies $\lVert \hat{\mu}_n - \mu_0\rVert = O_p(1/\sqrt{n}).$
\end{ass}
In partially identified models Assumption 1.i typically means that there exists a distribution consistent with the identifying restrictions. Assumption 1.ii usually holds in applications\footnote{We thank an anonymous referee for noting that our main results extend to the case where $\Theta$ is unknown, with boundedness retained as a regularity condition, see Supplemental Appendix \ref{unknown_Theta}. When structural restrictions give a known $\Theta$ (e.g. \citet{santos}), imposing $\theta\in\Theta$, as in the main text, can sharpen estimated bounds.}, for example under bounded outcomes in \citet{MP2000}. If $\Theta$ is a polytope, its defining inequalities are included as deterministic components of both $\mu_0$ and $\hat{\mu}_n$. The rate in Assumption 1.iii is typically warranted by the CLT and the Delta-Method\footnote{We focus on the $\sqrt{n}-$case for expositional simplicity, but our results generalize to any rate.}. 

The following primal and dual solution sets will prove useful in our discussion:
\begin{align*}
    \mathcal{S}(\mu) \equiv \underset{\tilde{M} \theta \geq \tilde{c}}{\arg \min} ~ \tilde{p}'\theta, \quad \Lambda(\mu) \equiv \underset{\tilde{M}'\lambda = \tilde{p}, ~ \lambda \geq 0}{\arg \max} ~ \tilde{c}' \lambda. 
\end{align*}
Equivalently, $\Lambda(\mu_0)$ is the set of optimal Lagrange multiplier vectors associated with the inequalities $M\theta\geq c$. Under Assumption 1, the estimand $V(\mu_0)$ is finite, $\Theta_I$ and $\mathcal{S}(\mu_0)$ are non-empty compact sets, and $\Lambda(\mu_0)$ is non-empty, although it may be unbounded. 
\subsection{Limitations of Standard Regularity Conditions}\label{limitations_sec}
To motivate our approach, we discuss commonly imposed regularity conditions and provide simple examples of empirical setups in which they fail. These examples are intentionally minimal; such issues are likely to arise even more readily in complex models, where the structure of $\Theta_I$ is harder to control. Assumption 1 imposes none of these conditions.
\begin{definition}
    Slater's condition (SC) is the assertion that $\text{Int}(\Theta_I) \ne \emptyset$.\footnote{We give simplified versions of assumptions here for ease of exposition. In the presence of `true' equalities, SC should be stated in terms of relative interior, allowing for point-identification along them. We also abstract from zero rows in $M$. LICQ condition below should similarly be restated to account for equalities, see \citet{gafarov2024simple}.}
\end{definition}
SC rules out point-identification of any linear function of $\theta$. In particular, it precludes point-identification of the target $p'\theta$ and point-identification of $\theta$, i.e. $|\Theta_I| = 1$. Existing methods based on the LP formulation \eqref{lp_init} with an estimated $M$ either rely on SC (Condition 1 in \citet{gafarov2024simple}) or introduce non-vanishing conservative perturbations to ensure that it holds with probability 1 (\citet{chorussel2023}). \citet{CHT} discuss SC as a sufficient condition for the consistency of criterion-based estimators of $\Theta_I$ in moment-inequality models (see p.14 in CHT and Section \ref{sec: available_estim}). 

\begin{example}[Interval-valued outcome]\label{example_interval}
Models with interval-valued outcomes (see \citet{manski2002inference}) result in LP sharp bounds on the parameters of interest. Let $Y^* = \theta_1 + \theta_2 D + \varepsilon$, where $Y^* \in \mathbb{R}$ is unobserved, $\theta = (\theta_1, \theta_2) \in \mathbb{R}^2$ are parameters, $D \in \{0,1\}$ is an endogenous treatment, and $Z \in \{0,1\}$ is an IV satisfying $\E[\varepsilon|Z] = 0$. Suppose that instead of $Y^*$, one observes bounds $Y_l \leq Y^* \leq Y_u$ a.s., and imposes an a priori restriction $|\theta_i| \leq \overline{\theta}$ for a known $\overline{\theta} > 0$ and $i = 1,2$. The sharp lower bound on $\theta_2$ is
\begin{align*}
    V(\mu_0) = \inf_{\theta \in [-\overline{\theta},\overline{\theta}]^2} ~ \theta_2 \quad \text{s.t. } \begin{cases}
        \E[Y_u | Z = z] \geq \theta_1 + \theta_2 \E[D | Z = z],\\
        \E[Y_l | Z = z] \leq \theta_1 + \theta_2 \E[D | Z = z],
    \end{cases} \quad z \in \{0,1\},
\end{align*}
where the canonical LP form is $V(\mu_0) = \inf_{\theta} ~ p' \theta$ subject to $M \theta\geq c$, with
\begin{align*}
    p &= \begin{pmatrix} 0 \\ 1 \end{pmatrix}, \quad
M = \left(
\begin{array}{cc}
1 & \mathbb{E}[D | Z=0] \\
-1 & -\mathbb{E}[D|Z=0] \\
1 & \mathbb{E}[D|Z=1] \\
-1 & -\mathbb{E}[D|Z=1] \\
\multicolumn{2}{c}{I_2} \\
\multicolumn{2}{c}{-I_2}
\end{array}
\right), \quad
    c = \begin{pmatrix}
        \E[Y_l|Z=0] \\
        -\E[Y_u|Z=0] \\
        \E[Y_l|Z=1] \\
        -\E[Y_u|Z=1] \\
        -\overline{\theta} \iota_2 \\
        -\overline{\theta} \iota_2
    \end{pmatrix}.
\end{align*}
In this case SC fails if $Y^*$ is measured exactly conditional on some instrument level $z = 0,1$, i.e. $\E[Y_u - Y_l|Z=z] = 0$, since $\E[Y^* |Z=z] = \theta_1 + \E[D|Z=z]\theta_2$ is then point-identified. Alternatively, SC fails if $\E[Y_l|Z = 1] = \E[Y_u|Z = 0]$, and the instrument relevance condition fails, i.e. $\E[D|Z = 1] = \E[D|Z=0]$, because the second and the third inequalities in $M$ point-identify $\theta_1 + \E[D|Z = 0]\theta_2$. In comparison, Assumption 1 holds if the model is consistent with the observed distribution of $(Y_l, Y_u, D, Z)$, and $M, c$ are estimated via sample analogues under standard conditions.
\end{example}

Even a near-failure of SC, when the true identified set $\Theta_I$ is `thin' relative to the sampling error, can be problematic for the methods that rely on it in finite samples. In Example \ref{example_interval}, this arises if the interval measurement of $Y^*$ is too precise conditional on some level of $Z$, or if $\E[Y_u | Z = 0] \approx \E[Y_l | Z = 1]$ and $Z$ is a weak instrument. Section \ref{mc_section} provides simulation evidence illustrating the consequences. Estimators that rely on SC thus face an undesirable tradeoff: higher identification power comes with poorer estimation quality.

\begin{definition}\label{def_licq}
    Linear independence constraint qualification (LICQ) is the assertion that, if at $\theta \in \mathcal{S}(\mu_0)$ the constraints indexed by $A \subseteq [q]$ bind, then $M_A$ has full row rank.
\end{definition}
The method of \citet{gafarov2024simple} relies on LICQ, while \citet{chorussel2023} introduce non-vanishing conservative perturbations to ensure that it holds with probability $1$. LICQ rules out overidentifying constraints at the optimum, which may be difficult to justify based on economic theory when many potentially parallel inequality constraints are present, as in \citet{holmes2011diffusion}. Near-failures of LICQ may similarly impair the finite-sample performance of procedures that rely on it, see Section \ref{mc_section}. LICQ may fail even in simpler scenarios, as in Example \ref{example_licq_failure}.
\begin{example}[Failure of LICQ]\label{example_licq_failure}
    Consider estimating the parameter $\max\{0,\E[X]\}$ under an a priori assumption that $\E[X] \leq \overline{\theta}$ for some $\overline{\theta} > 0$. One can view it as the value of a simple LP, $\max\{0, \mathbb{E}[X]\} = \min_{\theta \leq \overline{\theta}} \theta  ~ \text{s.t.} ~ \theta \geq 0, \theta \geq \mathbb{E}[X]$. To map it into \eqref{lp_init}, define
    \begin{align*}
        p = 1, \quad M = (
            1, 1, -1)', \quad c = (0 , \E[X] , -\overline{\theta})'.
    \end{align*}
    If $\mathbb{E}[X] = 0$, we have $\mathcal{S}(\mu_0) = \{0\}$, and two inequality constraints bind at the optimum, $\theta = \E[X] = 0$, so that LICQ fails. This corresponds to the parameter-on-the-boundary problem (see \citet{andrews1999estimation, andrews2000} and \citet{santos2019}). In contrast, Assumption 1 only requires $\E[X^2] < \infty$ if $c$ is estimated via a sample analogue based on an i.i.d. sample. 
\end{example}
\begin{definition} No flat faces condition (NFF) is the assertion that $|\mathcal{S}(\mu_0)| = 1$.
\end{definition}
The notion of flat faces thus corresponds to $|\mathcal{S}(\mu_0)| \ne 1$, i.e. the situation in which the bound on the target parameter $p'\theta$ is attained at multiple partially identified features $\theta$.
\setcounter{example}{0}
\begin{example}[cont'd]
    Suppose that one wishes to find a sharp lower bound on the average outcome, $\E[Y^*]$. This corresponds to $p = (1, \E[D])'$, with $M$ and $c$ as before. In case $\E[D|Z = 0] = \E[D|Z=1]$, the gradients of inequalities 1-4 in $M$ are parallel to $p$. If the model is consistent with the data and $\overline{\theta}$ is large enough, then $V(\mu_0) = \max_{z = 0,1} \E[Y_l|Z =z]$, attained by any $\theta \in [-\overline{\theta}, \overline{\theta}]^2$ on the hyperplane $p'\theta = V(\mu_0)$. Therefore, NFF fails.
\end{example}

Assumption 1 does not impose SC, LICQ, nor NFF. As we discuss below, estimating $V(\mu_0)$ without SC is challenging, since $V(\cdot)$ may be discontinuous at $\mu_0$ (\citet[Theorem 4]{robinson1975stability}). If any of SC, NFF or a version of LICQ fail, the `naïve' approach to inference that combines the plug-in estimator $V(\hat{\mu}_n)$ with bootstrap is invalid, see Remark \ref{remark_bootstrapfail}.

\subsection{Available estimators}\label{sec: available_estim}
We now briefly discuss the performance of the existing estimators under Assumption 1 alone. The sample LP defining $V(\hat{\mu}_n)$ may often be infeasible when $\Theta_I$ is `thin' relative to the sampling error. Even if $\Theta_I(\hat{\mu}_n)$ is non-empty a.s. for every $n$, $V(\hat{\mu}_n)$ may not be consistent for $V(\mu_0)$ if SC fails, as the following simplified version of Example \ref{example_interval} illustrates. 
\begin{figure}\centering
\begin{subfigure}{0.5\textwidth}\centering
    \begin{tikzpicture}[scale = 1.5]
    \tikzstyle{every node}=[font=\footnotesize]
    \draw[->] (-2,0) -- (2,0) node[right] {$\theta_1$};
    \draw[->] (0,-1.5) -- (0,1.5) node[above] {$\theta_2$};
    \draw[dashed, gray] (-1,-1.5) -- (-1,1.5);
    \draw[dashed, gray] (1,-1.5) -- (1,1.5);
    \draw[thick] (-1.5,-1.5) -- (1.5,1.5) node[above right] {$\theta_2 \leq \theta_1$};
    \draw[->, line width= .1mm, black] (1.5,1.5) --(1.5,1.4);
    \draw[->, line width= .1mm, black] (-1.5,-1.5) --(-1.5,-1.6);
    \draw[thick, red] (-1.5,-1.2) -- (1.5,1.2) node[above right] {$\theta_2 \geq (1+b)\theta_1$};
    \draw[->, line width=0.1mm, red] (-1.5,-1.2) -- (-1.5,-1.1);
    \draw[->, line width=0.1mm, red] (1.5,1.2) -- (1.5,1.3);
    \fill[gray!20] (0,0) -- (1,1) -- (1,.8) -- (0,0) -- cycle;
    \filldraw[black] (0,0) circle (1pt);
    \draw (-1,0.05) -- (-1,-0.05) node[below] {$-1$};
    \draw (1,0.05) -- (1,-0.05) node[below] {$1$};
    \draw (0,-0.05) node[below left] {$0$};
    \draw[->, line width= .1mm] (1,1.4) -- (0.9,1.4) node[above] {$\theta_1 \leq 1$};
    \draw[->, line width = .1mm] (-1,1.4) -- (-0.9, 1.4) node[above] {$\theta_1 \geq -1$}; 
\end{tikzpicture}
\caption{$b < 0, ~ V(b) = 0$}
\end{subfigure}%
\begin{subfigure}{0.5\textwidth}\centering
\begin{tikzpicture}[scale = 1.5]
\tikzstyle{every node}=[font=\footnotesize]
    \draw[->] (-2,0) -- (2,0) node[right] {$\theta_1$};
    \draw[->] (0,-1.5) -- (0,1.5) node[above] {$\theta_2$};
    \draw[dashed, gray] (-1,-1.5) -- (-1,1.5);
    \draw[dashed, gray] (1,-1.5) -- (1,1.5);
    \draw[thick] (-1.5,-1.5) -- (1.5,1.5) node[above right] {$\theta_2 = \theta_1$};
    \draw[thick, red] (-1.5,-1.5) -- (1.5,1.5);
    \draw[->, line width= .1mm, white] (-1.5,-1.5) --(-1.5,-1.6);
    \filldraw[black] (-1,-1) circle (1pt);
    \draw (-1,0.05) -- (-1,-0.05) node[below] {$-1$};
    \draw (1,0.05) -- (1,-0.05) node[below] {$1$};
    \draw (0,-0.05) node[below left] {$0$};
    \draw[->, line width= .1mm] (1,1.4) -- (0.9,1.4) node[above] {$\theta_1 \leq 1$};
    \draw[->, line width = .1mm] (-1,1.4) -- (-0.9, 1.4) node[above] {$\theta_1 \geq -1$}; 
\end{tikzpicture}
\caption{$b = 0, ~ V(b) = -1$}
\label{fig_point_cons_b0}
\end{subfigure}
\caption{The feasible region in \eqref{example_prop1} for two values of $b$.}
\label{fig_point_cons}
\end{figure}

\setcounter{example}{2}
\begin{example}\label{example_incons}
    Consider the following LP:
    \begin{align}\label{example_prop1}
        V(b) = \underset{\theta \in \mathbb{R}^2}{\min} ~\theta_1 \quad \text{s.t.}: \theta_2 \geq (1+b)\theta_1,~ \theta_2 \leq \theta_1, ~ \theta_1 \in [-1;1],
    \end{align}
    where $b$ is estimated via $\hat{b}_n = b + \frac{1}{n}\sum_{i = 1}^n U_i$ with $U_i \sim U[-1;1]$ i.i.d. Suppose in population $b = 0$, as in Figure \ref{fig_point_cons_b0}. The true value is then $V(0) = -1$, attained at $\theta^* = - \iota$. The plug-in estimator collapses to $V(\hat{b}_n) = - \mathds{1}\{ \hat{b}_n \geq 0\}$, and does not converge to $-1$ in probability.
\end{example}
Another class of estimators of $V(\mu_0)$ is obtained by projecting a set-estimator of $\Theta_I$ onto the direction $p$ (see e.g. \citet{santos}). For a measurable $\hat{\Theta}_I \subseteq \mathbb{R}^d$ that is compact, convex and non-empty w.p.a.1, if $\mathrm{d}_H$ is the Hausdorff distance, $$|\min_{\theta \in \hat{\Theta}_I}p'\theta - V(\mu_0)| \leq  \lVert p \rVert\text{d}_H(\hat{\Theta}_I,\Theta_I),$$
so a consistent estimator of $\Theta_I$ allows to consistently estimate $V(\mu_0)$. \citet{CHT} proposed to construct $\hat{\Theta}_I$ as the set of minimizers or near-minimizers of an appropriate criterion function. For concreteness, take $\hat{Q}(\theta) = \lVert (\hat{c}_n-\hat{M}_n\theta)^+\rVert^2$, used by e.g. \citet{holmes2011diffusion}. If exact minimizers of $\hat{Q}$ are used, one usually has $\text{d}_H(\hat{\Theta}_I,\Theta_I) = O_p(1/\sqrt{n})$ under an additional degeneracy condition\footnote{See condition C.3 in \citet{CHT}}. Unfortunately, the latter may fail if SC fails, and the `non-expanded' estimator may be inconsistent for the same reason as the plug-in estimator.
\setcounter{example}{2}
\begin{example}[cont'd]
    Suppose the parameter space is a box $\Theta = [-2;2]^2$. The sample LP is always feasible, so $\arg\min_{\theta \in \Theta} \hat{Q}(\theta) = \Theta_I(\hat{\mu}_n)$ if the box constraints are slack. This occurs w.p.a.1, so $\min_{\theta \in \hat{\Theta}_I} p'\theta$ coincides with the inconsistent plug-in estimator w.p.a.1. 
\end{example}
To circumvent this, \citet{CHT} propose to instead use $k_n/n$-near-minimizers of $\hat{Q}$, for a slowly diverging $k_n = o(n)$. This \textit{set-expansion} approach can be shown to be $\sqrt{n/k_n}$-consistent for $V(\mu_0)$ under Assumption 1. While its rate is `almost' $1/\sqrt{n}$, it may be rather conservative in finite samples, see Section \ref{mc_section}. 

\section{Estimation and Inference}\label{section_consistency}
We construct an estimator that is consistent for $V(\mu_0)$ under Assumption 1 by replacing the constraints in \eqref{lp_init} with a penalty on their violations. For $w \in \mathbb{R}^q_+$, define
\begin{align*}
     L(\theta;\mu, w) \equiv  \tilde p'\theta + w'(\tilde c - \tilde M\theta)^+, \quad V(\mu; w) \equiv     
    \underset{\theta \in \Theta}{\min}~ L(\theta; \mu, w),
\end{align*}
where the latter value is finite, and we let $\mathcal{S}(\mu; w)$ denote the corresponding solution set.

The penalized problem enforces the constraints $\theta \in \Theta$ by design, which is natural when these are warranted by the underlying economic structure. For instance, if $\theta$ collects probability masses, we may take $\Theta = \{\theta \in [0,1]^d: \sum^d_{j=1}\theta_j = 1\}$. When no such set is known, one may instead minimize $L$ over $\mathbb R^d$ (see Supplemental Appendix
\ref{unknown_Theta}). 

The penalized and LP objectives agree on feasible points, so the penalized problem yields a weakly smaller value in general. When the constraints $M\theta \geq c$ are feasible and $w$ is large enough, the penalty term $w'(c-M\theta)^+$ enforces them, and the two values coincide.

\begin{ass}[Penalty parameter]
     There exists a $\lambda^* \in \Lambda(\mu_0)$, such that $w > \lambda^*$.
 \end{ass}
\begin{lemma}\label{lemma_penalty} 
    For any $\mu \in \mathbb{R}^m$, and $w \in \mathbb{R}^q_{+}$, if $\Theta_I(\mu) \subseteq \Theta$, then
\begin{align*}
    V(\mu; w) \leq V(\mu) \in \mathbb{R} \cup \{+\infty\}. 
\end{align*}
Moreover, under Assumptions 1 and 2: i) $V(\mu_0; w) = V(\mu_0)$, and ii) $\mathcal{S}(\mu_0; w) = \mathcal{S}(\mu_0)$. 
\end{lemma}
The set of dual solutions may be unbounded under Assumption 1, but it is non-empty. Therefore, a large enough $w$ always satisfies Assumption 2 under Assumption 1. While it is unclear what penalty is large enough absent further conditions, an arbitrarily large value may reproduce the failure of unregularized estimators discussed above. This is because a larger $w$ penalizes violations caused by sampling error more heavily. 

\setcounter{example}{2}
\begin{example}[cont'd]\label{lm_example}
Consider \eqref{example_prop1} with $b = 0$. The population dual solution set $\Lambda(0)=\{(l, l, 1, 0):l\ge0\}$ is unbounded. Still, any fixed $w > (0,0,1,0)'\in \Lambda(0)$ satisfies Assumption 2. However, if the penalty is so large that $w > \hat{\lambda}_n=|\hat{b}^{-1}_n|(1,1,0,0)'$, and $\hat{b}_n < 0$ in the sample, the estimator $V(\hat{b}_n; w)$ coincides with the plug-in estimator, and selects $\mathcal{S}(\hat{b}_n;w) =\{(0,0)\}$. This follows by Lemma \ref{lemma_penalty}, since $\hat{\lambda}_n$ is the unique sample dual solution if $\hat{b}_n < 0$. Because $\hat{b}_n \xrightarrow{p} 0$, eventually $w > \hat{\lambda}_n$ fails with high probability, and every penalty minimizer has $\theta_1 =-1$. While this conclusion holds asymptotically for any fixed $w$, it need not be a good description of the estimator's finite-sample behavior if $w$ is `too large'.
\end{example}

To address this observation, we study $w \to \infty$ asymptotics. For brevity, from now on we treat $w \geq 0$ as a scalar inducing the penalty vector $w \iota$.\footnote{Results below extend immediately to the case when the coordinates of the penalty vector differ.} Without loss of generality, $p \ne 0$ is treated as known in what follows\footnote{\label{fn:estimated_p}For an auxiliary variable $t \in \mathbb{R}$, add the constraints $t - \hat{p}_n'\theta \geq 0$ and $-t \geq -C$ for some large fixed $C >0$, insert $t$ into the first position of $\theta$ and redefine $p \equiv  (1 ~ 0 \dots 0)'$ (see \citet{gafarov2024simple} for more details).}, and we write $\mu=\text{vec}(\tilde M,\tilde c) \in \mathbb{R}^m$, with $m=q(d+1)$.

\begin{prop}\label{first_penalty_theorem}
    Under Assumption 1, for any $w_n \to \infty$ w.p.a.1 with $w_n = o_p(\sqrt{n})$, 
    \begin{align*}
        V(\hat{\mu}_n; w_n)  - V(\mu_0) = O_p\left(\frac{w_n}{\sqrt{n}}\right).
    \end{align*}
\end{prop}
Letting $w_n$ diverge avoids the need to know the threshold for
population sharpness $V(\mu_0) = V(\mu_0;w_n)$, while $w_n = o_p(\sqrt n)$ preserves the regularizing role of our relaxation.

\subsection{Debiased penalty estimator}\label{section_debiased}
The $w_n/\sqrt{n}$ rate in Proposition \ref{first_penalty_theorem} stems from the slowly vanishing penalty term. A natural question is whether the $1/\sqrt{n}$ rate could be attained by removing this term, i.e. \textit{debiasing} the penalty value function. We define \textit{the debiased penalty function estimator} as
\begin{align*}
    \hat{V}(\hat{\mu}_n;w_n) \equiv \max_{\theta \in \mathcal{S}(\hat{\mu}_n; w_n)} p'\theta.
\end{align*}
Showing that it is $\sqrt{n}-$consistent is one of our main contributions.

\begin{theorem}\label{root_n_pointwise}
    Suppose Assumption 1 holds and either i) $\mathcal{S}(\mu_0) \subseteq \text{Int}(\Theta)$, or ii) $\Theta$ is a polytope. For any $w_n \to \infty$ w.p.a.1 with $w_n = o_p(\sqrt{n})$,
    \begin{align*}
        \underset{\theta \in \mathcal{S}(\hat{\mu}_n; w_n)}{\max} ~ \left|p'\theta - V(\mu_0)\right| = O_p\left(\frac{1}{\sqrt{n}}\right).
    \end{align*}
\end{theorem}
\begin{remark}\label{remark_xalp}
    Theorem \ref{root_n_pointwise} implies that any measurable selection from $p'\mathcal{S}(\hat{\mu}_n; w_n)$, including $\hat{V}(\hat{\mu}_n;w_n)$, yields a $\sqrt{n}-$consistent estimator of $V(\mu_0)$. We choose the maximum because it gives the least conservative estimate of a sharp lower bound among the penalized minimizers. Computing $\hat{V}(\cdot)$ amounts to solving two LPs, see \eqref{eq: lp_reform} in Appendix \ref{app_imp}. 
\end{remark}
\begin{remark}
The discussion in Section \ref{section_uniformity} sheds light on the role of $w_n$ and the trade-off involved in its selection. The practical guidance on selecting $w_n$ is then developed on the basis of our results and random matrix theory in Supplemental Appendix \ref{ap_pensel}. 
\end{remark}

At a high level, Theorem \ref{root_n_pointwise} can be viewed as a post-regularization result. The penalized estimator regularizes the potentially discontinuous LP and selects local identifying structure. The debiased estimator retains that structure but removes the regularization bias, restoring $\sqrt{n}-$consistency. This parallels the post-Lasso procedure (see \citet{bc2013}), with the active face playing the role of the selected support, as discussed below. By contrast, the set expansion estimator regularizes the problem by enlarging the estimated feasible set. The expansion directly enters the reported endpoint and can induce substantial conservative bias in finite samples (see Section \ref{mc_section}).

We now provide a technical discussion of Theorem \ref{root_n_pointwise}. By regularizing the problem, the penalty criterion first ensures that $\mathcal{S}(\hat{\mu}_n;w_n)$ is localized around $\mathcal{S}(\mu_0)$. Since $w_n\to\infty$, any sample minimizer must be nearly feasible in the population. As $w_n=o_p(\sqrt n)$, by Proposition \ref{first_penalty_theorem}, it must also be nearly-optimal. Hence, it must be close to the population solution set. This argument is in the spirit of M-estimation (see e.g. Chapter 5 of \citet{Vaart_1998}) and set-estimation literature (e.g. \citet{CHT}). However, it only gives consistency, but not yet a $1/\sqrt n$ rate. 

To refine this localization, we use the linear structure of an LP solution set. For simplicity of discussion, suppose that the penalized problem has a unique solution $\hat{\theta}$ in the sample. Because this problem is equivalent to an LP w.p.a.1 (see \eqref{eq: lp_reform} in Appendix \ref{app_imp}), $\hat{\theta}$ must be a vertex, meaning that inequalities $\hat{A}$ that bind at it have a full-column-rank submatrix $\hat{M}_{\hat{A}}$. Note that we omit the $n$ subscript when indexing by a set. This $\hat{A}$ has an important property of being a `nice face' w.p.a.1. 
\begin{definition}[Nice face]
    A nonempty $A \subseteq [q]$ corresponds to 
\textit{a nice face}\footnote{A nice face is not necessarily a valid $k-$face of the true polytope $\Theta_I$.} $F \equiv \{\theta \in \mathbb{R}^d:M_A \theta = c_A\}$ if $F$ is nonempty, and $p'\theta = V(\mu_0)$ for any $\theta \in F$. 
\end{definition}
The argument for why $\hat A$ is nice w.p.a.1 uses the above localization and the finiteness of the collection of possible binding sets (see Appendix \ref{ap_proof_th22}). This property is important, because a nice face identifies the value. Once it is found, it suffices to estimate $M_{\hat{A}}$ and $c_{\hat{A}}$ to estimate $V(\mu_0)$, which can be done at $1/\sqrt{n}$ rate. Formally, on the event that $\hat{A}$ is nice, 
\begin{align}\label{optimality}
    \exists \theta^* \in \mathbb{R}^d: M_{\hat A} \theta^* = c_{\hat A}, ~ \quad p'\theta^* = V(\mu_0).
\end{align}

Furthermore, because $\hat{A}$ is nice, $p'\theta$ is constant on $\{\theta: M_{\hat{A}}\theta = c_{\hat{A}}\}$. The vector $p$ is thus orthogonal to the null space of $M_{\hat A}$. By the fundamental theorem of linear algebra, we have
\begin{align}\label{range_req}
    \exists r \in \mathbb{R}^{|\hat{A}|}: p =  M'_{\hat A}r.
\end{align}
Therefore, $V(\mu_0) = p'\theta^* = r'c_{\hat{A}}$. Because inequalities $\hat{A}$ bind at $\hat{\theta}$ in the sample, 
\begin{align}\label{decomp_value}
    p'\hat{\theta} - V(\mu_0) = r'M_{\hat{A}}\hat{\theta} - r'c_{\hat A} = r'((\hat{c}_{\hat{A}} - c_{\hat{A}}) +(M_{\hat{A}} - \hat{M}_{\hat{A}})\hat{\theta}) = O_p(1/\sqrt{n}).
\end{align}
The estimation error in \eqref{decomp_value} stems only from sampling error in the selected inequalities, without the diverging factor \(w_n\). This explains why removing the penalty term from the reported value restores the \(1/\sqrt n\) rate. Interestingly, we do not need to re-estimate $\hat{\theta}$ from the selected system in the second step, because $\hat{M}_{\hat{A}}\hat{\theta} = \hat{c}_{\hat{A}}$ holds by construction. The above argument extends to a non-singleton $\mathcal{S}(\hat{\mu}_n; w_n)$, because a non-empty and bounded argmin of an LP is a convex hull of finitely many vertex-solutions.

\subsection{Inference}\label{section_inference}
We now show how our estimation results may be used to perform inference on $V(\mu_0)$. The debiased penalty estimator supplies the first-stage object $(\hat{A}, \hat{\theta})$ needed for inference. Conditional on $\hat{A}$ being a nice face,  we have the linear representation \eqref{decomp_value} around the population value. If the selected face were fixed, asymptotic normality of \(\hat\mu_n\) would yield inference directly. When NFF or LICQ fail, however, the selected $\hat A$ may fluctuate with \(\sqrt n(\hat\mu_n-\mu_0)\), so the asymptotic distribution of our estimator $\hat{V}$ need not be Gaussian.

We address this by combining sample splitting with Neyman-orthogonal estimation in the spirit of \citet{ddml}. Intuitively, in the context of subvector inference, our target is the bound $V(\mu_0)$, while a $\theta^*$ that attains it may be viewed as a nuisance parameter. Unlike in double/debiased ML settings, where one starts from a moment condition identifying the target and nuisance parameters \citep{ddml}, our primitive restrictions are inequalities. The ability to select a nice face supplies the missing structure: the binding equalities \eqref{optimality} and the condition \eqref{range_req} pin down \(p'\theta^*\)\footnote{They may not pin down the nuisance parameter $\theta^*$: there are DGPs with a unique sample minimizer, where $\text{rk}(M_{\hat{A}})<d$ with non-vanishing probability.}. Having access to $\hat{A}$ thus reduces inference to a debiasing problem: it remains to construct an estimating function $\psi$ that is first-order insensitive to errors in the estimation of nuisance parameters.

\begin{definition}[Optimal triplet]
    $(A, \theta, r) \in 2^{[q]}\times \mathcal{S}(\mu_0) \times \mathbb{R}^{|A|}$ is an optimal triplet if i) $|A| \geq d$; ii) $M_A\theta = c_A$, and iii) $p = M_A'r$.
\end{definition}

For $\mu = \text{vec}(\tilde M,\tilde c) \in \mathbb{R}^m$, and $T=(A, \theta, r) \in 2^{[q]}\times \mathbb{R}^d\times \mathbb{R}^{|A|}$, consider the map
\begin{align*}
    \psi(\mu,T) \equiv p'\theta + r'(\tilde{c}_A - \tilde{M}_A \theta).
\end{align*}
If $T$ is an optimal triplet, then $\psi(\mu_0, T) = V(\mu_0)$. Moreover, $\nabla_{\theta ,r}\psi(\mu_0, T) = 0$, so $\psi$ is orthogonal. Our procedure uses $\psi$ and is valid under the following standard assumptions.

\begin{ass}[Random sample]
    $\hat{\mu}_n = \hat{\mu}(\mathcal{W}_n)$ is a measurable function of the sample $\mathcal{W}_n \equiv \{W_1, W_2, \dots, W_n\}$, where $W_i \in \mathbb{R}^{d_W}, ~ i \in [n]$ are i.i.d. random vectors.
\end{ass}
\begin{ass}[Asymptotic normality]
     The estimator $\hat{\mu}_n$ is such that, for $\Sigma < \infty$,
     \begin{align*}
         \sqrt{n}(\hat{\mu}_n - \mu_0) \xrightarrow{L} \mathbb{Z} =\text{vec}(\mathbb Z^M,\mathbb Z^c) \sim \mathcal{N}(0, \Sigma).
     \end{align*}
 \end{ass}
\begin{ass}[Variance estimator]
   A measurable and a.s. positive-semidefinite $\hat{\Sigma}_n = \hat{\Sigma}(\mathcal{W}_n)$ satisfies $\hat{\Sigma}_n\xrightarrow{p} \Sigma$.
\end{ass}
 In applications, asymptotic normality often follows from the CLT and the Delta Method when $\hat{\mu}_n$ is a smooth function of the sample moments (as in Section \ref{mc_section}), while $\hat{\Sigma}_n$ can be obtained from the sample covariance matrix of $W_i$ with the Delta Method, or via bootstrap.

Let us randomly split $\mathcal{W}_n$ into two disjoint and exhaustive folds $\mathcal{W}^{(f)}$ of size $n_f$, such that $\lim_{n\to \infty} n_1/n \in (0,1)$, and write $\hat\mu^{(f)}=\text{vec}(\hat M^{(f)},\hat c^{(f)})$ and $\hat\Sigma^{(f)}$ for the coefficient and covariance estimates from fold $f$. We estimate an optimal triplet via $\hat{T}\equiv (\hat{A},\hat{\theta},\hat{r})$ on the first fold and form the statistic $\breve{V} \equiv \psi(\hat{\mu}^{(2)}, \hat{T})$. The vertex is estimated\footnote{In our simulations, $\hat{\theta}$ solving \eqref{eq: lp_reform} usually satisfies \eqref{vertex_estim}. Should it not, i) draw $\tau \sim U[-1,1]^{d}$; ii) find $\hat{\theta}, \hat{a} \in \arg \min_{\theta \in \overline{\Theta}, a \geq 0} \tau'\theta$ s.t. $p'\theta + w_{n_1} \iota'a \leq V(\hat{\mu}^{(1)}; w_{n_1}), a \geq \hat{c}^{(1)} - \hat{M}^{(1)}\theta,$ and $p'\theta \geq \hat{V}(\hat{\mu}^{(1)}; w_{n_1})$, where $\overline{\Theta}$ is as in \eqref{eq: lp_reform}; iii) check $\text{rk}(\hat{M}^{(1)}_{J(\hat{\theta}; \hat{\mu}^{(1)})}) = d$ -- if not, go to i). This finds $\hat{\theta}$ with exponentially vanishing failure probability.} as
\begin{align}\label{vertex_estim}
    \hat{\theta} \in \arg \max_{\theta \in \mathcal{S}(\hat{\mu}^{(1)}; w_{n_1})} p'\theta ~ \quad \text{s.t.:} \quad \text{rk}(\hat{M}^{(1)}_{J(\theta; \hat{\mu}^{(1)})}) = d, \quad \hat{A} \equiv J(\hat{\theta};\hat{\mu}^{(1)}),
\end{align}
where $J(\theta; \mu) \subseteq [q]$ indexes the inequalities $\tilde{M}\theta \geq \tilde{c}$ that bind at $\theta$. It remains to estimate a representer, i.e. a solution $r$ to the stationarity equation \eqref{range_req}. We now illustrate how the choice of $\hat{r}$ affects the asymptotic distribution of $\breve{V}$. In Supplemental Appendix \ref{proof_main_inf}, we show that there is a random $\theta^* \in \mathcal S(\mu_0)$ with $M_{\hat A}\theta^*=c_{\hat A}$ w.p.a.1 and $\|\hat\theta-\theta^*\|=O_p(n^{-1/2})$, and that, for any first-fold-measurable $\hat r = O_p(1)$,
\begin{align}\label{motivation}
    \sqrt{n_2}(\breve V-V(\mu_0))
    &=\sqrt{n_2}\hat r'\left((\hat c^{(2)}_{\hat A}-c_{\hat A})-(\hat M^{(2)}_{\hat A}-M_{\hat A})\hat\theta\right) \\ \notag
    &+\sqrt{n_2}(p-\hat M_{\hat A}^{(1)\prime}\hat r)'(\hat\theta-\theta^*)+O_p(n^{-1/2}).
\end{align}
The first term is asymptotically Gaussian conditional on the first fold, with variance $\sigma^2(\hat T,\Sigma)$, where $\sigma^2(T,\Sigma)\equiv r'\Omega(A,\theta,\Sigma)r$, and $\Omega(A,\theta,\Sigma)\equiv\Var(\mathbb Z_A^c-\mathbb Z_A^M\theta)$. The second term may be interpreted as the leading bias term. Exactly solving the sample stationarity equation nullifies the latter, but may `blow up' $\sigma^2(\hat T,\Sigma)$ when $M_{\hat A}$ is singular. One should therefore use a regularized solution $\hat r$ that makes the stationarity residual vanish while keeping the variance bounded. Theorem \ref{th: main_inference} below is formulated under the corresponding high-level condition that is immediate to verify for many standard choices of $\hat r$ (e.g. Tikhonov or TSVD regularization). Our concrete recommendation is to choose $\hat r$ by minimizing an upper bound on the leading MSE in \eqref{motivation}, see \eqref{reg_v} in Appendix \ref{app_imp}. Our final assumption is
\begin{ass}[Non-degeneracy] $\sigma(T, \Sigma) > 0$ for any optimal triplet $T$.
\end{ass}

Assumption 6 rules out cases in which the selected nice face leaves no first-order variation in \eqref{motivation}. This may occur, for example, if the selected constraints are deterministic, i.e. $M_{\hat{A}}$ and $c_{\hat{A}}$ are known. On the event that Assumption 6 fails for the selected optimal triplet, $\breve{V}$ is typically superconsistent, and the lower confidence bound in Theorem \ref{th: main_inference} is below $V(\mu_0) + O_p(n^{-1})$. However, in that case the procedure may undercover. Remark \ref{remark_drop_varpos} allows to drop Assumption 6 at the cost of potential conservativeness. 
\begin{theorem}\label{th: main_inference}
     Suppose that Assumptions 1, 3--6 hold and either i) $\mathcal S(\mu_0)\subseteq\mathrm{Int}(\Theta)$, or ii) $\Theta$ is a polytope. Let $w_n\to\infty$ w.p.a.1 with $w_n=o_p(\sqrt n)$, and let $\hat r\in\mathbb R^{|\hat A|}$ be $\mathcal{W}^{(1)}-$measurable and satisfy $\hat r=O_p(1)$,  $\|p-\hat M_{\hat A}^{(1)\prime}\hat r\|=o_p(1)$. Then, for any $\alpha\in(0;1)$,
    \begin{align*}
    \mathbb{P}\left[\frac{\sqrt{n_2}}{\sigma(\hat T, \hat{\Sigma}_n)} \left(\breve{V} - V(\mu_0)\right) \leq z_{1-\alpha}\right] = 1-\alpha + o(1). 
    \end{align*}
\end{theorem}
\begin{remark}\label{remark_drop_varpos}
     In applications with deterministic constraints in $M \theta \geq c$ that may determine the bound, or whenever degeneracy is a concern, Assumption 6 may be dropped if the standard error is replaced with $\max\{{\sigma}(\hat{T},\hat{\Sigma}_n),\underline{\sigma}\}$ for some fixed $\underline{\sigma} > 0$. Such a CI controls level for $\alpha <0.5$, but its size may be conservative. 
\end{remark}

\begin{remark}\label{remark_cht_ci}
Our procedure allows to construct a confidence region for the identified interval \(p'\Theta_I=[V(\mu_0),V^u(\mu_0)]\) of the type considered in \citet{CHT} by using a joint Gaussian approximation of its endpoints (\citet[Section 3.3.2]{gafarov2024simple} gives a similar construction). Run the procedure above with \(-p\) on the same folds and denote the resulting triplet and estimator by \(\hat T^u=(\hat A^u,\hat\theta^u,\hat r^u)\) and \(\breve V^u\), so that \(-\breve V^u\) estimates \(V^u(\mu_0)\). For triplets $T=(A,\theta,r)$ and $T^u$,  write \(Z(T)\equiv r'(\mathbb Z_A^c-\mathbb Z_A^M\theta)\) and $\mathbf{\Sigma}(T, T^u, \Sigma) \equiv \Var(\text{vec}(Z(T),Z(T^u)))$. If the conditions of Theorem \ref{th: main_inference} hold for both $p$ and $-p$, then
\begin{align*}
    \P\left[p'\Theta_I \subseteq \left[\breve V - \frac{\hat c_{1-\alpha}}{\sqrt{n_2}}, -\breve V^u + \frac{\hat c_{1-\alpha}}{\sqrt{n_2}}\right]\right] = 1-\alpha + o(1),
\end{align*}
where $\hat{c}_{1-\alpha}$ is the $1-\alpha$ conditional quantile of $\max\{\mathbb{V}_1, \mathbb{V}_2\}$ for 
$\mathbb{V}|\mathcal{W}_n \sim N(0,\mathbf{\Sigma}(\hat{T}, \hat{T}^u, \hat{\Sigma}_n)).$ Assumption 6 can be dropped in line with Remark \ref{remark_drop_varpos}. This construction straightforwardly extends to multiple supporting directions. It also suggests a test of the null that $p'\theta$ is point-identified, i.e. $H_0: V(\mu_0) = V^u(\mu_0)$ vs $H_a: V^u(\mu_0) > V(\mu_0)$: reject if and only if
\begin{align*}
    \frac{\sqrt{n_2}(-\breve V^u-\breve V)}{\max\{\underline{\sigma},\sqrt{\iota' \mathbf{\Sigma}(\hat{T},\hat{T}^u, \hat{\Sigma}_n)\iota }\}} > z_{1-\alpha}.
\end{align*}
Under the same conditions, except Assumption 6, the test has asymptotic level \(\alpha\), and is consistent against fixed alternatives. It may be useful for constructing tighter \citet{manskiimbens2004} two-sided confidence intervals for the partially identified parameter.
\end{remark}
\begin{remark}\label{remark_bootstrapfail}
    In the Supplemental Appendix, we analyze the Hadamard directional derivative of $V(\cdot)$ and show that the common practice of combining bootstrap with the plug-in estimator (\citet{blundell2007changes}, \citet{kreider2012identifying}, \citet{pepper}, \citet{siddique}, \citet{de2017effect}, \citet{mentalhealth}) consistently estimates the asymptotic distribution under Assumptions 1, 3, 4 if and only if i) SC, ii) NFF, and iii) a version of LICQ hold at $\mu_0$. We also construct alternative inference procedures based on the plug-in, set-expansion, and $V(\cdot; w)$ using \citet{hong2015numerical} and \citet{santos2019}. 
\end{remark}
\label{main_inference}
\section{Uniform Asymptotic Theory}\label{section_uniformity}
Example \ref{example_incons} suggests that, under Assumption 1 alone, small changes in the underlying distribution may produce large jumps in the LP value. In the spirit of the findings of \citet{hirano}, this precludes the existence of a uniformly consistent LP estimator over the \textit{unrestricted set of measures}\footnote{We show this formally in Supplemental Appendix \ref{app:uniform_impossibility}.} that collects all DGPs consistent with Assumption 1.
\begin{ass}[Unrestricted set of measures] The vector-valued functional $\P \to \mu_0(\P) = \text{vec}(M(\P),c(\P))$ and the set of probability measures $\mathcal{P}$ are such that $\mu_0(\mathcal{P}) = \{ y \in \mathbb{R}^m \text{ s.t. } \Theta_I(y) \ne \emptyset, \Theta_I(y) \subseteq \Theta\}$ for a known and fixed compact $\Theta$.
\end{ass}
The following example gives a familiar special case of this obstruction.
\begin{example}[IV as an LP]\label{lin_iv}
    Consider a linear IV model with endogenous covariates $X \in \mathbb{R}^{d-1}$, outcome $Y \in \mathbb{R}$ and instruments $Z \in \mathbb{R}^{d-1}$. The moment conditions identifying the coefficient vector $\theta^* \in \mathbb{R}^d$ are
\begin{align}\label{linear_IV}
    \E[(1 ~ Z')'(Y-(1 ~ X')\theta^*)] = 0.
\end{align}
Write \eqref{linear_IV} as $R  \theta^* = b$. The sharp lower or upper bound on \(\theta^*_j\) under  \eqref{linear_IV} is \(V(\mu_0)\) or \(-V(\mu_0)\), with \(p=e_j\) or \(p=-e_j\), respectively, and $M  = (1 ~ -1)' \otimes R$, ~ $c = ( 1 ~ -1)' \otimes  b.$ 

Estimating all such bounds is equivalent to estimating \(\theta^*\) when \(\Theta_I=\{\theta^*\}\). If the moment matrix is nonsingular, i.e. IV is relevant, \(\theta^* = R^{-1}b\) is point-identified and can be estimated at \(1/\sqrt n\) rate by plug-in (2SLS), or by the debiased penalty estimator. Uniformly over the unrestricted class, however, this is impossible: along sequences for which the smallest singular value of $\text{Cov}(X,Z)$ approaches zero, arbitrarily small perturbations of the moments can move \(\theta^*\), and hence the LP bounds, by an order-one amount. This is known as the weak-IV problem (see e.g. \citet{staiger1997instrumental}).
\end{example}

\subsection{Uniform consistency under the $\delta$-condition}\label{sec_unif_cons}
It is natural to seek a mild restriction on $\mathcal{P}$ under which a uniformly consistent estimator exists. Examples \ref{example_interval} and \ref{example_licq_failure} suggest that such a restriction should not rule out point-identification, over-identification, or solution multiplicity, as these cases arise naturally in economic applications. Example \ref{lin_iv} instructs that it should at least rule out sequences along which the relevant faces of $\Theta_I$ become arbitrarily ill-conditioned. The \(\delta\)-condition below is constructed in line with these observations and guarantees uniform consistency of $\hat{V}$.

For a non-empty and bounded $\tilde{\Theta} = \{\theta \in \mathbb{R}^d: \tilde{M}\theta \geq \tilde{c}\}$, define the following quantities.
\begin{definition}[Face condition number]

    For a $k$-face\footnote{See Chapters 2, 3 in \citet{grunbaum1967convex} for the definition. Note that it may be that $\text{dim}\tilde{\Theta} < d$.} $f$ of $\tilde{\Theta}$ with $0\leq k<d$, described by binding constraints $A \subseteq [q]$ with $|A| \geq d - k$ and $\text{rk}(\tilde{M}_A) = d - k$, the face condition number is
    \begin{align*}
        \tilde{\kappa}(f) \equiv \underset{B \subseteq A: ~ \text{rk} (\tilde{M}_B) = |B| = d-k}{\min} \mathrm{sv}_{d - k}(\tilde{M}_B).
    \end{align*}
\end{definition}
\begin{definition}[Polytope condition number]\label{def_polcond}
    The polytope condition number of $\tilde{\Theta}$ is
    \begin{align*}
        \kappa(\tilde{\Theta}) \equiv \underset{f \text{ is a $k-$face of $\tilde{\Theta}$, $0 \leq k < d$}}{\min} \tilde{\kappa}(f) = \underset{f \text{ is a vertex of $\tilde{\Theta}$}}{\min} \tilde{\kappa}(f) > 0.
    \end{align*}
\end{definition}
The equality in Definition \ref{def_polcond} is proven in the Supplemental Appendix. We now state the $\delta-$condition.
\begin{definition}[$\delta$-condition]
    For any $\delta > 0$, the class of measures satisfying the $\delta-$condition is $\mathcal{P}^\delta \equiv \{\P \in \mathcal{P}: \kappa(\Theta_I(\P)) > \delta\}$.
\end{definition}

Technically, the $\delta$-condition bounds away from zero the smallest singular value of every full-rank size-$d$ subset of inequalities binding at a vertex of $\Theta_I$. Under LICQ, each vertex is defined by one such square system. Even when more than $d$ inequalities bind at a vertex, the $\delta-$condition holds for some $\delta>0$: there are finitely many full-rank size-$d$ subsets, each with a positive smallest singular value. Thus, it allows LICQ to fail, see Example \ref{example_licq_failure} (cont'd).

Intuitively, the $\delta$-condition is a low-level restriction that allows to rule out sequences of measures along which the bound $V(\mu_0)$ is weakly identified, uniformly over directions $p$. Along such sequences the active restrictions deliver an endpoint for $p'\theta$ through an increasingly ill-conditioned system of equalities, so that sampling-error-sized perturbations of the constraints can open a nearby face on which $p'\theta$ moves by first order. 

\setcounter{example}{0}
\begin{example}[cont'd]
For concreteness, suppose $\E[D\mid Z=z]$ and $\E[Y_k\mid Z=z]$, $k=l,u$, are increasing in $z$, and the box constraints are slack over $\Theta_I$. Write $\Delta_D=\E[D\mid Z=1]-\E[D\mid Z=0]>0$. The lower bound on $\theta_2$ is attained at the intersection of the upper restriction for $Z=0$ and the lower restriction for $Z=1$, and equals
$(\E[Y_l\mid Z=1]-\E[Y_u\mid Z=0])/\Delta_D$. Moreover, every vertex is generated by one restriction at each value of $Z$, so
\[
\kappa(\Theta_I)
=\mathrm{sv}_2\!
\begin{pmatrix}
1&\E[D\mid Z=0]\\
1&\E[D\mid Z=1]
\end{pmatrix}
\asymp \Delta_D
=\frac{\Cov(D,Z)}{\Var(Z)}.
\]
Thus, when $\Var(Z)$ is bounded away from zero, the $\delta$-condition is equivalent, up to constants, to ruling out weak instruments. As $\Delta_D\downarrow0$, the two active restrictions become nearly parallel, and an $O(n^{-1/2})$ perturbation of either active outcome moment moves the bound by $O((\sqrt n\Delta_D)^{-1})$, which is order one along weak-IV sequences with $\Delta_D\asymp n^{-1/2}$.
\end{example}
\setcounter{example}{1}
\begin{example}[cont'd]
Suppose $\E[X]=0$. LICQ fails, because both $\theta\geq0$ and $\theta \geq \E[X]$ bind at the optimum $\theta^*=0$. Nevertheless, the $\delta$-condition holds for any $\delta<1$: every full-rank submatrix of binding constraints at any vertex is a scalar $1$ or $-1$, giving $\kappa(\Theta_I)=1$.
\end{example}
\setcounter{example}{2}
\begin{example}[cont'd]
To illustrate how the $\delta$-condition rules out weak identification of the bound in \eqref{example_prop1}, take the sequence of true DGPs with $b_n=-n^{-2}$. The distributions of the observed sample $\{b+ U_i\}^n_{i=1}$ under $b=b_n$ and $b=0$ are statistically indistinguishable, while $V(b_n)=0$ and $V(0)=-1$. Thus, $V(\cdot)$ is weakly identified along $b_n$. Such sequences are ruled out by the $\delta$-condition for a fixed $\delta$, as $\kappa(\Theta_I(b)) \asymp |b|$ when $|b|>0$ is small. However, at $b=0$, the  $\delta-$condition holds for any $\delta < (\sqrt5-1)/2 = \kappa(\Theta_I(0))$ despite the failure of SC and LICQ; taking $p=(-1 ~1)'$ also makes NFF fail without changing $\kappa(\Theta_I)$.
\end{example}

\setcounter{example}{3}
\begin{example}[cont'd]
    Under IV relevance, the identified polytope is a point. Every full-rank binding submatrix at its unique vertex equals the moment matrix $R$, up to row signs and permutations, so $\kappa(\Theta_I(\P)) = \mathrm{sv}_d \left(R(\P)\right)$. Thus, if $X, Z$ are mean-zero (pre-centered\footnote{Generally, $\kappa(\Theta_I(\P)) \asymp \min\{1,\mathrm{sv}_{d-1}(\operatorname{Cov}_{\P}(X,Z))\}$ up to constants over classes with bounded means.}), then, for $\delta < 1$, the $\delta-$condition is the familiar non-weak-IV condition:
    \begin{align*}
        \mathrm{sv}_{d-1}(\text{Cov}_\P(X,Z)) > \delta. 
    \end{align*}
    To discuss partial identification in this example, suppose that the parameter region is a box\footnote{    
    Unless $\Theta$ is formed by structural restrictions, it should be understood as a reporting region, and not as a source of identification. In that sense, if a sample solution is determined by $\Theta$, the researcher should understand this as a failure of identification. Fix a realization of noise with $\hat{R}^{-1}\hat{b} \in \Theta$, but $V(\hat\mu_n; w_n) \ne p'\hat{R}^{-1}\hat{b}$. Then, it is straightforward to show that in that sample $\hat{V}, V(\hat\mu_n; w_n) < p'\hat{R}^{-1}\hat{b}$, and, if one expands $\Theta$ via $-\underline{\theta}, \overline{\theta} \to + \infty$, the bounds diverge $\hat{V}, V(\hat\mu_n; w_n) \to -\infty$. This is inevitable under weak IV with no prior information on $\theta$.} $\Theta=[\underline\theta,\overline\theta]^d$ that we encode in $M\theta \geq c$. Then, for some $p$ the identified interval for $p'\theta$ has positive length if $\mathrm{sv}_d(R(\P)) = 0$ and $\Theta$ is wide enough. In that case, the $\delta$-condition does not restrict $\mathrm{sv}_d(R(\P))$, but instead requires the vertices formed by the intersection of linearly independent IV moments with some sides of the box $\Theta$ to be well-conditioned. 
\end{example}

We now relate the $\delta-$condition to conditions previously considered in the literature. Denote the subsets of measures from $\mathcal{P}$ that satisfy SC and NFF (at $p$) by $\mathcal{P}^{SC}$ and $\mathcal{P}^{NFF}$. Further define a uniform LICQ condition (see \citet{gafarov2024simple}): for some $\varepsilon \geq 0$ and $\P \in \mathcal{P}$,
\begin{align*}
\mathbb{P} \in \mathcal{P}^{LICQ; \varepsilon} \text{ iff } M_A(\P) \in \mathbb{R}^{d\times d}, \mathrm{sv}_d(M_A(\P)) > \varepsilon ~ \text{for any vertex }  A \text{ of } \Theta_I(\P),
\end{align*}
where \textit{vertex} refers to the set of indices of inequalities binding at a vertex of $\Theta_I(\cdot)$.
\begin{prop}\label{delta_properites}
Under Assumption 7, $d\ge2$, $q\ge d+2$,
and $\mathrm{Int}(\Theta)\ne\varnothing$, we have:
\begin{enumerate}
\item $\bigcup_{\delta>0}\mathcal{P}^{\delta}=\mathcal{P}$;
\item $\bigcup_{\delta>0}\mathcal{P}^{LICQ;\delta}\cup\mathcal{P}^{SC}\cup\mathcal{P}^{NFF}\subsetneq\mathcal{P}$;
\item $\mathcal{P}^{LICQ;\delta}\subsetneq\mathcal{P}^{\delta}$ for any $\delta>0$.
\end{enumerate}
\end{prop}
Proposition \ref{delta_properites} makes three points. First, any $\P$ from the unrestricted set $\mathcal{P}$ has a positive polytope condition number, and therefore satisfies the $\delta-$condition for a small enough $\delta > 0$. Second, at some measures in $\mathcal{P}$, SC, LICQ and NFF all fail. Such a measure need not satisfy the $\delta$-condition for a given fixed $\delta > 0$, but there always exists a $\delta > 0$ small enough that it does. Third, the uniform LICQ condition is strictly stronger than the $\delta$-condition. 

Note that, for a fixed $\delta>0$, the class $\mathcal P^\delta$ excludes some $\P$ satisfying SC, LICQ (with $\varepsilon=0$), or NFF, even jointly. These pointwise conditions, however, do not allow for uniform consistency: $\mathcal P^{LICQ;0} \cap \mathcal{P}^{SC} \cap \mathcal{P}^{NFF}$ does not admit a uniformly consistent estimator\footnote{This follows from a minimal modification of Example \ref{example_incons}, see Supplemental Appendix \ref{app:uniform_impossibility}.}.

The following Theorem is our main uniform estimation result\footnote{It is stated for a fixed $p$ for ease of exposition, but can be shown to hold uniformly over $\lVert p \rVert \lesssim 1$ by redefining the identified set in line with footnote \ref{fn:estimated_p} and examining the resulting $\kappa(\cdot)$.}. 
\begin{theorem}\label{uniform_penalty_theorem}
    Under Assumption 7, suppose that i) $\hat{\mu}_n(\P) -\mu_0(\P) = O_{\mathcal{P}}(1/\sqrt{n})$, and ii) $w_n(\P) > 0$ is such that 
    $\inf_{K > 0}\liminf_{n \to \infty} \inf_{\P \in \mathcal{P}} \P[w_n > K ] = 1.$ Then, 
    \begin{align}\label{uniform_validity}
        \lim_{\varepsilon\to \infty}\underset{n \to \infty}{\limsup} \underset{\mathbb{P} \in \mathcal{P}}{~\sup~} \mathbb{P}\left[\frac{\sqrt{n}}{w_n}(\hat{V}(\hat{\mu}_n;w_n) - V(\mu_0)) \geq \varepsilon\right] &= 0,\\
        \label{uniform_non_cons}
        \sup_{\delta > 0}  \lim_{\varepsilon\to \infty}\underset{n \to \infty}{\limsup} \underset{\mathbb{P} \in \mathcal{P}^\delta}{~\sup~} \mathbb{P}\left[\sqrt{n}(V(\mu_0) -\hat{V}(\hat{\mu}_n;w_n) ) \geq \varepsilon\right] &= 0.
    \end{align}
\end{theorem}

Whenever $V(\mu_0)$ corresponds to the sharp lower bound on a partially identified parameter of interest, \eqref{uniform_validity} can be understood as a uniform validity property: $\hat{V}$ overestimates the sharp lower bound by at most $O_\mathcal{P}(w_n/\sqrt{n})$ uniformly over the unrestricted set of measures. This is because the penalty problem is a relaxation of the LP \eqref{lp_init} with a weakly smaller population value, for which $V(\hat{\mu}_n;w_n) \geq \hat{V}(\hat{\mu}_n;w_n)$ is uniformly $\sqrt{n}/w_n$-consistent. 

\setcounter{example}{3}
\begin{example}[cont'd]
 Consider the consequence of the $\delta-$condition's failure. The penalty objective over $\Theta$ is $p'\theta + w_n\|b -R\theta\|_1$. If $R$ is nearly singular and $\Theta$ is nonbinding, the multiplier needed to enforce $R\theta=b$ is large for some unit direction $p$, as $\min_{\lambda \in \Lambda(\mu_0)}\lVert \lambda \rVert_\infty \geq \lVert R^{-\prime}p\rVert_\infty$. The penalized problem may then accept violations of the IV moments that permit movement in a weakly identified direction and lower $p'\theta$ (see Lemma \ref{lemma_penalty}). When this occurs, the sample penalty value is uniformly close to a smaller population value $V(\mu_0; w_n) < V(\mu_0)$. Because $\hat{V}(\hat{\mu}_n; w_n) \leq V(\hat{\mu}_n; w_n)$, we get one-sided validity in \eqref{uniform_validity}: the estimator remains uniformly valid for the bound over $\mathcal{P}$, but can be conservative.
\end{example}

Equation \eqref{uniform_non_cons} asserts that $\hat{V}$ is also not more than $O_{\mathcal{P}^{\delta}}(1/\sqrt{n})$ conservative for the bound uniformly over the set of measures that satisfy the $\delta-$condition. The role of the $\delta-$condition is to uniformly bound the norm of at least one dual solution of \eqref{lp_init} and hence ensure that the growing penalty eventually enforces the population restrictions, giving $V(\mu_0; w_n) = V(\mu_0)$ by Lemma \ref{lemma_penalty}. Appendix \ref{ap: pf_unif_pen_th} gives the sharper argument establishing the $1/\sqrt{n}$ rate.

Both \eqref{uniform_validity} and \eqref{uniform_non_cons} hold for the first-stage estimator $V(\hat{\mu}_n; w_n)$ in the place of its post-regularized version $\hat{V}$ at the same rates, because our uniform argument does not depend on the ability of the estimator to select a nice face. Intuitively, this reinforces the parallel between our results and the findings of \citet{bc2013}. Uniformly, the post-regularization procedure is only shown to perform no worse than the first-stage one. Whenever a nice face can be selected `well', which parallels consistent variable screening in Lasso, the post-regularized $\hat{V}$ may converge faster than $V(\hat{\mu}_n; w_n)$, as in Theorem \ref{root_n_pointwise}.

\subsection{Uniform correction to face selection for inference}\label{sec:uniform_inference}
The inference procedure in Theorem \ref{th: main_inference} conditions on the face $\hat{A}$ selected by the debiased penalty estimator. Uniformly, however, face selection may be unstable: along local sequences, a feasible basis with an $O(1/\sqrt{n})$ optimality gap in \eqref{lp_init} may be statistically indistinguishable from an optimal one. The procedure below makes a uniform correction, augmenting $\hat{A}$ with square bases that cannot be rejected as primal-dual optimal. For each candidate basis, we then use the linear expansion in \eqref{decomp_value} and take the lower envelope of the resulting confidence bands. 

In what follows, a basis is a nonempty $A\subseteq[q]$. It is square and nonsingular if $M_A$ is, in which case $A$ induces the primal and dual candidates $\theta(A)\equiv M_A^{-1}c_A$ and $\lambda(A)\equiv M_A^{-\prime}p$. Whenever well-defined, their sample counterparts are denoted by $\hat{\theta}(A)$ and $\hat{\lambda}(A)$.

\begin{ass}[Admissible class]\label{ass:admissible}
A non-empty collection $\mathcal P_0\subseteq\mathcal P$ satisfies i) $\mathcal P_0\subseteq\mathcal P^\delta$ for some $\delta>0$, ii) $\limsup_{n\to\infty}\sup_{\P\in\mathcal P_0} d_{BL_1}(\mathcal L_\P(\sqrt n(\hat\mu_n(\P)-\mu_0(\P))),N(0,\Sigma(\P)))=0$, and iii) $\sup_{\P\in\mathcal P_0}\max\{\lVert \Sigma(\P)\rVert ,\lVert M(\P)\rVert\}<\overline\sigma$ for some $\overline\sigma>0$.
\end{ass}

\begin{ass}[Variance estimator]\label{ass:var_uniform}
There exists an a.s. positive-semidefinite estimator $\hat\Sigma_n$, such that $\limsup_{n\to\infty}\sup_{\P\in\mathcal P_0}\P[\|\hat\Sigma_n-\Sigma(\P)\|>\varepsilon]=0$ for every $\varepsilon>0$.
\end{ass}

Assumption \ref{ass:admissible}.ii requires the Gaussian approximation to $\sqrt n(\hat\mu_n-\mu_0)$ to hold uniformly over $\mathcal P_0$, while Assumption \ref{ass:var_uniform} requires uniform consistency of $\hat\Sigma_n$ for the covariance matrix of that approximation. Assumption \ref{ass:admissible}.i imposes the $\delta$-condition, while iii) ensures tightness of the resulting Gaussian family. We do not require i.i.d. observations or variance positivity. 

Proposition \ref{theor_j_star} in the Supplemental Appendix implies that, for every $\P\in\mathcal P_0$, there exists a well-conditioned square $J^*\subseteq[q]$, with $\mathrm{sv}_d(M_{J^*}) > \delta$, inducing a basic optimal solution\footnote{If there are multiple such bases, $J^* = J^*(\P)$ refers to the one with the maximal $\mathrm{sv}_d(M_{J^*})$.},
\begin{align*}
    \theta(J^*)=M_{J^*}^{-1}c_{J^*}\in\mathcal S(\mu_0), \quad \lambda(J^*)=M_{J^*}^{-\prime}p\ge0.
\end{align*}

If $J^*$ were known, inference would be immediate. Since it is unknown, we leverage the sharp implications of $J^*$ inducing a basic optimal solution. Any square nonsingular $A$ induces a basic optimal solution if and only if it is primal and dual feasible,
\begin{equation}\label{eq:GI}
        M_{A'}\theta(A)-c_{A'}\ge0,
        \qquad
        \lambda(A)\ge0.
\end{equation}

We test the sample analogues of \eqref{eq:GI}. Collect the primal-dual violations in $\hat{g}(A)$ as $\hat{g}_{A'}(A) \equiv (\hat{c}_{A'} - \hat{M}_{A'} \hat{\theta}(A))$, and $\hat{g}_{A}(A) \equiv -\hat{\lambda}(A)$. The standard error for the $j$'th coordinate of the violation vector is $\hat s_j(A)$; its formula is provided in Appendix \ref{app_imp}. To improve power, we use GMS thresholding as in \citet{andrews2010inference}. Define
\[
\hat B(A)
=
\left\{
j\in [q]:
\frac{\sqrt n \hat{g}_{j}(A)}{\hat{s}_{j}(A)}
\ge -k_n
\right\},
\]
where $k_n >0$ is a slowly diverging deterministic sequence. The choice $k_n=\sqrt{2\log\log n}$ performs well in our simulations. The set $\hat B(A)$ contains inequalities that are violated or statistically close to binding; safely satisfied inequalities are excluded from the test. 

Define $\hat{\mathcal A}$ as the collection of all square $A\subseteq[q]$ such that $\mathrm{sv}_d(\hat M_A)>w_n^{-1}$, and
\begin{align}\label{test_of_a_basis}
\max_{j \in \hat{B}(A)} \frac{\sqrt n \hat{g}_{j}(A)}{\hat{s}_{j}(A)}\leq 
\hat q_{1-\alpha_{f,n}}(A),
\end{align}
where $\hat q_{1-\alpha_{f,n}}(A)$ is the analytical max-statistic critical value, see \eqref{cv_def} in Appendix \ref{app_imp}. The screening level $\alpha_{f,n}$ is assumed to vanish, so that $J^*$ is retained w.p.a.1 uniformly\footnote{One could treat $\alpha_{f,n} < \alpha$ as fixed by using a Bonferroni correction in Theorem \ref{thm:full_sample_uniform_inference}. We do not do this in Section \ref{mc_section}, as we find that including $\hat A$ in $\mathcal{C}_n$ (see below) always offsets the finite-sample screening error.}. 

Let $\hat\theta$ be a vertex-solution of the penalty problem\footnote{If $\Theta$ is not polyhedral, it may fail to exist in finite samples. In that case, set $\mathcal{C}_n = \hat{\mathcal{A}}$.}, selecting $\hat A = J(\hat{\theta}; \hat{\mu}_n)$. Pointwise, $\hat A$ is a nice face w.p.a.1, so we always include it in the set of candidates $\mathcal{C}_n \equiv \{\hat{A}\} \cup \hat{\mathcal{A}}$. Because $\hat A$ need not be square, and $M_{\hat A}$ may not have full column-rank, we compute its representer $\hat{r}$ via \eqref{reg_v} in Appendix \ref{app_imp}, using full-sample quantities and $n_2=n$.

Finally, let $\underline{\hat\sigma}(A)\equiv\max\{\sigma(A,\hat\theta(A),\hat{r}(A),\hat\Sigma_n),\underline\sigma\}$ for a fixed $\underline\sigma>0$\footnote{$\underline{\sigma} = 0$ is valid under uniform strict positivity of relevant variance components.}, where the representer $\hat{r}(A)$ is $\hat r$ if $A=\hat A$ and $\hat\lambda(A)$ otherwise. The lower confidence bound is 
\[
        \widehat{\mathrm{LCB}}_{1-\alpha}
        =
        \min_{A\in \mathcal{C}_n}
        \left\{
        p'\hat\theta(A)
        -
        z_{1-\alpha}\frac{\underline{\hat\sigma}(A)}{\sqrt n}
        \right\}.
\]

\begin{theorem}\label{thm:full_sample_uniform_inference}
Suppose that Assumptions \ref{ass:admissible} and \ref{ass:var_uniform} hold, $w_n\to\infty$ w.p.a.1 uniformly over $\mathcal P$ with $w_n/\sqrt n=o_{\mathcal P}(1)$, $k_n \uparrow \infty$, and $\alpha_{f,n}\downarrow 0$. Then, 
\begin{align}\label{eq: unif_sizec}
\liminf_{n\to\infty}\inf_{\P\in\mathcal P_0}
\P\!\left[
        V(\mu_0(\P))\geq \widehat{\mathrm{LCB}}_{1-\alpha}
\right]\geq 1-\alpha, \quad \alpha \in (0; 1/2).
\end{align}
Moreover, if $k_n = o(\sqrt{n})$, $\ln(\alpha_{f,n}) = o(n)$,  and either i) $\mathcal{S}(\mu_0) \subseteq \mathrm{Int}(\Theta)$, or ii) $\Theta$ is a polytope, then for any fixed $\P\in\mathcal P$, with $\sqrt{n}(\hat\mu_n-\mu_0) \xrightarrow{d}N(0,\Sigma)$, and $\hat\Sigma_n \xrightarrow{p} \Sigma$,
\begin{align}\label{eq: point_adapt}
        \widehat{\mathrm{LCB}}_{1-\alpha}
        =
        V(\mu_0(\P))+O_\P(n^{-1/2}), 
\end{align}
and if also $\mathcal S(\mu_0)=\{\theta^*\}$, LICQ holds at $\theta^*$, and $\sigma^* =\sigma(J^*, \theta^*,\lambda(J^*),\Sigma) \geq \underline{\sigma}$, then 
\begin{align}\label{eq: oracle_ef}
    \widehat{\mathrm{LCB}}_{1-\alpha} =p^\prime \hat{\theta}(J^*) - z_{1-\alpha}\frac{\sigma^*}{\sqrt{n}} + o_\P\left(\frac{1}{\sqrt{n}}\right), \quad \P[V(\mu_0) \geq \widehat{\mathrm{LCB}}_{1-\alpha}] \to 1-\alpha.
\end{align}
\end{theorem}
\begin{remark}\label{efficiency_remark}
Equation \eqref{eq: unif_sizec} establishes uniform size control over classes satisfying the $\delta$-condition, while \eqref{eq: point_adapt} asserts that excess length of $\widehat{\text{LCB}}_{1-\alpha}$ vanishes at $n^{-1/2}$ rate. With multiple basic optimal solutions, the minimum of the associated lower confidence bounds may be conservative. Under LICQ and a unique optimum, \eqref{eq: oracle_ef} shows that it is asymptotically equivalent to the oracle procedure that knows the identity of the unique optimal basis $J^*$. When $\mu_0$ and $\hat{\mu}_n$ collect unconditional and sample moments respectively, and the data is i.i.d., $\sigma^2(J^*,\theta^*,\lambda(J^*),\Sigma)$ has the same form as the semiparametric efficiency bound in \citet{kaido2014asymptotically}, although their result does not cover estimated $M$. When $M$ is known and $c = \E[W_i]$, their framework applies and the two variances coincide exactly.
\end{remark}

\begin{remark}\label{remark_comput}
Computationally, the pointwise inference procedure in Section \ref{section_inference} with $\hat r$ from \eqref{reg_v} amounts to solving two LPs and a quadratic problem. Thus, it is feasible even for large $q, d$ (e.g. $d = 100$). In contrast, computing the full set $\hat{\mathcal A}$ may be costly when $q$ and $d$ are large. Supplemental Appendix \ref{app: computation_unif} describes two algorithms that reduce this cost. Our simulations in Section \ref{mc_section} illustrate that these remain feasible in higher dimensions.
\end{remark}
\begin{remark}
Combining the construction in Remark \ref{remark_cht_ci} with the procedure in this section, one may obtain uniformly valid confidence regions for $p'\Theta_I$ and $\Theta_I$, as well as tests of point identification with uniform size control, see Supplemental Appendix \ref{app:two_sided_fs}.
\end{remark}
We conclude this section by comparing Theorem \ref{thm:full_sample_uniform_inference} with existing results. By constructing a confidence bound for $V(\mu_0)$ directly, our method avoids the conservativeness associated with projecting confidence regions valid for $\theta_0$ (e.g. \citet{andrews2010inference} and \citet{cox2023simple}\footnote{While \citet{cox2023simple} allow to profile out nuisance coordinates, these must enter the inequalities with a deterministic or conditionally known submatrix. This
restriction need not hold in our general estimated-$M$ setting.}), or $\Theta_I$ (e.g. \citet{CHT}) onto the direction $p$. Relative to procedures calibrated to perform one-sided\footnote{The two-sided targets differ: the extension in Supplemental Appendix \ref{app:two_sided_fs} gives confidence regions for $p'\Theta_I$, which is stronger than covering $p'\theta_0$ for $\theta_0 \in \Theta_I$, as in KMS. Covering $p'\Theta_I$ is necessary when the range of counterfactual outcomes permitted by the maintained restrictions is of interest, rather than the counterfactual generated by one true model. This distinction is relevant for robust policy choice (see \citet{HENRY2012256}).} subvector inference in linear moment inequality models (e.g. \citet{kaido2019confidence}, KMS), the key advantages of our procedure are its pointwise $\sqrt{n}-$consistency \eqref{eq: point_adapt} and, in regular cases, oracle equivalence \eqref{eq: oracle_ef} guarantees. These have not been shown for KMS to the best of our knowledge, and may explain why the simulated excess lengths of our confidence regions compare favorably to those of KMS in Section \ref{mc_section}. Moreover, by avoiding parameter-space test-inversion, our procedure exhibits a sizable computational advantage in the simulations. These advantages come at the cost of relying on the $\delta$-condition for uniform validity, which is not assumed by KMS. 

\section{Empirically calibrated Monte Carlo}\label{mc_section}
This section assesses the finite-sample performance of our estimation and inference procedures using Monte Carlo simulations. To reflect the empirical relevance of our procedures, we tightly calibrate the simulations to the setting of \citet{holmes2011diffusion}, a canonical application of the linear moment inequality approach in industrial organization.

\subsection{Design}
\citet{holmes2011diffusion} studies Walmart's dynamic store-location problem, using revealed preference from observed location choices to infer the value of economies of density. Intuitively, opening stores close to existing locations increases cannibalization, but may be optimal if the resulting proximity to distribution centers generates sufficiently large savings.

The parameter of primary economic interest is $\tau$, the per-period cost of serving a store located one mile farther from the closest relevant distribution center. While direct trucking costs provide a natural benchmark, the central empirical question is whether the total value of density economies is substantially larger. In the partially identified setting, answering this question amounts to obtaining a lower bound on $\tau$. Holmes estimates that this lower bound is almost four times as large as the direct cost of trucking.

\paragraph*{Linear moment inequalities.}
\citet{holmes2011diffusion} uses deviations $a$ from Walmart's observed store opening policy to construct revealed-preference restrictions that partially identify $\tau$. For each $a$, $\tilde y_a$ is the measured difference in the present value of operating profits under the observed opening policy and $a$, while $x_a$ collects the corresponding differences in the present values of distribution-distance miles and log population density around the stores, as well as its square. Optimality of the observed policy implies $y_a\geq x_a'\theta$ for $\theta=(\tau ~ \omega_1 ~\omega_2)'$, where $\tilde y_a=y_a+\epsilon_a$ is measured with error. The baseline specification partitions the deviations $a$ into $K=12$ groups to construct nonnegative instruments $z_a\equiv (z_{a,k})^{12}_{k=1}$. Assuming instrument exogeneity, \citet{holmes2011diffusion} obtains linear moment inequalities
\begin{align}\label{moment_ineq}
    \E[z_a\tilde y_a]
    \geq 
    \E[z_{a}x_a']\theta.
\end{align}
This maps directly into \eqref{lp_init}. Define $c_k \equiv -\E[z_{a,k}\tilde y_a], ~
M_{k} \equiv -\E[z_{a,k}x'_a],$ for $k = 1, \dots, 12$ and append $M$ and $c$ with the inequalities defining the set $\Theta$, as described below. The identified set is $\Theta_I=\{\theta\in \mathbb{R}^d: M\theta\geq c\}$, so the lower bound on $\tau$ is obtained by taking $p=e_1$ in \eqref{lp_init}; the upper bound is obtained analogously as $-V(\mu_0)$ with $p=-e_1$. 
\paragraph*{Calibration.}
 The natural sampling unit is a store location. The data contain $n=3,176$ stores, from which Holmes constructs $522,967$ deviations that resequence the opening of two stores. Our calibrated population polytope, defined by $(M,c)$, is based on 12 inequalities estimated using the deviation-level data in the replication package of \citet{holmes2011diffusion}. In each simulation, we then draw $n$ Gaussian perturbations around the calibrated parameters, so that for the resulting sample averages,
\begin{align*}
    \sqrt{n}\,\text{vec}(\hat M_n-M,\hat c_n-c)
    \sim N(0,\Sigma),
\end{align*}
where, following Holmes, we estimate $\Sigma$ by subsampling stores and recomputing the moments from the deviations for which both stores were sampled.\footnote{This is justified because the mean deviations may be written as smooth functions of empirical moments. We ignore first-stage demand estimation error in \(y_a\) when calibrating \(\Sigma\), as the replication files do not contain the required store-level data; Table XI in \citet{holmes2011diffusion} shows that accounting for it has little effect on CI length.} This design preserves the empirical covariance, while allowing us to vary the sample size and the geometry of the population polytope in a controlled way, as explained below. 

The geometry of the estimated polytope is quite unstable in this application. In particular, $7$ out of $12$ robustness checks reported in Table XIII in \citet{holmes2011diffusion} produce an empty estimated (plug-in) polytope. When subsampling the stores, we find that the estimated polytope is empty in $7\%$ of the replications. Conditional on non-emptiness, in $63\%$ of the draws the polytope's vertex active sets differ from the full-sample ones.

\paragraph*{Irregular designs.}
 Empty or thin estimated feasible sets arise commonly in empirical setups utilizing linear revealed-preference inequalities (e.g. \citet{ho2009}, \citet{wollmann2018trucks} and \citet{houde2023nexus}). We use two families of perturbations of Holmes's polytope to stress-test our procedures, see Figure \ref{fig:holmes-collapse-designs}. The first is a `row-collapse' design. We find a pair $(r,s)$ of rows binding on $\Theta_I$ that impose opposing restrictions on nearly the same linear combination of parameters\footnote{The pair $r, s$ was chosen from facets of $\Theta_I$ so that the cosine distance between the original $h_r, -h_s$ is minimized, with the corresponding angle of around $5.3^\circ$. This pair also minimizes the cosine distance between the normals $M_r, -M_{s}$ over all estimated rows of $M$, with the angle of just $0.08^\circ$. }, and move one of them toward the opposite of the other. Writing $h_j=(M_{j},c_{j})$, the path replaces row $h_r$ by the normalized convex combination
\begin{align*}
    h_r(\ell)
    =
    \|h_r\|\,
    \frac{(1-\ell)(-h_s)+\ell h_r}
         {\|(1-\ell)(-h_s)+\ell h_r\|},
    \quad \ell\in[0,1],
\end{align*}
leaving other rows unchanged. Thus, $\ell=1$ gives the original polytope, while at $\ell=0$ the selected pair jointly imposes $M_s\theta = c_s$. The linear function $M_s\theta$ therefore becomes point-identified, although $\theta$ remains partially identified. As $\ell$ decreases from $1$ to $0$, the identified set progressively
shrinks towards a two-dimensional limit. The normalization prevents changes in the scale of row $r$ relative to the fixed sampling-noise covariance $\Sigma$ from confounding the effect of changing the polytope's geometry.
 
The second irregular design we consider is `centroid-collapse'. Let $\mathcal{V}$ be the vertex set of the baseline polytope and let $\bar\theta_V=|\mathcal{V}|^{-1}\sum_{\theta\in\mathcal{V}} \theta$ be its vertex centroid. For $\ell \in[0,1]$, keep $M$ fixed and replace $c$ by $\ell c + (1-\ell)M\bar\theta_V.$ Note that $\ell=1$ gives the original polytope, while $\ell=0$ collapses the polytope to the centroid, resulting in point-identification of $\theta = \overline{\theta}_V$. For $\ell \in (0,1)$, the original polytope is scaled by $\ell$ around $\bar\theta_V$. 

We keep the covariance matrix $\Sigma$ fixed, so smaller $\ell$ yield tighter identified sets relative to the same sampling noise. At $\ell=0$, SC and LICQ (at all vertices) fail in both designs. 
\begin{figure}
\centering
\begingroup
\definecolor{collapseOriginal}{HTML}{7B858E}
\definecolor{collapseHalf}{HTML}{267899}
\definecolor{collapseLimit}{HTML}{B94D39}
\captionsetup[subfigure]{position=top}
\begin{subfigure}[t]{.48\linewidth}
\centering
\caption{Row-collapse}
\label{fig:holmes-row-collapse}
\begin{tikzpicture}[x=.19\linewidth,y=.1425\linewidth,line join=round,line cap=round,font=\small]
\path[use as bounding box] (.5,.45) rectangle (5.2,6.05);
\filldraw[fill=black!2,draw=black!25,densely dashed,line width=.4pt] (1.92346,1.17500) -- (4.94049,1.22084) -- (4.67654,1.74474) -- (1.65951,1.69890) -- cycle;
\draw[black!55,line width=.5pt] (1.92346,1.17500) -- (1.65951,1.69890);
\draw[black!55,line width=.4pt] (1.91146,1.19881) -- (1.91146,1.14381);
\node[font=\scriptsize,anchor=north east,inner sep=2pt] at (1.91146,1.14381) {$3.33$};
\draw[black!55,line width=.4pt] (1.67150,1.67509) -- (1.67150,1.62009);
\node[font=\scriptsize,anchor=north east,inner sep=2pt] at (1.67150,1.62009) {$4.92$};
\node at (1.24148,1.21695) {$\tau$};
\draw[black!55,line width=.5pt] (4.94049,1.22084) -- (1.92346,1.17500);
\draw[black!55,line width=.4pt] (4.80336,1.21875) -- (4.80336,1.16375);
\node[font=\scriptsize,anchor=north,inner sep=2pt] at (4.80336,1.16375) {$2.04$};
\draw[black!55,line width=.4pt] (2.06060,1.17708) -- (2.06060,1.12208);
\node[font=\scriptsize,anchor=north west,inner sep=2pt] at (2.06060,1.12208) {$4.30$};
\node (omegaOneLabel) at (3.61198,0.77792) {$\omega_1$};
\draw[black!55,line width=.5pt] (1.65951,1.69890) -- (1.65951,5.77916);
\draw[black!55,line width=.4pt] (1.65951,2.93219) -- (1.59451,2.93219);
\node[font=\scriptsize,anchor=east,inner sep=2pt] at (1.59451,2.93219) {$-0.50$};
\draw[black!55,line width=.4pt] (1.65951,5.64359) -- (1.59451,5.64359);
\node[font=\scriptsize,anchor=east,inner sep=2pt] at (1.59451,5.64359) {$-0.17$};
\path let \p1=($(4.94049,1.22084)!(omegaOneLabel.north)!(1.92346,1.17500)$),
          \p2=(omegaOneLabel.north), \n1={veclen(\x1-\x2,\y1-\y2)}
      in (1.65951,3.73903) ++(-\n1,0)
      node[anchor=east] {$\omega_2$};
\fill[collapseOriginal,opacity=0.08] (4.62724,1.56830) -- (3.17368,1.69791) -- (1.98520,1.32673) -- (2.07230,1.20126) -- (3.43560,1.29781) -- cycle;
\fill[collapseHalf,opacity=0.14] (4.62724,1.56830) -- (3.17368,1.69791) -- (2.27408,1.41695) -- (2.40556,1.22486) -- (3.43560,1.29781) -- cycle;
\fill[collapseLimit,opacity=0.17] (3.17368,1.69791) -- (3.43560,1.29781) -- (4.62724,1.56830) -- cycle;
\draw[collapseOriginal!75!black,dash pattern=on 2pt off 1.5pt,line width=.65pt] (4.62724,1.56830) -- (2.43659,1.28524);
\draw[collapseOriginal!75!black,dash pattern=on 2pt off 1.5pt,line width=.65pt] (2.43659,1.28524) -- (2.07230,1.20126);
\draw[collapseOriginal!75!black,dash pattern=on 2pt off 1.5pt,line width=.65pt] (2.43659,1.28524) -- (1.98520,1.32673);
\draw[collapseOriginal!75!black,dash pattern=on 2pt off 1.5pt,line width=.65pt] (3.43560,1.29781) -- (3.17368,1.69791);
\draw[collapseOriginal!75!black,line width=.65pt] (4.62724,1.56830) -- (3.17368,1.69791) -- (1.98520,1.32673) -- (2.07230,1.20126) -- (3.43560,1.29781) -- cycle;
\draw[collapseHalf!75!black,line width=.65pt] (4.62724,1.56830) -- (3.17368,1.69791) -- (2.27408,1.41695) -- (2.40556,1.22486) -- (3.43560,1.29781) -- cycle;
\draw[collapseLimit!75!black,line width=.65pt] (3.17368,1.69791) -- (3.43560,1.29781) -- (4.62724,1.56830) -- cycle;
\draw[black!55,dash pattern=on 2pt off 1.5pt,line width=.4pt] (4.62724,5.51299) -- (4.62724,1.56830);
\draw[black!55,dash pattern=on 2pt off 1.5pt,line width=.4pt] (2.43659,2.81983) -- (2.43659,1.28524);
\draw[black!55,dash pattern=on 2pt off 1.5pt,line width=.4pt] (3.43560,4.30317) -- (3.43560,1.29781);
\draw[black!55,dash pattern=on 2pt off 1.5pt,line width=.4pt] (2.07230,2.45280) -- (2.07230,1.20126);
\draw[black!55,dash pattern=on 2pt off 1.5pt,line width=.4pt] (3.17368,4.66455) -- (3.17368,1.69791);
\draw[black!55,dash pattern=on 2pt off 1.5pt,line width=.4pt] (1.98520,2.56002) -- (1.98520,1.32673);
\fill[collapseOriginal!32,opacity=.50] (4.62724,5.51299) -- (3.43560,4.30317) -- (3.17368,4.66455) -- cycle;
\fill[collapseHalf!32,opacity=.65] (4.62724,5.51299) -- (3.43560,4.30317) -- (3.17368,4.66455) -- cycle;
\fill[collapseLimit!55,opacity=.65] (4.62724,5.51299) -- (3.43560,4.30317) -- (3.17368,4.66455) -- cycle;
\fill[collapseHalf!23,opacity=.65] (2.40556,2.90513) -- (2.27408,3.07157) -- (3.17368,4.66455) -- (3.43560,4.30317) -- cycle;
\fill[collapseOriginal!23,opacity=.50] (2.07230,2.45280) -- (1.98520,2.56002) -- (3.17368,4.66455) -- (3.43560,4.30317) -- cycle;
\fill[collapseHalf!80,opacity=.65] (4.62724,5.51299) -- (2.96678,3.47164) -- (2.40556,2.90513) -- (3.43560,4.30317) -- cycle;
\fill[collapseOriginal!80,opacity=.50] (4.62724,5.51299) -- (2.43659,2.81983) -- (2.07230,2.45280) -- (3.43560,4.30317) -- cycle;
\draw[collapseOriginal,line width=.65pt,dash pattern=on 2pt off 2pt,opacity=.80] (2.43659,2.81983) -- (1.98520,2.56002);
\draw[collapseOriginal,line width=.65pt] (4.62724,5.51299) -- (2.43659,2.81983);
\draw[collapseOriginal,line width=.65pt] (2.43659,2.81983) -- (2.07230,2.45280);
\draw[collapseOriginal,line width=.65pt] (3.43560,4.30317) -- (2.07230,2.45280);
\draw[collapseOriginal,line width=.65pt] (4.62724,5.51299) -- (3.43560,4.30317);
\draw[collapseOriginal,line width=.65pt] (3.17368,4.66455) -- (1.98520,2.56002);
\draw[collapseOriginal,line width=.65pt] (4.62724,5.51299) -- (3.17368,4.66455);
\draw[collapseOriginal,line width=.65pt] (3.43560,4.30317) -- (3.17368,4.66455);
\draw[collapseOriginal,line width=.65pt] (2.07230,2.45280) -- (1.98520,2.56002);
\draw[collapseHalf,line width=.9pt,dash pattern=on 2pt off 2pt,opacity=.80] (2.96678,3.47164) -- (2.27408,3.07157);
\draw[collapseHalf,line width=.9pt] (4.62724,5.51299) -- (2.96678,3.47164);
\draw[collapseHalf,line width=.9pt] (2.96678,3.47164) -- (2.40556,2.90513);
\draw[collapseHalf,line width=.9pt] (3.43560,4.30317) -- (2.40556,2.90513);
\draw[collapseHalf,line width=.9pt] (4.62724,5.51299) -- (3.43560,4.30317);
\draw[collapseHalf,line width=.9pt] (3.17368,4.66455) -- (2.27408,3.07157);
\draw[collapseHalf,line width=.9pt] (4.62724,5.51299) -- (3.17368,4.66455);
\draw[collapseHalf,line width=.9pt] (3.43560,4.30317) -- (3.17368,4.66455);
\draw[collapseHalf,line width=.9pt] (2.40556,2.90513) -- (2.27408,3.07157);
\fill[collapseOriginal] (2.43659,2.81983) circle[radius=0.975pt];
\fill[collapseOriginal] (2.07230,2.45280) circle[radius=0.975pt];
\fill[collapseOriginal] (1.98520,2.56002) circle[radius=0.975pt];
\fill[collapseHalf] (2.96678,3.47164) circle[radius=0.975pt];
\fill[collapseHalf] (2.40556,2.90513) circle[radius=0.975pt];
\fill[collapseHalf] (2.27408,3.07157) circle[radius=0.975pt];
\draw[collapseLimit,line width=1.05pt] (4.62724,5.51299) -- (3.43560,4.30317);
\draw[collapseLimit,line width=1.05pt] (3.43560,4.30317) -- (3.17368,4.66455);
\draw[collapseLimit,line width=1.05pt] (3.17368,4.66455) -- (4.62724,5.51299);
\fill[collapseLimit] (4.62724,5.51299) circle[radius=0.975pt];
\fill[collapseLimit] (3.43560,4.30317) circle[radius=0.975pt];
\fill[collapseLimit] (3.17368,4.66455) circle[radius=0.975pt];
\begin{scope}[overlay,x=1cm,y=1cm]
\node[anchor=south west,opacity=0] (faceLabelOriginal) at (.75,5.98) {$M_s\theta=c_s$};
\coordinate (faceTargetOriginal) at (3.74551,4.82691);
\end{scope}
\path let \p1=($(faceLabelOriginal.south)-(faceTargetOriginal)$)
      in (3.74551,4.82691) ++({.5625*\x1},{.5625*\y1})
      node[anchor=south,text=collapseLimit] (faceLabel) {$M_s\theta=c_s$};
\draw[collapseLimit,line width=.5pt] (faceLabel.south) -- (3.74551,4.82691);
\end{tikzpicture}
\end{subfigure}\hfill
\begin{subfigure}[t]{.48\linewidth}
\centering
\caption{Centroid-collapse}
\label{fig:holmes-centroid-collapse}
\begin{tikzpicture}[x=.19\linewidth,y=.1425\linewidth,line join=round,line cap=round,font=\small]
\path[use as bounding box] (.5,.45) rectangle (5.2,6.05);
\filldraw[fill=black!2,draw=black!25,densely dashed,line width=.4pt] (1.92346,1.17500) -- (4.94049,1.22084) -- (4.67654,1.74474) -- (1.65951,1.69890) -- cycle;
\draw[black!55,line width=.5pt] (1.92346,1.17500) -- (1.65951,1.69890);
\draw[black!55,line width=.4pt] (1.91146,1.19881) -- (1.91146,1.14381);
\node[font=\scriptsize,anchor=north east,inner sep=2pt] at (1.91146,1.14381) {$3.33$};
\draw[black!55,line width=.4pt] (1.67150,1.67509) -- (1.67150,1.62009);
\node[font=\scriptsize,anchor=north east,inner sep=2pt] at (1.67150,1.62009) {$4.92$};
\node at (1.24148,1.21695) {$\tau$};
\draw[black!55,line width=.5pt] (4.94049,1.22084) -- (1.92346,1.17500);
\draw[black!55,line width=.4pt] (4.80336,1.21875) -- (4.80336,1.16375);
\node[font=\scriptsize,anchor=north,inner sep=2pt] at (4.80336,1.16375) {$2.04$};
\draw[black!55,line width=.4pt] (2.06060,1.17708) -- (2.06060,1.12208);
\node[font=\scriptsize,anchor=north west,inner sep=2pt] at (2.06060,1.12208) {$4.30$};
\node (omegaOneLabel) at (3.61198,0.77792) {$\omega_1$};
\draw[black!55,line width=.5pt] (1.65951,1.69890) -- (1.65951,5.77916);
\draw[black!55,line width=.4pt] (1.65951,2.93219) -- (1.59451,2.93219);
\node[font=\scriptsize,anchor=east,inner sep=2pt] at (1.59451,2.93219) {$-0.50$};
\draw[black!55,line width=.4pt] (1.65951,5.64359) -- (1.59451,5.64359);
\node[font=\scriptsize,anchor=east,inner sep=2pt] at (1.59451,5.64359) {$-0.17$};
\path let \p1=($(4.94049,1.22084)!(omegaOneLabel.north)!(1.92346,1.17500)$),
          \p2=(omegaOneLabel.north), \n1={veclen(\x1-\x2,\y1-\y2)}
      in (1.65951,3.73903) ++(-\n1,0)
      node[anchor=east] {$\omega_2$};
\fill[collapseOriginal,opacity=0.08] (4.62724,1.56830) -- (3.17368,1.69791) -- (1.98520,1.32673) -- (2.07230,1.20126) -- (3.43560,1.29781) -- cycle;
\fill[collapseHalf,opacity=0.14] (3.79117,1.48225) -- (3.06439,1.54706) -- (2.47015,1.36147) -- (2.51370,1.29873) -- (3.19535,1.34701) -- cycle;
\draw[collapseOriginal!75!black,dash pattern=on 2pt off 1.5pt,line width=.65pt] (4.62724,1.56830) -- (2.43659,1.28524);
\draw[collapseOriginal!75!black,dash pattern=on 2pt off 1.5pt,line width=.65pt] (2.43659,1.28524) -- (2.07230,1.20126);
\draw[collapseOriginal!75!black,dash pattern=on 2pt off 1.5pt,line width=.65pt] (2.43659,1.28524) -- (1.98520,1.32673);
\draw[collapseOriginal!75!black,dash pattern=on 2pt off 1.5pt,line width=.65pt] (3.43560,1.29781) -- (3.17368,1.69791);
\draw[collapseOriginal!75!black,line width=.65pt] (4.62724,1.56830) -- (3.17368,1.69791) -- (1.98520,1.32673) -- (2.07230,1.20126) -- (3.43560,1.29781) -- cycle;
\draw[collapseHalf!75!black,line width=.65pt] (3.79117,1.48225) -- (3.06439,1.54706) -- (2.47015,1.36147) -- (2.51370,1.29873) -- (3.19535,1.34701) -- cycle;
\fill[collapseLimit] (2.95510,1.39621) circle[radius=1.275pt];
\draw[black!55,dash pattern=on 2pt off 1.5pt,line width=.4pt] (4.62724,5.51299) -- (4.62724,1.56830);
\draw[black!55,dash pattern=on 2pt off 1.5pt,line width=.4pt] (2.43659,2.81983) -- (2.43659,1.28524);
\draw[black!55,dash pattern=on 2pt off 1.5pt,line width=.4pt] (3.43560,4.30317) -- (3.43560,1.29781);
\draw[black!55,dash pattern=on 2pt off 1.5pt,line width=.4pt] (2.07230,2.45280) -- (2.07230,1.20126);
\draw[black!55,dash pattern=on 2pt off 1.5pt,line width=.4pt] (3.17368,4.66455) -- (3.17368,1.69791);
\draw[black!55,dash pattern=on 2pt off 1.5pt,line width=.4pt] (1.98520,2.56002) -- (1.98520,1.32673);
\fill[collapseOriginal!32,opacity=.50] (4.62724,5.51299) -- (3.43560,4.30317) -- (3.17368,4.66455) -- cycle;
\fill[collapseHalf!32,opacity=.65] (3.79117,4.61594) -- (3.19535,4.01103) -- (3.06439,4.19172) -- cycle;
\fill[collapseHalf!23,opacity=.65] (2.51370,3.08585) -- (2.47015,3.13946) -- (3.06439,4.19172) -- (3.19535,4.01103) -- cycle;
\fill[collapseOriginal!23,opacity=.50] (2.07230,2.45280) -- (1.98520,2.56002) -- (3.17368,4.66455) -- (3.43560,4.30317) -- cycle;
\fill[collapseHalf!80,opacity=.65] (3.79117,4.61594) -- (2.69585,3.26936) -- (2.51370,3.08585) -- (3.19535,4.01103) -- cycle;
\fill[collapseOriginal!80,opacity=.50] (4.62724,5.51299) -- (2.43659,2.81983) -- (2.07230,2.45280) -- (3.43560,4.30317) -- cycle;
\draw[collapseOriginal,line width=.65pt,dash pattern=on 2pt off 2pt,opacity=.80] (2.43659,2.81983) -- (1.98520,2.56002);
\draw[collapseOriginal,line width=.65pt] (4.62724,5.51299) -- (2.43659,2.81983);
\draw[collapseOriginal,line width=.65pt] (2.43659,2.81983) -- (2.07230,2.45280);
\draw[collapseOriginal,line width=.65pt] (3.43560,4.30317) -- (2.07230,2.45280);
\draw[collapseOriginal,line width=.65pt] (4.62724,5.51299) -- (3.43560,4.30317);
\draw[collapseOriginal,line width=.65pt] (3.17368,4.66455) -- (1.98520,2.56002);
\draw[collapseOriginal,line width=.65pt] (4.62724,5.51299) -- (3.17368,4.66455);
\draw[collapseOriginal,line width=.65pt] (3.43560,4.30317) -- (3.17368,4.66455);
\draw[collapseOriginal,line width=.65pt] (2.07230,2.45280) -- (1.98520,2.56002);
\draw[collapseHalf,line width=.9pt,dash pattern=on 2pt off 2pt,opacity=.80] (2.69585,3.26936) -- (2.47015,3.13946);
\draw[collapseHalf,line width=.9pt] (3.79117,4.61594) -- (2.69585,3.26936);
\draw[collapseHalf,line width=.9pt] (2.69585,3.26936) -- (2.51370,3.08585);
\draw[collapseHalf,line width=.9pt] (3.19535,4.01103) -- (2.51370,3.08585);
\draw[collapseHalf,line width=.9pt] (3.79117,4.61594) -- (3.19535,4.01103);
\draw[collapseHalf,line width=.9pt] (3.06439,4.19172) -- (2.47015,3.13946);
\draw[collapseHalf,line width=.9pt] (3.79117,4.61594) -- (3.06439,4.19172);
\draw[collapseHalf,line width=.9pt] (3.19535,4.01103) -- (3.06439,4.19172);
\draw[collapseHalf,line width=.9pt] (2.51370,3.08585) -- (2.47015,3.13946);
\fill[collapseOriginal] (4.62724,5.51299) circle[radius=0.975pt];
\fill[collapseOriginal] (2.43659,2.81983) circle[radius=0.975pt];
\fill[collapseOriginal] (3.43560,4.30317) circle[radius=0.975pt];
\fill[collapseOriginal] (2.07230,2.45280) circle[radius=0.975pt];
\fill[collapseOriginal] (3.17368,4.66455) circle[radius=0.975pt];
\fill[collapseOriginal] (1.98520,2.56002) circle[radius=0.975pt];
\fill[collapseHalf] (3.79117,4.61594) circle[radius=0.975pt];
\fill[collapseHalf] (2.69585,3.26936) circle[radius=0.975pt];
\fill[collapseHalf] (3.19535,4.01103) circle[radius=0.975pt];
\fill[collapseHalf] (2.51370,3.08585) circle[radius=0.975pt];
\fill[collapseHalf] (3.06439,4.19172) circle[radius=0.975pt];
\fill[collapseHalf] (2.47015,3.13946) circle[radius=0.975pt];
\filldraw[fill=collapseLimit,draw=white,line width=0.525pt] (2.95510,3.71889) circle[radius=2.025pt];
\begin{scope}[overlay,x=1cm,y=1cm]
\node[anchor=south west,opacity=0] (centroidLabelOriginal) at (4.1,5.98) {$\bar\theta_V$};
\coordinate (centroidTargetOriginal) at (3.02510,3.80889);
\end{scope}
\path let \p1=($(centroidLabelOriginal.south)-(centroidTargetOriginal)$),
          \n1={.5625*veclen(\x1,\y1)}
      in (3.02510,3.71889) ++(\n1,0)
      node[anchor=west,text=collapseLimit] (centroidLabel) {$\bar\theta_V$};
\draw[collapseLimit,line width=.5pt] (centroidLabel.west) -- (3.02510,3.71889);
\end{tikzpicture}
\end{subfigure}
\par\medskip
\begin{tikzpicture}[x=1cm,y=1cm,line join=round,line cap=round,font=\small]
\filldraw[fill=collapseOriginal,fill opacity=0.1,draw=collapseOriginal,line width=.75pt] (3.05000,0.19000) rectangle (3.47000,0.42000);
\node[anchor=west,inner sep=0pt] at (3.62000,0.30500) {$\ell=1$};
\filldraw[fill=collapseHalf,fill opacity=0.19,draw=collapseHalf,line width=.75pt] (6.25000,0.19000) rectangle (6.67000,0.42000);
\node[anchor=west,inner sep=0pt] at (6.82000,0.30500) {$\ell=0.5$};
\filldraw[fill=collapseLimit,fill opacity=0.3,draw=collapseLimit,line width=.75pt] (9.05000,0.19000) rectangle (9.47000,0.42000);
\node[anchor=west,inner sep=0pt] at (9.62000,0.30500) {$\ell=0$};
\end{tikzpicture}
\endgroup
\caption{Population identified set in the collapse designs.}
\label{fig:holmes-collapse-designs}
\end{figure}

\subsection{Estimation}

\paragraph*{Benchmarks.}
We compare $\hat V$ with two criterion-based estimators. The first is the minimum-$L_2$-violation estimator used by \citet{holmes2011diffusion}, which minimizes $p'\theta$ over $\arg\min_{\theta\in\Theta}\hat Q(\theta)$, where
$\hat Q(\theta)=\|(\hat c_n-\hat M_n\theta)^+\|^2$; see Section
\ref{sec: available_estim}. Whenever the sample LP is feasible, this estimator
coincides with the plug-in estimator $V(\hat\mu_n)$.\footnote{We do not report
$V(\hat\mu_n)$ separately because the sample LP is frequently infeasible in
the collapsed designs. In the original design, it is infeasible in only
$0.32\%$ of simulations, so the two estimators are nearly identical.} The second benchmark is a set-expansion estimator based on
\citet{CHT}. It minimizes $p'\theta$ over
$\hat\Theta_I^{\mathrm{SE}}
=\{\theta\in\Theta:\hat Q(\theta)\leq
\max\{\hat Q^{\min}+\tilde{k}_n/n,\mathcal Q_{\tilde\alpha}/n\}\}$,
where $\hat Q^{\min}=\min_{\theta\in\Theta}\hat Q(\theta)$ and
$\mathcal Q_{\tilde\alpha}$ is the oracle $\tilde\alpha$-quantile of
$\sup_{\theta\in\Theta_I}n\hat Q(\theta)$ under the given DGP.
\paragraph*{Parameters.}
Throughout Section \ref{mc_section}, the penalty vector is
\[
(w_n)_j
=
\frac{\|p\|d\log n}
{\|\hat M_{j}\|\delta_{0.15}\log 100},
\qquad
j\in[q],
\]
where
$\delta_{\bar\alpha}
=(\sqrt{1-2\log(1-\bar\alpha)}-1)$.
Supplemental Appendix \ref{ap_pensel} motivates this choice, and Table \ref{tab:holmes-penalty-tilde-alpha-robustness-n3176-n10000} reports sensitivity to alternative values of $\bar\alpha$.

For set expansion, we use $\tilde{k}_n=0$ and
$\tilde\alpha=0.5$. These choices deliberately favor the set expansion estimator.
Absent Condition C.3 of \citet{CHT}, its consistency generally requires a
diverging expansion sequence, while interpreting $\hat{\Theta}_I^{\text{SE}}$ as a
confidence region would ordinarily suggest a larger value of
$\tilde\alpha$, e.g. $0.95$. Either change weakly enlarges the contour and makes the lower bound more conservative, while even at $\tilde{k}_n=0$ and
$\tilde\alpha=0.5$ the estimator is substantially conservative in all designs reported in Table \ref{tab:holmes-estimator-collapsed-paths}.

We take $\Theta$ to be a coordinatewise threefold enlargement of the
original identified set's projections, intersected with Holmes's sign
restrictions $\omega_1\geq0$ and $\omega_2\leq0$. The original projections
are
$[3.33,4.92]\times[2.04,4.30]\times[-0.50,-0.17]$, while
$\Theta=[1.75,6.50]\times[0,6.56]\times[-0.83,0]$. Inequalities forming $\Theta$ are not altered in the collapse designs. Across all designs in Table \ref{tab:holmes-estimator-collapsed-paths}, the boundary of $\Theta$ determines the selected penalty vertex in at most $0.08\%$ of simulations. 

\paragraph*{Results.}
Table \ref{tab:holmes-estimator-collapsed-paths} reports relative root-MSE and
bias of the considered estimators. In the original design the sample LP is infeasible in only 0.32\% of the simulations, so the Min. $L_2$ estimator mostly coincides with the plug-in estimator. The penalty and Min.\ $L_2$ estimators perform similarly in this design. The only material difference is for the upper bound on $\omega_1$, where the root-MSE of Min.\ $L_2$ is $12\%$ larger at $n=3176$. This similarity may be a positive finding: in a related setting with regular convex moment-inequality models, \citet{kaido2014asymptotically} show that the plug-in estimator is semiparametrically efficient. Thus, when $\Theta_I$ is more `regular', there may be little scope to improve upon the plug-in estimator. The relevant question is whether penalty remains stable as regularity deteriorates.\footnote{The results of \citet{kaido2014asymptotically} do not cover the estimated $M$ setting. We nevertheless conjecture that, under SC, LICQ, and NFF, $V(\hat\mu_n)$, $\hat V$, and the minimum-$L_2$ estimator are asymptotically equivalent and efficient. Under these conditions, the unique optimal basis $J^*$ is selected w.p.a.1, and the common asymptotic variance equals $\sigma^2(J^*, \theta^*, \lambda(J^*),\Sigma)$, see Remark \ref{efficiency_remark}.}

The difference between penalty and Min. $L_2$ becomes more pronounced as the identified set approaches an irregular configuration.  At the row-collapse limit, the RMSE of Min.\ $L_2$ is $2.21$ times that of penalty at $n=3176$, rising to $3.72$ at $n=10000$. At the centroid collapse limit, the gap is smaller, at around $15\%$ for both sample sizes. This contrast aligns with the theory: centroid collapse yields point identification, whereas row collapse produces a two-dimensional identified set. Min. $L_2$ is consistent in the former case, but may fail to be consistent in the latter. 

The Min. $L_2$ estimator also develops an anticonservative bias near $\ell = 0$. For the row-collapse, it tends to overstate the lower bound on $\tau$ by around $20\%$ of the original identified interval's length at both sample sizes. This corresponds to the set-coverage failure discussed in Section \ref{sec: available_estim}: a non-expanded criterion set may omit parts of $\Theta_I$ even in large samples, causing its projected lower endpoint to be too large. Set expansion removes this direction of bias, but overcorrects. Its RMSE exceeds that of penalty in every reported design due to a large conservative bias.

\paragraph*{Summary.} Our simulations indicate that the penalty estimator performs similarly to the plug-in estimator in the more regular designs, while remaining more stable in the irregular configurations. It does not inherit the anticonservative bias of the Min. $L_2$ estimator in the collapse limits, and avoids the substantial conservativeness of the set expansion estimator. 

\begin{table}
\centering
\caption{Estimator comparison for the original and collapsed Holmes polytopes}
\footnotesize
\label{tab:holmes-estimator-collapsed-paths}
\setlength{\tabcolsep}{2.8pt}
\begin{tabular}{l*{12}{c}}
\toprule
 & \multicolumn{6}{c}{$n=3176$} & \multicolumn{6}{c}{$n=10000$} \\
\cmidrule(lr){2-7}\cmidrule(lr){8-13}
 & \multicolumn{2}{c}{Penalty} & \multicolumn{2}{c}{Min. $L_2$} & \multicolumn{2}{c}{Set Exp.} & \multicolumn{2}{c}{Penalty} & \multicolumn{2}{c}{Min. $L_2$} & \multicolumn{2}{c}{Set Exp.} \\
\cmidrule(lr){2-3}\cmidrule(lr){4-5}\cmidrule(lr){6-7}\cmidrule(lr){8-9}\cmidrule(lr){10-11}\cmidrule(lr){12-13}
 & RRMSE & Bias & RRMSE & Bias & RRMSE & Bias & RRMSE & Bias & RRMSE & Bias & RRMSE & Bias \\
\midrule
\multicolumn{1}{l}{Direction} & \multicolumn{12}{c}{\emph{Original polytope}} \\
\midrule
$\tau$ Lower & 1.00 & 0.01 & 1.00 & 0.01 & 1.83 & -0.27 & 1.00 & 0.00 & 1.00 & 0.00 & 1.88 & -0.15 \\
$\tau$ Upper & 1.00 & -0.02 & 1.00 & -0.02 & 1.68 & 0.34 & 1.00 & 0.00 & 1.00 & 0.00 & 1.73 & 0.20 \\
$\omega_1$ Lower & 1.00 & 0.00 & 1.00 & 0.00 & 2.29 & -0.39 & 1.00 & 0.00 & 1.00 & 0.00 & 2.39 & -0.22 \\
$\omega_1$ Upper & 1.00 & -0.01 & 1.12 & -0.02 & 2.23 & 0.19 & 1.00 & 0.00 & 1.00 & 0.00 & 2.90 & 0.11 \\
$\omega_2$ Lower & 1.00 & 0.03 & 1.00 & 0.03 & 1.47 & -0.16 & 1.00 & 0.00 & 1.00 & 0.00 & 2.40 & -0.09 \\
$\omega_2$ Upper & 1.00 & 0.00 & 1.00 & 0.00 & 2.73 & 0.39 & 1.00 & 0.00 & 1.00 & 0.00 & 2.91 & 0.22 \\
\midrule
\multicolumn{1}{l}{$\ell$} & \multicolumn{12}{c}{\emph{Row-collapse, $\tau$ lower}} \\
\midrule
0 & 1.00 & 0.01 & 2.21 & 0.20 & 1.70 & -0.26 & 1.00 & 0.00 & 3.72 & 0.19 & 1.76 & -0.15 \\
0.01 & 1.00 & 0.01 & 2.01 & 0.16 & 1.72 & -0.26 & 1.00 & 0.00 & 3.00 & 0.12 & 1.76 & -0.15 \\
0.05 & 1.00 & 0.01 & 1.41 & 0.05 & 1.76 & -0.26 & 1.00 & 0.00 & 1.22 & 0.01 & 1.78 & -0.15 \\
0.1 & 1.00 & 0.00 & 1.05 & 0.01 & 1.79 & -0.26 & 1.00 & 0.00 & 1.01 & 0.00 & 1.80 & -0.15 \\
0.25 & 1.00 & 0.00 & 1.00 & 0.00 & 1.82 & -0.26 & 1.00 & 0.00 & 1.00 & 0.00 & 1.83 & -0.15 \\
0.5 & 1.00 & 0.00 & 1.00 & 0.00 & 1.83 & -0.26 & 1.00 & 0.00 & 1.00 & 0.00 & 1.86 & -0.15 \\
\midrule
\multicolumn{1}{l}{$\ell$} & \multicolumn{12}{c}{\emph{Centroid-collapse, $\tau$ lower}} \\
\midrule
0 & 1.00 & 0.03 & 1.15 & 0.06 & 3.24 & -0.37 & 1.00 & 0.02 & 1.16 & 0.03 & 3.29 & -0.21 \\
0.01 & 1.00 & 0.02 & 1.15 & 0.05 & 2.92 & -0.33 & 1.00 & 0.01 & 1.16 & 0.03 & 2.72 & -0.17 \\
0.05 & 1.00 & 0.00 & 1.10 & 0.03 & 2.06 & -0.24 & 1.00 & 0.00 & 1.05 & 0.01 & 1.80 & -0.12 \\
0.1 & 1.00 & -0.01 & 1.02 & 0.01 & 1.75 & -0.22 & 1.02 & 0.00 & 1.00 & 0.00 & 1.84 & -0.13 \\
0.25 & 1.03 & 0.00 & 1.00 & 0.00 & 1.88 & -0.25 & 1.01 & 0.01 & 1.00 & 0.01 & 1.71 & -0.13 \\
0.5 & 1.00 & 0.01 & 1.00 & 0.01 & 1.74 & -0.25 & 1.00 & 0.01 & 1.00 & 0.01 & 1.63 & -0.13 \\
\bottomrule
\end{tabular}
\begin{minipage}{0.98\linewidth}
\vspace{0.35em}
\footnotesize \emph{Notes:} 5{,}000 MC replications. RRMSE divides the estimator's root-MSE by the smallest RMSE in the same design and $n$. Bias is the mean error divided by the length of $p'\Theta_I$ for the original $\Theta_I$. For scale, in the centroid-collapse the SD of Penalty, divided by the original $\tau$-interval length, is $0.177$ at $\ell=1$ and $0.117$ at $\ell=0$ when $n=3176$, and $0.095$ and $0.064$ when $n=10000$, resp. $\mathcal{Q}_{\tilde{\alpha}}$ uses 50{,}000 draws of $\sup_{\theta\in\Theta_I}n \hat{Q}(\theta)$. 
\end{minipage}
\end{table}

\subsection{Inference}
We now study the performance of inference procedures producing one-sided $95\%$ confidence intervals for $p'\theta$ for $\theta \in \Theta_I$. SS\ denotes the sample-splitting procedure in Section \ref{section_inference}, while FS\ denotes the full-sample basis-screening procedure in Section \ref{sec:uniform_inference}. Cov. denotes simulated coverage, while E.L. denotes simulated median excess length relative to the identified set's length in a given direction for the original Holmes polytope. 

\paragraph*{Benchmarks.}
We compare the proposed procedures with the pointwise and uniform procedures of \citet{gafarov2024simple}, denoted by BG PW.\ and BG Unif. respectively, and with the test-inversion procedures of \citet{andrews2010inference} (AS) and \citet{kaido2019confidence} (KMS). BG procedures were implemented using the code kindly provided by the author\footnote{The simulations in \citet{gafarov2024simple} were carried out for $\theta$ lying in a $[-2,2]$ box. Since the magnitude of $\lVert \theta \rVert$ is central to the choice of tuning parameters for BG procedures, we transformed $\theta$ linearly to align $\Theta$ with $[-2,2]^3$.}. AS and KMS were implemented by adapting the E-A-M implementation of \citet{kaido2017calibrated} to non-separable moments\footnote{As of 09.03.2026, the original package downloaded from \url{https://github.com/MatthewThirkettle/calibrated-projection-MATLAB} is only applicable to known $M$-setups. Apart from augmenting it, we also developed a different Python implementation with alternative kriging and minimization routines. It allowed us to obtain alternative E-A-M candidates for the same samples, passing the same statistical test. The $\dagger$ signs in Table \ref{tab:penalty-bg-kms-original-n3176-n10000-fullsample} were obtained by comparing the Matlab endpoints to the seedwise lower envelope of the two implementations.}. We use a $50-50$ split for SS. The FS\ candidate set was computed using the vertex-traversal procedure in Supplemental Appendix \ref{app: computation_unif}. We set $\alpha_{f,n} = 0.05$ at $n = 3176$, and $\alpha_{f,n}\propto(\log n)^{-1}$. The GMS sequence is ${k}_n = \sqrt{2 \log \log n}$.

\paragraph*{Results.}
Table \ref{tab:penalty-bg-kms-original-n3176-n10000-fullsample} reports simulated coverage and median excess length for the original polytope and irregular designs. We first discuss the original polytope. Despite the mechanical $\sqrt{2}$ s.e. cost of sample splitting, SS\ produces median excess length comparable to that of KMS and always smaller than that of AS. Its coverage is close to nominal at the empirical sample size and reaches the nominal level at $n=10000$. FS\ controls size at both sample sizes, and tends to have smaller excess length than SS, KMS and AS. Both BG procedures appear to be rather unstable in our designs. A possible explanation for this is poor geometry of the population polytope. After row-norm normalization, the uniform conditioning constant ($\eta(P)$ in \citet{gafarov2024simple}) in our baseline design is only $0.0034$, relative to $1$ in BG's simulations, while the measure of proximity to LICQ failure ($s(P)$ in \citet{gafarov2024simple}) is $0.01$, relative to $1.414$ in BG's simulations.

We now discuss the results in the irregular setups. We omit BG procedures, which are not applicable because the sample LP is infeasible in more than $10\%$ of the simulations already for $\ell = 0.5$. SS\ remains stable along both paths: coverage ranges from $0.91$ to $0.96$, and the excess length changes little as the identified set collapses. The decline in coverage at intermediate values of $\ell$ disappears as the design approaches point identification, while coverage does not deteriorate under row-collapse. FS\ maintains nominal coverage along the row-collapse path and is less conservative than AS and KMS in terms of excess length. The ranking in terms of excess length changes near point identification along the centroid-collapse path, where AS and KMS have smaller E.L. and FS\ becomes conservative in coverage. This coincides with a substantial increase in the number of retained bases: the average size of $\mathcal C_n$ rises from $2.47$ at $\ell=1$ to $13.39$ at $\ell=0$. Consequently, the difference between FS\ and SS\ E.L. narrows. 

\begin{table}
\centering
\caption{Coverage and median excess length of one-sided $95\%$ confidence intervals}
\footnotesize
\label{tab:penalty-bg-kms-original-n3176-n10000-fullsample}
\begin{tabular}{l*{12}{c}}
\toprule
 & \multicolumn{2}{c}{SS}
 & \multicolumn{2}{c}{FS}
 & \multicolumn{2}{c}{KMS}
 & \multicolumn{2}{c}{AS}
 & \multicolumn{2}{c}{BG Unif.}
 & \multicolumn{2}{c}{BG PW.} \\
\cmidrule(lr){2-3}\cmidrule(lr){4-5}
\cmidrule(lr){6-7}\cmidrule(lr){8-9}
\cmidrule(lr){10-11}\cmidrule(lr){12-13}
 & Cov. & E.L. & Cov. & E.L.
 & Cov. & E.L. & Cov. & E.L.
 & Cov. & E.L. & Cov. & E.L. \\
\midrule
\multicolumn{1}{l}{Direction}
 & \multicolumn{12}{c}{\emph{Original polytope, $n=3176$}} \\
\midrule
$\tau$ Lower
& 0.93 & 0.39 & 0.97 & 0.31 & 0.98 & 0.32 & 1.00 & 0.49 & 0.95 & 0.26 & 0.95 & 0.25 \\
$\tau$ Upper
& 0.93 & 0.60 & 0.97 & 0.54 & 0.98 & 0.57 & 1.00 & 0.85 & 0.95 & 0.34 & 0.94 & 0.32 \\
$\omega_1$ Lower
& 0.93 & 0.43 & 0.95 & 0.31 & 0.99 & 0.55\textsuperscript{$\dagger$} & 1.00 & 0.66 & 0.91 & 0.61 & 0.00 & -0.44 \\
$\omega_1$ Upper
& 0.93 & 0.17 & 0.98 & 0.15 & 0.98 & 0.13 & 1.00 & 0.21 & 0.97 & 0.16 & 0.40 & -0.01 \\
$\omega_2$ Lower
& 0.93 & 0.16 & 0.96 & 0.11 & 0.99 & 0.13 & 1.00 & 0.19 & 0.97 & 0.15 & 0.61 & 0.01 \\
$\omega_2$ Upper
& 0.93 & 0.35 & 0.95 & 0.24 & 1.00 & 0.41 & 1.00 & 0.53 & 0.85 & 0.39 & 0.00 & -0.45 \\
\midrule
\multicolumn{1}{l}{Direction}
 & \multicolumn{12}{c}{\emph{Original polytope, $n=10000$}} \\
\midrule
$\tau$ Lower
& 0.96 & 0.23 & 0.96 & 0.16 & 0.96 & 0.16 & 1.00 & 0.27 & 0.95 & 0.15 & 0.95 & 0.15 \\
$\tau$ Upper
& 0.95 & 0.34 & 0.96 & 0.27 & 0.96 & 0.29 & 1.00 & 0.45 & 0.93 & 0.21 & 0.93 & 0.20 \\
$\omega_1$ Lower
& 0.95 & 0.24 & 0.95 & 0.17 & 0.98 & 0.24\textsuperscript{$\dagger$} & 1.00 & 0.39 & 0.86 & 0.19 & 0.00 & -0.38 \\
$\omega_1$ Upper
& 0.96 & 0.09 & 0.96 & 0.06 & 0.99 & 0.07 & 1.00 & 0.11 & 0.98 & 0.08 & 0.82 & 0.03 \\
$\omega_2$ Lower
& 0.96 & 0.08 & 0.96 & 0.06 & 0.99 & 0.07\textsuperscript{$\dagger$} & 1.00 & 0.10 & 0.98 & 0.08 & 0.88 & 0.04 \\
$\omega_2$ Upper
& 0.95 & 0.19 & 0.95 & 0.14 & 0.98 & 0.21 & 1.00 & 0.32 & 0.77 & 0.10 & 0.00 & -0.33 \\
\midrule
\multicolumn{1}{l}{$\ell$}
 & \multicolumn{12}{c}{\emph{Row-collapse, $\tau$ lower, $n=3176$}} \\
\midrule
0
& 0.94 & 0.43 & 0.96 & 0.31 & 0.98 & 0.37 & 1.00 & 0.55 & {} & {} & {} & {} \\
0.01
& 0.93 & 0.43 & 0.95 & 0.31 & 0.98 & 0.37 & 1.00 & 0.54 & {} & {} & {} & {} \\
0.05
& 0.93 & 0.42 & 0.95 & 0.30 & 0.98 & 0.36 & 1.00 & 0.53 & {} & {} & {} & {} \\
0.1
& 0.94 & 0.41 & 0.95 & 0.29 & 0.98 & 0.34 & 1.00 & 0.51 & {} & {} & {} & {} \\
0.25
& 0.93 & 0.39 & 0.96 & 0.28 & 0.98 & 0.32 & 1.00 & 0.49 & {} & {} & {} & {} \\
0.5
& 0.93 & 0.39 & 0.96 & 0.29 & 0.98 & 0.31 & 1.00 & 0.48 & {} & {} & {} & {} \\
\midrule
\multicolumn{1}{l}{$\ell$}
 & \multicolumn{12}{c}{\emph{Centroid-collapse, $\tau$ lower, $n=3176$}} \\
\midrule
0
& 0.96 & 0.35 & 1.00 & 0.35 & 0.99 & 0.23 & 1.00 & 0.31 & {} & {} & {} & {} \\
0.01
& 0.96 & 0.36 & 1.00 & 0.36 & 0.99 & 0.25 & 1.00 & 0.32 & {} & {} & {} & {} \\
0.05
& 0.96 & 0.39 & 1.00 & 0.38 & 0.99 & 0.29 & 1.00 & 0.38 & {} & {} & {} & {} \\
0.1
& 0.95 & 0.41 & 0.99 & 0.40 & 0.99 & 0.35 & 1.00 & 0.44 & {} & {} & {} & {} \\
0.25
& 0.91 & 0.38 & 0.98 & 0.40 & 0.99 & 0.40 & 1.00 & 0.54 & {} & {} & {} & {} \\
0.5
& 0.91 & 0.37 & 0.97 & 0.31 & 0.99 & 0.34 & 1.00 & 0.51 & {} & {} & {} & {} \\
\bottomrule
\end{tabular}
\begin{minipage}{0.98\linewidth}
\vspace{0.35em}
\footnotesize
\emph{Notes:}  \(^{\dagger}\) marks AS/KMS cells for which we verified that the reported E-A-M median E.L. understates that implied by the KMS test
by at least $5\%$. Results are based on 5{,}000 MC replications for SS,
FS, and BG and 500 MC replications for AS and KMS. AS and KMS use
499 bootstrap draws.
\end{minipage}
\end{table}
\begin{table}
\centering
\caption{Computation time in minutes ($t$) in higher dimensions}
\footnotesize
\label{tab:augmented-holmes-comparison}
\begin{tabular}{lccccccccccc}
\toprule
 & \multicolumn{3}{c}{SS} & \multicolumn{4}{c}{FS} & \multicolumn{4}{c}{KMS} \\
\cmidrule(lr){2-4}\cmidrule(lr){5-8}\cmidrule(lr){9-12}
$d$
& Cov. & E.L. & Mean $t$
& Cov. & E.L. & Mean $t$ & Max $t$
& Cov. & E.L. & Mean $t$ & Max $t$ \\
\midrule
\multicolumn{12}{c}{\emph{Positive $H$ -- unique minimizer}} \\
\midrule
3
& 0.93 & 0.39 & $9{\cdot}10^{-5}$
& 0.97 & 0.31 & $2{\cdot}10^{-4}$ & $4{\cdot}10^{-4}$
& 0.97 & 0.33 & 9.46 & 19.42 \\
5
& 0.93 & 0.39 & $9{\cdot}10^{-5}$
& 0.96 & 0.28 & $4{\cdot}10^{-4}$ & $6{\cdot}10^{-4}$
& 0.97 & 0.31\textsuperscript{$\dagger$} & 10.27 & 25.62 \\
10
& 0.93 & 0.39 & $2{\cdot}10^{-4}$
& 0.97 & 0.28 & $9{\cdot}10^{-4}$ & $2{\cdot}10^{-3}$
& 0.94 & 0.30\textsuperscript{$\dagger$} & 18.75 & 59.37 \\
12
& 0.93 & 0.39 & $2{\cdot}10^{-4}$
& 0.96 & 0.29 & $1{\cdot}10^{-3}$ & $3{\cdot}10^{-3}$
& 0.92 & 0.29\textsuperscript{$\dagger$} & 22.92 & 84.88 \\
15
& 0.93 & 0.40 & $3{\cdot}10^{-4}$
& 0.96 & 0.30 & $2{\cdot}10^{-3}$ & $8{\cdot}10^{-3}$
& \multicolumn{3}{c}{} & $>120$ \\
\midrule
\multicolumn{12}{c}{\emph{Zero $H$ -- flat face}} \\
\midrule
5
& 0.93 & 0.40 & $1{\cdot}10^{-4}$
& 0.97 & 0.35 & $5{\cdot}10^{-4}$ & $1{\cdot}10^{-3}$
& 0.97 & 0.33 & 12.17 & 21.95 \\
10
& 0.94 & 0.40 & $2{\cdot}10^{-4}$
& 0.97 & 0.30 & 0.02 & 0.03
& 0.97 & 0.33 & 19.74 & 30.50 \\
12
& 0.94 & 0.40 & $2{\cdot}10^{-4}$
& 0.97 & 0.29 & 0.08 & 0.12
& 0.97 & 0.33 & 22.15 & 30.79 \\
15
& 0.94 & 0.40 & $3{\cdot}10^{-4}$
& 0.97 & 0.28 & 0.67 & 0.95
& 0.97 & 0.33 & 23.03 & 37.06 \\
\bottomrule
\end{tabular}
\begin{minipage}{0.98\linewidth}
\vspace{0.35em}
\footnotesize
\emph{Notes:} \textsuperscript{$\dagger$} marks KMS cells for which we verified that the reported E-A-M median E.L. understates that implied by the
KMS test by $>5\%$. FS, KMS use $100$ common MC draws and $1001$ bootstrap draws; $5000$ MCs for SS. Simulations used a MacBook Pro with a 14-core Apple M3 Max chip and 36GB RAM.
\end{minipage} 
\end{table}

\paragraph*{Computational scaling.} We finally assess how the computational cost of the inference procedures changes as the dimension of the LP grows. To do so while preserving the empirically calibrated bound and sampling variation, we augment the parameter vector with \(d_1\) nuisance coordinates \(\eta\) and replace the original system by $M\theta -H\eta\geq c, ~ \eta\in[0,1]^{d_1}.$ We retain the covariance matrix \(\Sigma\) for the estimated \(M,c\) coefficients, while \(H\in\mathbb R^{q\times d_1}\) and the added box constraints are non-random. We consider two polar cases: (i) \(H_{ij}-1/2\sim\mathrm{Bernoulli}(1/2)\) i.i.d., drawn once and fixed across simulations, and (ii) \(H=0\). In either case, the projection of \(\Theta_I\) onto the original coordinates, and hence the sharp lower bound on \(\tau\), is unchanged. Under (i), \(\eta=0\) at the unique optimum, so no additional optimal bases are generated. This is favorable to FS: rather than search globally over bases, one may `look around' the penalty-selected basis $\hat{A}$, as explained in Supplemental Appendix \ref{app: computation_unif}. Its computational cost then depends on the number of bases whose penalty objective \(L(\hat{\theta}(A);\hat{\mu}_n,w_n)\) lies close to the penalty optimum, rather than on the total number of bases, which grows combinatorially. Case (ii) is the opposite extreme: \(\mathcal{S}(\mu_0)=\{\theta^*\}\times[0,1]^{d_1}\) is a flat face, generating \(2^{d_1}\) optimal bases with the same value \(p'\theta^*\). It provides an adversarial computational design for FS. In the interest of space, we compare our methods only with KMS, which appears to be the relevant benchmark based on Table \ref{tab:penalty-bg-kms-original-n3176-n10000-fullsample}. For a direct statistical and timing comparison, FS and KMS use 100 common MC draws; SS uses 5,000.

Table \ref{tab:augmented-holmes-comparison} shows that the coverage and excess length of SS and FS remain stable as \(d\) grows, while their computational advantage over KMS is substantial. SS is computationally trivial even for $d = 15, q = 42$ in either design: it amounts to solving one penalty LP and one quadratic representer problem, both of which have polynomial complexity in \(q,d\). The cost of our FS implementation is `output-sensitive.' The cost per visited basis is polynomial in \(q,d\), but their number depends on how many bases have penalty objective close to \(V(\hat\mu_n; w_n)\). Thus, its cost grows slowly with dimension when the number of such bases remains manageable, as in the positive-\(H\) design, but may grow exponentially near a flat face; under \(H=0\), the sublevel set contains \(2^{d_1}\) optimal bases. Even in this deliberately adverse design each FS replication finished in under a minute at \(d=15\) and \(q=42\). 

KMS becomes substantially more costly as \(d\) grows and, in the unique-minimizer design, its E-A-M approximated endpoints exhibit optimization error. By construction\footnote{Because no noise is added to the system, we use the recommended KMS \(\rho\) for \(d=3\) and \(q=12\) for all \(d\).}, in either design the exact KMS confidence regions are nested across \(d\) within each sample and therefore cannot become narrower as nuisance coordinates are added. Under positive \(H\), the reported endpoints instead move materially inward and coverage deteriorates with \(d\). Because we use the same random draws across dimensions, this violation of nesting is directly attributable to premature convergence of E-A-M.\footnote{Following \citet{andrews2023}, we attempted tightening E-A-M convergence criteria and adding evaluation points, but were similarly unable to remove the pattern. Our baseline implementation choices are those recommended by KMS, apart from \texttt{maxit} that we increase to \(2000\) to avoid truncation.} \citet{andrews2023} document a similar violation of a theoretical nesting relation and attribute it to E-A-M optimization error. Premature convergence produces both an inaccurately narrow interval and an artificially short runtime. Mean runtimes may therefore understate the computational cost of KMS and should be interpreted jointly with the maximum runtimes. At \(d=15\), at least one of the 100 positive-\(H\) replications exceeded the two-hour cutoff.

\paragraph*{Summary.} Overall, the simulations demonstrate strong statistical and computational performance of both SS and FS procedures. Despite having only pointwise coverage guarantees, SS remains stable across the considered irregular geometries and is computationally trivial, while FS delivers uniform coverage and tends to have the smallest excess length except when all population feasible bases become nearly indistinguishable in value. Both procedures remain applicable where existing LP-based methods can fail; relative to general test-inversion benchmarks, their E.L. is often smaller, while computation is orders of magnitude faster and retains its numerical accuracy in higher dimensions.

\appendix
\section{Implementation details}\label{app_imp}
\paragraph*{Details for Section \ref{section_debiased}} From the proof of Theorem \ref{root_n_pointwise}, the following hold w.p.a.1:
\begin{align}\label{eq: lp_reform}
    V(\hat{\mu}_n; w_n) &= \min_{(\theta, a) \in \overline{\Theta} \times \mathbb{R}^q} p'\theta + w_n\iota'a, \text{ s.t.: } a\geq 0, a \geq \hat{c}_n - \hat{M}_n \theta,\\ \notag
    \hat{V}(\hat{\mu}_n;w_n) &= \max_{(\theta, a) \in \overline{\Theta} \times \mathbb{R}^q} p'\theta \text{ s.t.: } a\geq 0, a \geq \hat{c}_n - \hat{M}_n \theta, \text{ and }p'\theta + w_n \iota'a \leq V(\hat{\mu}_n; w_n),
\end{align}
where $\overline{\Theta}$ is $\mathbb{R}^d$ or $\Theta$ if the latter is a polytope. 
\paragraph*{Details for Section \ref{section_inference}} One may obtain $\hat r$ by minimizing the estimated MSE bound:
\begin{align}\label{reg_v}
\hat r
\in
\arg\min_{r\in\mathbb R^{|\hat A|}}
\left\{
        r'\Omega_+(\hat A, \hat{\theta}, \hat{\Sigma}^{(1)})r
        +
        n_2 \max_{\theta\in\Theta_+}
\left(
(p-\hat M^{(1)\prime}_{\hat A}r)'(\hat\theta-\theta)
\right)^2 \right\},
\end{align}
where $\Omega_+$ floors the eigenvalues of $\Omega$ at a fixed $\underline\sigma^2>0$, and $\Theta_+$ is a fixed compact set with $\Theta\subset\mathrm{Int}(\Theta_+)$. Lemma \ref{lem:representer_rate} establishes validity of this choice. The first term estimates the leading variance in \eqref{motivation}, while the second upper bounds the squared stationarity bias. 

The regularization $\Omega_+$ ensures that $\hat{r}$ remains $O_p(1)$ even if the asymptotic residual covariance matrix is singular: if Assumption 6 is strengthened to the requirement that $\Omega(A, \theta, \Sigma)$ be positive definite for any optimal triplet $(A, \theta, \cdot)$, Theorem \ref{th: main_inference} holds with $\Omega_+ = \Omega$. The regularization $\Theta_+$ ensures that the bias term vanishes even if $\hat{\theta}$ is on the boundary of $\Theta$: if it is assumed that $\mathcal{S}(\mu_0) \subseteq \text{Int}(\Theta)$, then $\Theta_+ = \Theta$ is similarly admissible. We use $\Omega_+ = \Omega$ and $\Theta_+ = \Theta$ in our simulations, but we recommend the regularized options in applications where nonstochastic inequalities may determine the LP optimum. 

\paragraph*{Details for Section \ref{sec:uniform_inference}} For a square basis $A$, if $\tilde{M}_A$ is invertible, let $g_{A'}(A,\mu) \equiv \tilde{c}_{A'} - \tilde{M}_{A'} \tilde{M}^{-1}_{A}\tilde{c}_A$, and $g_{A}(A, \mu) \equiv -(\tilde{M}_A^{-\prime}p)$. Its directional derivative in the direction $h = \text{vec}(h^M,h^c)$ is
        \begin{align}\label{first_order}
            D g_{A'}(A,\mu)[h] &= 
                h^c_{A'} - \tilde{M}_{A'}\tilde{M}^{-1}_{A} h^c_{A} - (h^M_{A'} -  \tilde{M}_{A'} \tilde{M}^{-1}_{A} h^M_{A} ) \tilde{M}^{-1}_{A}\tilde{c}_{A}, \\ \notag
            D g_A(A,\mu)[h] &=\tilde{M}^{-\prime}_{A}(h^M_{A})'\tilde{M}^{-\prime}_{A}p. 
        \end{align}
    
    The directional derivative in \eqref{first_order} allows to find the Jacobian of $\mu \to g(A, \mu)$, denoted by $D g(A,\mu)\in\mathbb R^{q\times m}$. For a nonsingular $\hat{M}_A$ and $j \in [q]$, define
    \begin{align*}
        \hat{s}^2_j(A) \equiv \max\{\hat{\Xi}(A)_{jj}, \underline{\sigma}^2\}, \quad  \hat{\Xi}(A) \equiv  D g(A,\hat\mu_n) \hat{\Sigma}_n D g(A,\hat\mu_n)'
    \end{align*}

Further define $\hat{\mathbb{G}}(A)|\mathcal{W}_n\sim N(0, \hat{\Gamma}(A))$, where
    \begin{align*}
        \hat{\Gamma}(A) \equiv  \hat{\Psi}^{-1}(A)\hat{\Xi}(A)\hat{\Psi}(A)^{-1}, \quad \hat{\Psi}(A) = \text{diag}((\hat{s}_j(A))_{j\in[q]}).
    \end{align*}
    For a nonempty set $B \subseteq [q]$, we define 
    \begin{align}\label{cv_def}
        \hat{q}_{1-\beta}(A, B) \equiv \max\{z_{1-\beta},\inf\{t \geq 0: \P[\max_{j \in B} \hat{\mathbb{G}}_j(A) \geq t|\mathcal{W}_n] \leq \beta\}\},
    \end{align}
    while $\hat{q}_{1-\beta}(A, \emptyset) \equiv 0$, and $\hat{q}_{1-\beta}(A) \equiv \hat{q}_{1-\beta}(A,\hat{B}(A))$. One can compute $\hat{q}_{1-\alpha_{f,n}}(A)$ analytically from the CDF of $\max_{j \in \hat{B}(A)}\hat{\mathbb{G}}_j(A)$, treating $\mathcal{W}_n$ as fixed.

\section{Estimation}\label{appA}
\subsection{Proof of Lemma \ref{lemma_penalty}}
\begin{proof}
    The inequality is proven in the main text. Under Assumptions 1, 2, fix $\lambda^* \in \Lambda(\mu_0)$ s.t. $w > \lambda^*$. Let $\pi \equiv c - M\theta$. By LP duality, $c'\lambda^* = V(\mu_0)$, and $p = M'\lambda^*$, so
    \begin{align}\label{eq: pen_gap}
        L(\theta;\mu_0, w) - V(\mu_0) = w'\pi^+ - \lambda^{*\prime}\pi \geq (w - \lambda^*)'\pi^+ \geq 0.
    \end{align}
    Taking $\min_{\theta \in \Theta}$ on both sides of \eqref{eq: pen_gap} gives the remaining direction of i). Equality in \eqref{eq: pen_gap} forces $\pi \leq  0$, establishing ii). This concludes the proof of the lemma. 
    \end{proof}

\subsection{Proof of Proposition \ref{first_penalty_theorem}}
\begin{proof}
For any $x_1, x_2 \in \mathbb{R}^q,$ $\lVert x^+_1 - x^+_2\rVert \leq \lVert x_1 - x_2\rVert.$ Using that, triangle and Cauchy-Schwarz inequalities, one establishes that $L(\theta, \cdot, w_n)$ is uniformly $w_n\sqrt{q(1 + C^2_\Theta)}$-Lipschitz over $\theta \in \Theta$, where $C_\Theta = \max_{\theta \in \Theta}\lVert \theta \rVert$. Because for $g_1, g_2 \in \mathcal{C}(\Theta)$, $|\inf_\Theta g_1 - \inf_\Theta g_2| \leq \sup_{\Theta} |g_1 - g_2|$, the same holds for $V(\cdot, w_n)$. By Assumption 1, $V(\hat{\mu}_n; w_n) - V(\mu_0; w_n) = O_p(w_n/\sqrt{n}).$ Fix $\lambda^* \in \arg \min_{\lambda \in \Lambda(\mu_0)} \lVert \lambda \rVert_\infty$. If $w_n > \lVert \lambda^*\rVert_\infty < \infty$, then $V(\mu_0; w_n) = V(\mu_0)$ by Lemma \ref{lemma_penalty}. This occurs w.p.a.1 as $w_n \to \infty$ w.p.a.1. Combine with the above rate to conclude the proof of the Proposition.
\end{proof}
\subsection{Proof of Theorem \ref{root_n_pointwise}}\label{ap_proof_th22}
We treat case i) below. Extension to ii) is straightforward and is given in Supplemental Appendix \ref{app:polytope_theta}. We will use the lemma that shows $\kappa(\tilde{\Theta})^{-1}$ to be a Hoffman constant\footnote{\citet{hoffman1952approximate} proved that $\exists H < \infty$, such that $H \geq \sup_{\theta \notin \tilde{\Theta}} \text{d}(\theta, \tilde{\Theta})/\lVert (\tilde{c} - \tilde M \theta)^+\rVert_1$, see e.g. \citet{li1993sharp} for characterizations of sharp Hoffman constants. $\kappa(\cdot)^{-1}$ need not be sharp, but it additionally guarantees the well-conditioned basic optimal solution used in Theorem \ref{thm:full_sample_uniform_inference} and admits an intuitive interpretation.}:
\begin{lemma}\label{lemma_polytope_anticoncentration}
    For any non-empty, bounded polytope $\tilde{\Theta} = \{\theta \in \mathbb{R}^d|\tilde{M}\theta \geq \tilde{c}\}$ and $\theta \in \mathbb{R}^d$,
    \begin{align*}
        \kappa(\tilde{\Theta})^{-1}\iota'(\tilde{c} - \tilde{M}\theta)^+  \geq \mathrm{d}(\theta, \tilde{\Theta}) \equiv \min_{\tilde{\theta} \in \tilde{\Theta}}\lVert \theta - \tilde{\theta}\rVert.
    \end{align*}
\end{lemma}
Lemma \ref{lemma_polytope_anticoncentration} is proven in the Supplemental Appendix. We now prove the theorem. 
\begin{proof}
    Fix $\mu_0=\text{vec}(M,c)$. We proceed in four steps.
 \paragraph*{Step 1} We first prove that $\exists$ $\{{k}_n\} \downarrow 0$ such that $\mathcal{S}(\hat{\mu}_n; w_n) \subseteq \mathcal{S}(\mu_0)^{{k}_n},$ w.p.a.1, where $A^\varepsilon$ is the open $\varepsilon-$expansion of $A$. Let $\theta_n \in \arg \max_{\theta \in \mathcal{S}(\hat{\mu}_n; w_n)} \text{d}(\theta,\mathcal{S}(\mu_0))$ be a measurable selection. Berge's maximum Theorem implies that $(\mu,w)\mapsto\mathcal S(\mu; w)$ is uhc, so Lemmata 17.4, 18.2 in \citet{aliprantis2007infinite} establish weak-measurability of $\mathcal{S}(\hat{\mu}_n; w_n)$; the existence of $\theta_n$ then follows from Theorem 18.19 in the above text. By Proposition \ref{first_penalty_theorem},
        \begin{align}\label{proof_pen_conv}
            p'\theta_n + w_n \iota'(\hat{c}_n - \hat{M}_n \theta_n)^+ - V(\mu_0) = o_p(1).
        \end{align}
        Furthermore, since $w_n = o_p(\sqrt{n})$, by Assumption 1, we have
        \begin{align}\label{bound_in_prob}
            w_n \sup_{\theta \in \Theta}\lVert \hat{c}_n - \hat{M}_n \theta - c + M \theta\rVert_{\infty} = o_p(1)
        \end{align}
        Because $\mathcal{S}(\hat{\mu}_n; w_n) \subseteq \Theta$, the first term in \eqref{proof_pen_conv} is $O_p(1)$. Thus, from \eqref{proof_pen_conv}, $w_n \iota' (\hat{c}_n - \hat{M}_n \theta_n)^+ = O_p(1)$. By triangle inequality and \eqref{bound_in_prob}, we therefore conclude $w_n \iota'(c - M \theta_n)^+ = O_p(1)$. As $w_n \to \infty$ w.p.a.1, it further follows that $(c - M \theta_n)^+ = o_p(1)$, so, for some ${k}_n \downarrow 0$, $c - M \theta_n \leq \iota {k}_n$ w.p.a.1. Conclude by applying Lemma \ref{lemma_polytope_anticoncentration} to $\mathcal{S}(\mu_0) = \Theta_I \cap \{\theta: p'\theta \leq V(\mu_0)\}$ and using \eqref{proof_pen_conv}.
    \paragraph*{Step 2} We now establish the existence of a vertex-solution w.p.a.1. The penalty problem is w.p.a.1. equivalent to the LP \eqref{eq: lp_reform}, with $\overline{\Theta} = \mathbb{R}^d$ by Step 1,  $\mathcal{S}(\mu_0)\subseteq \text{Int}(\Theta)$ and convexity of $L(\cdot)$. As $\mathcal{S}(\hat{\mu}_n; w_n)$ is compact, \eqref{eq: lp_reform} has a non-empty and bounded solution set w.p.a.1. Therefore, \eqref{eq: lp_reform} has a solution $(a, \theta)$ determined by a full-column-rank ($d + q$) system of binding inequalities w.p.a.1. Pick the largest $A_1, A_2 \subseteq [q]$, such that the latter has the form 
    \begin{align*}
        \begin{cases}
            a_{A_1} = 0\\
            a_{A_2} = \hat{c}_{A_2} - \hat{M}_{A_2}\theta
        \end{cases} & \iff  \begin{cases}
            a_{A_1} = 0\\
            a_{A^*} = \hat{c}_{A^*} - \hat{M}_{A^*}\theta\\
            a_{\hat{A}} = \hat{c}_{\hat{A}} - \hat{M}_{\hat{A}}\theta 
        \end{cases} \iff \begin{cases}
            a_{A_1} = 0\\
            a_{A^*} = \hat{c}_{A^*} - \hat{M}_{A^*}\theta\\
            \hat{M}_{\hat{A}}\theta = \hat{c}_{\hat{A}}   
        \end{cases}\\
        &\iff \begin{pmatrix}
            I_{|A_1|} & 0 & 0\\
            0 & I_{|A^*|} & \hat{M}_{A^*}\\
            0 & 0 &  \hat{M}_{\hat{A}} 
        \end{pmatrix}\begin{pmatrix}
            a_{A_1}\\
            a_{A^*}\\
            \theta
        \end{pmatrix} = \begin{pmatrix}
            0 \\
            \hat{c}_{A^*}\\
            \hat{c}_{\hat{A}}
        \end{pmatrix},
    \end{align*}
    where $A^* \equiv \{j \in [q]: \hat{M}_{j} \theta < \hat{c}_{j}\}$ and $\hat{A} \equiv A_2 \setminus A^*$, $\iff$ means equivalence up to row operations, the first transition is by definition, the second follows by subtraction and $\hat{A} \subseteq A_1$, as $a \geq 0$ from \eqref{eq: lp_reform}, and the last uses $A_1 \cap A^* = \emptyset$. The rank of the above block matrix is $d + q = |A_1| + |A^*| + \text{rk}(\hat{M}_{\hat{A}})$, so $\text{rk}(\hat{M}_{\hat{A}}) = d$, as $|A_1| + |A^*|  \leq q$ by $A_1 \cap A^* = \emptyset$. 

    \paragraph*{Step 3} Let the (finite) set of all vertex-solutions be $\mathcal{V}^*_n \subseteq \mathcal{S}(\hat{\mu}_n; w_n)$. Its weak-measurability follows by the same arguments as used in Step 1. We will implicitly work on the event $\{\mathcal{V}^*_n \ne \emptyset\}$. Select any measurable $\theta^*_n \in \mathcal{V}^*_n$, and let $\hat{A} \equiv J(\theta_n^*;\hat\mu_n)$. Consider collections $\mathcal{E} \equiv \{A \subseteq [q]: M_A \theta \ne c_{A} ~ \forall \theta \in \mathcal{S}(\mu_0)\}$ and $\mathcal{F} \equiv \{A \subseteq [q]: p \notin \mathcal{R}(M'_A)\}$, and the events  $E_n \equiv \{\hat{A} \in \mathcal{E}\}$ and $F_n \equiv \{\hat{A} \in \mathcal{F}\}$. We wish to show that $\mathbb{P}[E'_n \cap F'_n] \to 1$.

First, consider $E_n$. The case $\mathcal E=\emptyset$ is trivial. Otherwise, compactness of $\mathcal S(\mu_0)$ and finiteness of $\mathcal E$ give $\varepsilon\equiv\min_{A\in\mathcal E}
\inf_{\theta\in\mathcal S(\mu_0)}\|M_A\theta-c_A\|>0$.
Let $\tilde\theta_n$ be the projection of $\theta_n^*$ onto
$\mathcal S(\mu_0)$. By Step 1, Assumption 1, and $\hat{M}_{\hat{A}}\theta^*_n = \hat{c}_{\hat A}$,
\[
\|M_{\hat A}\tilde\theta_n-c_{\hat A}\|
\le \|M\|\|\tilde\theta_n-\theta_n^*\|
+C_\Theta\|\hat M_n-M\|+\|\hat c_n-c\|=o_p(1),
\]
On $E_n$, the left-hand side is at least $\varepsilon$,
so $\mathbb P[E_n]\to0$.

Secondly, consider $F_n$. The KKT conditions of \eqref{eq: lp_reform} are $a \geq 0, \omega \geq 0, \tau \geq 0,$ and
    \begin{align}\label{pos_lin_hul}
         p = \hat{M}_{n}'\tau, ~w_n\iota = \omega + \tau, ~ 
        \omega'a = 0, ~ 
        \tau'(\hat{c}_n - \hat{M}_n \theta - a) = 0, 
        a \geq \hat{c}_n - \hat{M}_n \theta.
    \end{align}
    From \eqref{pos_lin_hul}, if $j \in A^* \equiv \{j: \hat{M}_{j}\theta^*_n - \hat{c}_{j} < 0\}$, then $a_j > 0$, and so $\omega_j = 0$, implying $\tau_j = w_n$. If $\hat{M}_{j} \theta_n^* -  \hat{c}_{j} > 0$, then $\hat{c}_{j} - \hat{M}_{j} \theta_n^{*} - a_j < 0$, and so $\tau_j = 0$. Finally, if $j \in \hat{A}$, then $\tau_j \in [0; w_n]$. Thus, by Assumption 1, from \eqref{pos_lin_hul},
    \begin{align}\label{eq_inter}
        p = w_n (\hat{M}_{A^*})'\iota + (\hat{M}_{\hat A})'\tau_{\hat A} = w_n (M_{A^*})'\iota + (M_{\hat A})'\tau_{\hat A} + O_p\left(w_n/\sqrt{n}\right)
    \end{align}
    For each $A\in\mathcal F$, fix $h(A) \in \mathbb{R}^d$ with $M_Ah(A)=0$ and $p'h(A)=-1$. These exist by FTLA. Multiply \eqref{eq_inter} by $h(\hat A)$ to get $w_n\iota'M_{A^*}h(\hat A)= -1 + O_p(w_n/\sqrt n)$ on $F_n$. This fails w.p.a.1, as $\iota'M_{A^*}h(\hat A)$ takes on finitely many values, and $w_n \to \infty$ w.p.a.1. Thus, $\P[F_n]\to0$.

\paragraph*{Step 4} Let $\mathbb A \equiv \{A\subseteq[q]:|A|\ge d,
\ A\text{ satisfies \eqref{optimality}, \eqref{range_req}} \text{ with } \theta^* \in \mathcal{S}(\mu_0)\}$ and $G_n\equiv\{\hat A\in\mathbb A\}$. Step 3 gives $\mathbb{P}[G_n] \to 1$. To each $A \in \mathbb{A}$, associate $r(A)$ from \eqref{range_req}. Then, $V(\mu_0) =  r(A)'c_A$ and $p = M_A'r(A)$. As $\hat{M}_{\hat{A}}\theta^*_n = \hat{c}_{\hat{A}}$, $\P[G_n] \to 1$, and $\max_{A \in \mathbb{A}} \lVert r(A) \rVert < \infty$, 
    \begin{align}\label{deviations_point}
        p'\theta^*_n - V(\mu_0) &=  r(\hat{A})'(\hat{c}_{\hat{A}} - c_{\hat{A}} + (M_{\hat{A}} - \hat{M}_{\hat{A}})\theta^*_n) = O_p(n^{-1/2}).
    \end{align}
Because $|\mathcal{V}^*_{n}| \leq C_{d,q}$ for a constant $C_{d,q} > 0$ that only depends on $d,q$, \eqref{deviations_point} strengthens to $\max_{\theta \in \mathcal{V}^*_{n}} |p'\theta - V(\mu_0)| = O_p(n^{-1/2}).$ To conclude the proof of the Theorem, note that any LP solution $\theta \in \mathcal{S}(\hat{\mu}_n; w_n)$ is a convex combination of vertex-solutions $\theta = \sum_{j} \omega_j \theta_j$, where $\{\theta_j\}^K_{j = 1} = \mathcal{V}^*_{n}$, $\omega_j \in [0;1]$, and $\sum_j \omega_j = 1$, and apply Hölder's inequality.
    \end{proof}
    \subsection{Proof of Theorem \ref{uniform_penalty_theorem}}\label{ap: pf_unif_pen_th}
\begin{proof}
Write $\nu(\theta) \equiv \iota'(c-M\theta)^+$,
$\hat\nu_n(\theta) \equiv \iota'(\hat c_n-\hat M_n\theta)^+$, and
\[
D_n \equiv \sup_{\theta\in\Theta}|\hat\nu_n(\theta)-\nu(\theta)|
\le\sqrt{q(1+C_\Theta^2)}\|\hat\mu_n-\mu_0\|
=O_{\mathcal P}(n^{-1/2}).
\]
For any $\hat\theta\in\mathcal S(\hat\mu_n;w_n)$ and
$\theta^*\in\mathcal S(\mu_0)$, sample optimality gives
\begin{align}\label{eq: one_sided_unifbound}
p'\hat\theta+w_n\hat\nu_n(\hat\theta)
\le p'\theta^*+w_n\hat\nu_n(\theta^*)
\le V(\mu_0)+w_nD_n.
\end{align}
This proves the upper bound, both for the penalized value and for
$\hat V$, uniformly over $\mathcal P$.

Fix $\delta>0$ and $\P\in\mathcal P^\delta$. By Lemma
\ref{lemma_polytope_anticoncentration},
\begin{align}\label{eq: bound_by_resid}
p'\hat\theta-V(\mu_0)
\ge-\|p\|\,\mathrm d(\hat\theta,\Theta_I)
\ge-\frac{\|p\|}{\delta}\nu(\hat\theta).
\end{align}
Adding $w_n \nu(\hat{\theta})$ to both sides of \eqref{eq: bound_by_resid}, using
$\hat\nu_n\ge\nu-D_n$ and \eqref{eq: one_sided_unifbound} gives
\begin{align}\label{eq: bound_on_resid}
\left(w_n-\frac{\|p\|}{\delta}\right)\nu(\hat\theta)
\le2w_nD_n.
\end{align}
The event $\{w_n\ge2\|p\|/\delta\}$ occurs w.p.a.1 uniformly over $\mathcal{P}^\delta$. On it, from \eqref{eq: bound_by_resid} and \eqref{eq: bound_on_resid},  
\begin{align*}
    p'\hat\theta-V(\mu_0)\ge-\frac{4\|p\|}{\delta}D_n.
\end{align*}
The lower bound follows for $\hat V$ and,
since the penalty is nonnegative, for $V(\hat{\mu}_n; w_n)$.
\end{proof}
\section{Supplemental Appendix}

\subsection{Penalty parameter selection}\label{ap_pensel}
For a fixed scalar $w$, $V(\hat{\mu}_n; w)$ converges to $V(\mu_0; w) \leq V(\mu_0)$, which need not equal $V(\mu_0)$, while $\hat{V} \leq V(\hat{\mu}_n; w)$ by definition. Therefore, $\hat{V}$ may have a non-zero conservative asymptotic bias when $w$ is so small that it fails to satisfy Assumption 2. However, if $w$ is too large, $V(\hat{\mu}_n; w)$ may fail to regularize the potentially unstable sample LP by enforcing the sample restrictions too literally, see Section \ref{section_consistency}. Intuitively, the choice of $w$ involves a bias-variance tradeoff: one would typically expect larger $w$ choices to produce less conservative estimates, at the cost of potentially larger variance of $\hat{V}$. 

Consider now $w \in \mathbb{R}^q_+$. The above discussion suggests that it may be desirable to set $w$ slightly larger than some $\lambda \in \Lambda(\mu_0)$, e.g. that with the smallest $\lVert \lambda \rVert_\infty$. Unfortunately, it is difficult to estimate the potentially unbounded dual solution set, which is not, in general, continuous w.r.t $\mu_0$ perturbations. Naïve choices such as $w = \hat{\lambda}_n$ tend to reproduce the unstable sample LP behavior, as in Example \ref{example_incons} in Section \ref{section_consistency}. Our calibration is instead developed by considering the basic dual solution $\lambda^* \in \Lambda(\mu_0)$ supported on a square $J^*$, with $\lambda^*_{J^*} = M^{-\prime}_{J^*} p$. Proposition \ref{theor_j_star} guarantees its existence. For it,
\begin{align}\label{lambda_star_bound}
    \lVert \lambda^*\rVert_\infty \leq \frac{\lVert p\rVert}{\mathrm{sv}_d(M_{J^*})},
\end{align}
where $\mathrm{sv}_d(M_{J^*}) \geq \kappa(\Theta_I)$ is related to the conditioning of the polytope. Unfortunately, the RHS of \eqref{lambda_star_bound} depends on the unknown $J^*$ that may be difficult to estimate in line with Example \ref{example_incons}. We thus propose to calibrate the scale of the RHS in \eqref{lambda_star_bound}. Note, however, that rescaling each inequality by a positive number changes the dual solutions, $\mathrm{sv}_d(M_{J^*})$ and $\kappa(\Theta_I)$,  even though it does not change $\Theta_I$ or the collection of basic optimal index sets. Unequal row scales may also require a large common penalty parameter to enforce $V(\mu_0; w) = V(\mu_0)$. This suggests normalizing the rows first. Suppose that the rows of $M$ are nonzero and define
\[
D=d^{-1/2}\operatorname{diag}(\|M_1\|,\ldots,\|M_q\|),\qquad
M^{\mathrm N}=D^{-1}M,\qquad c^{\mathrm N}=D^{-1}c.
\]
This preserves $\Theta_I=\{M^N\theta \geq c^N\}$, and gives each row of $M^{\mathrm N}$ norm $\sqrt d$. We have
\[
\lambda^*_j\le
\frac{\sqrt d\,\|p\|}
{\|M_j\|\,\mathrm{sv}_d(M^{\mathrm N}_{J^*})},\quad j\in[q],
\]
Hence, the smallest singular value of a normalized basic optimal system gives a sufficient common penalty level for the normalized residuals.

 We propose a simple calibration for $\mathrm{sv}_d(M^{\mathrm N}_{J^*})$ in the spirit of the reference-distribution choice of $\rho$ in \citet[Section 2.4]{kaido2019confidence}. Random matrix theory supplies the reference distribution in our case. For a $d\times d$ matrix $\Xi_d$ with independent, potentially non-i.i.d., real entries of mean zero, variance one, and uniformly bounded sufficiently high moments, \citet[Theorem 6.7]{tao2010random} show that, as $d\to\infty$,
\[
\sqrt d\,\mathrm{sv}_d(\Xi_d)\xrightarrow{L}\aleph,\qquad
\P[\aleph\le t]=1-\exp(-t^2/2-t),\quad t\ge0.
\]
Since $\E\| (\Xi_d)_j\|^2=d=\|M_j^{\mathrm N}\|^2$, we use
$\delta_{\bar\alpha}/\sqrt d$ as a benchmark for
$\mathrm{sv}_d(M^{\mathrm N}_{J^*})$, where
$\delta_{\bar\alpha}=\sqrt{1-2\log(1-\bar\alpha)}-1$
is the $\bar\alpha$-quantile of $\aleph$. This suggests the rule, for $n \geq 2$,
\begin{equation}\label{eq:penalty_calibration}
(w_n)_j
=\frac{ d\,\|p\|}{\delta_{\bar\alpha}\|\hat M_j\|} \cdot \overline{w}_n,
\qquad \overline{w}_n=\frac{\log n}{\log100},\qquad j\in[q].
\end{equation}
In Section \ref{sec:uniform_inference}, we similarly recommend using the normalized systems for basis screening, i.e. replace $\mathrm{sv}_d(\hat{M}_A)>w^{-1}_{n}$ in the definition of $\hat{\mathcal{A}}$ with $\mathrm{sv}_d(\hat{M}_A^{\mathrm{N}}) > (\overline{w}_n d \lVert p \rVert/\delta_{\overline{\alpha}})^{-1}$.

The factor $\overline{w}_n$ is normalized to $1$ at $n = 100$ and is allowed to grow at $\log n$ to ensure population sharpness $V(\mu_0; w_n) = V(\mu_0)$ in problems with progressively larger dual solutions (worse conditioning\footnote{One can show that $d(\theta, \Theta_I) \leq \iota'(c-M\theta)^+\cdot\max_{\lVert p\rVert = 1}\min_{\lambda \in \Lambda(\mu_0)} \lVert \lambda \rVert_\infty$. The latter is the sharp Hoffman constant for $\Theta_I$, see Appendix \ref{ap_proof_th22}. It can be used in the proofs of estimation results in the place of $\kappa(\Theta_I)^{-1}$, which is also a Hoffman constant, but potentially non-sharp.}) in larger samples. Our baseline $\bar \alpha = 0.15$ is a moderate lower-tail choice. Table \ref{tab:holmes-penalty-tilde-alpha-robustness-n3176-n10000} illustrates that a wide range of alternative values, as small as $0.05$ and as large as $0.3$, warrants good performance of the estimator in our simulations. For $\bar \alpha = 0.4$, the estimator becomes conservative for $V(\mu_0)$ at $n = 3176$, but this effect largely disappears by $n = 10000$ by virtue of a $14\%$ increase in $\overline{w}_n$. 

When $p$ is estimated, use $\hat p_n$ in
\eqref{eq:penalty_calibration}, with $d$ the original dimension,
and assign the penalty weight $1+\varepsilon$ to the unnormalized auxiliary
residual $(\hat p_n'\theta-t)^+$, e.g. $\varepsilon=0.01$.
Indeed,
\[
\min_{t\in\mathbb R}\{t+(1+\varepsilon)(\hat p_n'\theta-t)^+\}
=\hat p_n'\theta,
\]
uniquely at $t=\hat p_n'\theta$ for any $\varepsilon > 0$. 

The theoretical justification of the tuning choices in this section is that for nonzero population rows, $M \to \lVert M_j\rVert^{-1}$ is smooth, and \eqref{eq:penalty_calibration} gives a common penalty $\propto \log n$ for the normalized system. It thus satisfies the rate requirements for the normalized system. The uniform conclusions apply under the corresponding assumptions for the normalized system, provided population row norms are uniformly bounded away from zero. 
\begin{table}
\centering
\caption{Penalty robustness to the choice of $\overline{\alpha}$: lower bound on $\tau$}
\label{tab:holmes-penalty-tilde-alpha-robustness-n3176-n10000}
\footnotesize
\setlength{\tabcolsep}{1.5pt}
\begin{tabular}{l*{8}{c}}
\toprule
$\overline{\alpha}$
& \multicolumn{2}{c}{$0.05$} & \multicolumn{2}{c}{$0.2$}
& \multicolumn{2}{c}{$0.3$} & \multicolumn{2}{c}{$0.4$} \\
\cmidrule(lr){2-3}\cmidrule(lr){4-5}
\cmidrule(lr){6-7}\cmidrule(lr){8-9}
$\ell$ & RRMSE & Bias & RRMSE & Bias
       & RRMSE & Bias & RRMSE & Bias \\
\midrule
\multicolumn{9}{c}{\emph{Centroid-collapse}, $n=3176$} \\
\midrule
1    & 1.00 & 0.01  & 1.00 & 0.01  & 0.98 & 0.00  & 1.17 & -0.02 \\
0.5  & 1.02 & 0.01  & 0.99 & 0.01  & 0.95 & -0.00 & 1.10 & -0.04 \\
0.25 & 1.03 & 0.01  & 0.99 & -0.01 & 0.98 & -0.03 & 1.05 & -0.06 \\
0.1  & 1.01 & -0.00 & 0.99 & -0.02 & 1.01 & -0.04 & 1.07 & -0.07 \\
0.05 & 1.02 & 0.01  & 0.99 & -0.01 & 1.00 & -0.03 & 1.04 & -0.05 \\
0.01 & 1.03 & 0.03  & 0.98 & 0.01  & 0.98 & -0.00 & 1.00 & -0.02 \\
0    & 1.04 & 0.04  & 0.98 & 0.02  & 0.98 & 0.01  & 0.98 & -0.01 \\
\midrule
\multicolumn{9}{c}{\emph{Row-collapse}, $n=3176$} \\
\midrule
0.5  & 1.00 & 0.00 & 1.00 & 0.00  & 0.99 & 0.00  & 1.23 & -0.02 \\
0.25 & 1.00 & 0.00 & 1.00 & 0.00  & 0.99 & 0.00  & 1.26 & -0.02 \\
0.1  & 1.00 & 0.00 & 1.00 & 0.00  & 0.99 & 0.00  & 1.23 & -0.01 \\
0.05 & 1.00 & 0.01 & 0.99 & 0.00  & 0.98 & 0.00  & 1.21 & -0.01 \\
0.01 & 1.00 & 0.01 & 0.99 & -0.00 & 0.97 & -0.00 & 1.19 & -0.02 \\
0    & 1.00 & 0.01 & 0.98 & -0.00 & 0.97 & -0.01 & 1.19 & -0.02 \\
\midrule
\multicolumn{9}{c}{\emph{Centroid-collapse}, $n=10000$} \\
\midrule
1    & 1.00 & 0.00 & 1.00 & 0.00  & 1.00 & 0.00  & 0.99 & 0.00  \\
0.5  & 1.00 & 0.01 & 1.00 & 0.01  & 1.00 & 0.01  & 0.94 & -0.00 \\
0.25 & 1.02 & 0.01 & 0.99 & 0.01  & 0.97 & 0.00  & 0.92 & -0.01 \\
0.1  & 1.02 & 0.00 & 0.99 & -0.01 & 0.98 & -0.02 & 1.00 & -0.03 \\
0.05 & 1.00 & 0.00 & 0.99 & -0.01 & 0.99 & -0.02 & 1.02 & -0.03 \\
0.01 & 1.02 & 0.02 & 0.99 & 0.01  & 0.97 & -0.00 & 0.97 & -0.01 \\
0    & 1.03 & 0.02 & 0.99 & 0.02  & 0.97 & 0.01  & 0.96 & -0.00 \\
\midrule
\multicolumn{9}{c}{\emph{Row-collapse}, $n=10000$} \\
\midrule
0.5  & 1.00 & 0.00  & 1.00 & 0.00  & 1.00 & 0.00  & 1.00 & 0.00  \\
0.25 & 1.00 & 0.00  & 1.00 & 0.00  & 1.00 & 0.00  & 1.00 & 0.00  \\
0.1  & 1.00 & 0.00  & 1.00 & 0.00  & 1.00 & 0.00  & 1.00 & 0.00  \\
0.05 & 1.00 & 0.00  & 1.00 & 0.00  & 1.00 & 0.00  & 1.00 & 0.00  \\
0.01 & 1.00 & -0.00 & 1.00 & -0.00 & 1.00 & -0.00 & 1.00 & -0.00 \\
0    & 1.00 & -0.00 & 1.00 & -0.00 & 1.00 & -0.00 & 1.00 & -0.00 \\
\bottomrule
\end{tabular}
\par
\begin{minipage}{0.98\linewidth}
\vspace{0.35em}
\footnotesize
\emph{Notes:} Results for 5{,}000 MC replications. RRMSE divides the RMSE of the debiased penalty estimator by its RMSE
at $\overline{\alpha}=0.15$ in the same design row and sample size.
Bias is the mean estimation error in the lower bound on $\tau$, divided
by the original-polytope $\tau$ interval length. Negative bias is
conservative.
\end{minipage}
\end{table}
\subsection{Two-sided full-sample inference}\label{app:two_sided_fs}
To construct a confidence interval for $p'\Theta_I$ that is valid uniformly under Assumptions \ref{ass:admissible} and \ref{ass:var_uniform}, apply the procedure in Section \ref{sec:uniform_inference} with $p$ and $-p$ on the same sample. Denote the resulting candidate collections by $\mathcal C_n$ and $\mathcal C_n^u$, and their associated triplets by $\hat T(A)=(A,\hat\theta(A),\hat r(A))$ and $\hat T^u(A)$,
respectively. For each $(A_1,A_2)\in\mathcal C_n\times\mathcal C_n^u$, let
$\hat c_{1-\alpha}(A_1,A_2)$ be the conditional $1-\alpha$ quantile of $\max\{\mathbb V_1,\mathbb V_2\}$, floored at
$\underline\sigma z_{1-\alpha}$, where
\[
\mathbb V\mid\mathcal W_n
\sim N\!\left(0,\mathbf\Sigma
(\hat T(A_1),\hat T^u(A_2),\hat\Sigma_n)\right),
\]
and $\mathbf\Sigma$ is defined in Remark \ref{remark_cht_ci}. Let $\widehat{\mathrm{CI}}_{1-\alpha}$ be the following confidence interval
\begin{align*}
\left[
\min_{(A_1, A_2)\in\mathcal C_n \times \mathcal{C}_n^u}p'\hat\theta(A_1)
-\frac{\hat  c_{1-\alpha}(A_1,A_2)}{\sqrt n},\,
\max_{(A_1, A_2)\in\mathcal C_n \times \mathcal{C}_n^u}p'\hat\theta^u(A_2)
+\frac{\hat  c_{1-\alpha}(A_1,A_2)}{\sqrt n}
\right].
\end{align*}
Under the conditions of Theorem \ref{thm:full_sample_uniform_inference},
\[
\liminf_{n\to\infty}\inf_{\P\in\mathcal P_0}
\P\!\left[p'\Theta_I\subseteq
\widehat{\mathrm{CI}}_{1-\alpha}\right]\geq 1-\alpha,
\qquad \alpha\in(0;1/2).
\]
Indeed, the well-conditioned population-optimal bases for both
directions are retained w.p.a.1 uniformly. On this event, the reported interval contains the interval constructed from this pair using its own critical value. Tests of point identification are obtained similarly, see also Remark \ref{remark_cht_ci}. The construction extends to any fixed finite collection of directions $p$. Intersecting such simultaneous confidence halfspaces gives a confidence region for $\Theta_I$.
\subsection{Computational simplifications for Section \ref{sec:uniform_inference}}\label{app: computation_unif}
Computing the set $\hat{\mathcal A}$ requires checking $\binom qd$ square bases. While this is fast in Section \ref{mc_section}, it may be costly in larger problems. Our inference procedure, however, remains valid if $\hat{\mathcal A}$ is replaced by a random set of candidates that contains $J^*$ w.p.a.1 uniformly over $\mathcal P_0$\footnote{\label{fn:candidate_coverage}In practice, a procedure may often succeed without drawing \(J^*\) itself, as the lower envelope may often cover if enough candidates whose second-order centering bias is small relative to their sampling variability are retained, provided these are not perfectly correlated.}. We describe two practical ways to obtain such candidate sets.

\paragraph*{Vertex-traversal approach.} For a data-dependent $\gamma_n\in(0,\infty]$ selected by the user, let
\[
        \hat{\mathcal A}_\gamma
        \equiv
        \{A\in\hat{\mathcal A}:
        L(\hat\theta(A);\hat\mu_n,w_n)\le V(\hat\mu_n; w_n)+\gamma_n, ~ \hat{\theta}(A) \in \overline{\Theta}\},
\]
where $\overline{\Theta}$ is a known non-empty polyhedron. The case $\gamma_n=\infty, \overline{\Theta} = \mathbb{R}^d$ gives the full $\hat{\mathcal A}$. More generally, it follows from Theorems \ref{uniform_penalty_theorem}, \ref{thm:full_sample_uniform_inference} that
\begin{cor}
    If $\sqrt n\,\gamma_n/w_n\to\infty$ w.p.a.1 uniformly over $\mathcal P_0$, $\overline{\Theta}$ is a polyhedron, such that $\Theta \subseteq \text{Int}(\overline{\Theta})$, then Theorem \ref{thm:full_sample_uniform_inference} is valid with $\hat{\mathcal A}_\gamma$ in the place of $\hat{\mathcal{A}}$.
\end{cor}

A finite $\gamma_n$ allows to compute $\hat{\mathcal A}_\gamma$ by traversing the vertices of the polyhedron
\[
    \mathcal K_n
    =
    \left\{
    (\theta,a)\in\overline{\Theta}\times\mathbb R^q:
    \hat M_n\theta+a\ge \hat c_n, ~a \geq 0
    \right\}.
\]
Let $\mathcal{L}_n(\theta,a) \equiv p'\theta+w_n \iota'a$, and suppose $\mathcal{L}_n$ is bounded from below on $\mathcal{K}_n$. To the penalty vertex-solution $\hat{\theta}$ associate $\hat{a} = (\hat{c}_n - \hat{M}_n \hat{\theta})^+$; this gives a vertex $(\hat{\theta},\hat{a})$ of $\mathcal{K}_n$. Assuming $(\hat{\theta}, \hat{a})$ also minimizes $\mathcal{L}_n$ over $\mathcal{K}_n$\footnote{This can be checked in the sample. Otherwise, commence the search at the globally minimizing vertex.}, branch to its adjacent vertices with a weakly larger value of $\mathcal{L}_n$, where two vertices are adjacent if they lie on a common one-dimensional face. Any vertex $(\theta, a)$ of $\mathcal{K}_n$ satisfies $a = (\hat{c}_n - \hat{M}_n\theta)^+$, and it is easy to show that all its adjacent vertices can be found by pivots of the binding constraints of type $M_{\text{aug}}\theta =c_{\text{aug}}$ only (accounting for $\overline{\Theta}$ facets, if any), i.e. the search remains $d$-dimensional. A newly found vertex $(\theta, a)$ is recorded and retained for further branching if
\begin{align}\label{traversal_crit}
    \mathcal{L}_n(\theta,a)\le V(\hat\mu_n; w_n) +\gamma_n .
\end{align}
Otherwise, the branch is pruned and other vertices in the queue are explored until it is empty. Once the search terminates, to each square subbasis of recorded vertices apply the test developed in Section \ref{sec:uniform_inference}. The set of bases for which \eqref{eq:GI} is not rejected is $\hat{\mathcal{A}}_\gamma$. This algorithm is exact, since every vertex satisfying \eqref{traversal_crit} can be reached from \((\hat\theta,\hat a)\) by a sequence of adjacent vertices along which \(\mathcal{L}_n\) is non-decreasing.  

The ideal $\gamma_n$ is equal to the sample optimality gap of the unknown \(J^*\),
\begin{align*}
\gamma^{\text{ideal}}_n \equiv p'\left(\hat\theta(J^*)-\hat\theta\right) 
+
w_n\iota'\left((\hat c_n-\hat M_n\hat \theta(J^*))^+ - (\hat c_n-\hat M_n\hat \theta)^+\right) = O_{\mathcal{P}_0}\left(\frac{w_n}{\sqrt{n}}\right).
\end{align*}
We propose a simple\footnote{We assume that $|p'\Theta| > 1$.} scale-normalized rule that exceeds $\gamma_n^{\text{ideal}}$ w.p.a.1 uniformly over $\mathcal{P}_0$,
\begin{align}\label{gamma_choice_rec}
    \gamma_n
    =
    \left(
    \max_{\theta\in\Theta}p'\theta-\min_{\theta\in\Theta}p'\theta
    \right)
    \frac{\log n \cdot w_n/\sqrt n}
    {\log n_0 \cdot w_{n_0}/\sqrt{n_0}}
    +
    w_n\iota'(\hat c_n-\hat M_n\hat\theta)^+.
\end{align}
The first term is a slowly enlarged \(w_n/\sqrt n\)-radius for the value component of \(\gamma^{\text{ideal}}_n\), measured in the natural units of the target \(p'\theta\). At the reference sample size \(n_0\), it equals one range of \(p'\theta\) over the parameter space \(\Theta\). We use $n_0=100$ in the simulations. The second term is the realized penalty wedge. It expands the search when the sample penalty objective is materially affected by infeasibility, precisely the case in which a radius based only on value differences may be too narrow. This adjustment is not needed for the asymptotic inclusion argument, but is intended to improve the finite-sample performance. 

The computational cost of the vertex-traversal procedure is output-dependent. Under NFF, it is typically\footnote{Unless the number of redundant inequalities at the unique vertex grows with $q,d$. `True' equalities can be accommodated by a straightforward modification of the traversal procedure.} polynomial in $d, q$. If NFF fails, the procedure may be costly for large $d, q$. In that case, one may terminate the traversal once a `reasonable' number of candidates $|\mathcal{C}_n|$ is generated, in line with the intuition in footnote \ref{fn:candidate_coverage}, or resort to the procedure below.

\paragraph*{Randomized candidate generation.}
The second approach is to generate candidate bases by repeatedly re-solving the penalty problem at randomly perturbed coefficient vectors. Form \(\hat\Sigma_{n,+}\) from \(\hat\Sigma_n\) by truncating the singular values from below at $\underline{\sigma} >0$. For \(b=1,\ldots,B\), draw $\hat\mu_b^*|\hat\mu_n,  \hat\Sigma_{n,+}\sim N(\hat\mu_n, n^{-1} \hat\Sigma_{n,+})$ independently. Solve the penalized problem over $\mathbb{R}^d$ (and not over $\Theta$) at each \((\hat\mu_b^*,w_n)\), record the selected vertex-solution, if any, and collect rank-accepted square subbases of its binding original inequalities. If \(\mathcal R_B\) is the set of square bases generated in this way, form the CI in Section \ref{sec:uniform_inference} with $\hat{\mathcal A} \cap \mathcal R_B$ in the place of $\hat{\mathcal{A}}$. 

Because the noise has full support, \(J^*\) is selected by one perturbation draw with positive probability \(\mathrm{p}_n(J^*)>0\) conditional on the data. It is missed after \(B\) draws with probability \((1-\mathrm{p}_n(J^*))^B\), so the number of draws required to include $J^*$ is of order \(\mathrm{p}_n(J^*)^{-1}\). In particular, as $B \to \infty$, $J^*$ is drawn w.p.a.1, conditional on the data. Unfortunately, this does not give an unconditional finite-$B$ guarantee of recovering $J^*$. The advantage of randomized search is computational: each iteration requires only one LP solve, so its computational cost can be expected to grow polynomially with $q, d$ for a fixed $B$. Table \ref{tab:uniform-candidate-generation-centroid} illustrates that a modest number of draws is sufficient to attain nominal coverage in our application. 
\begin{table}
\centering
\caption{Candidate generation comparison:
centroid-collapse, lower bound on $\tau$, $n=3176$}
\footnotesize
\label{tab:uniform-candidate-generation-centroid}
\setlength{\tabcolsep}{4.5pt}
\begin{tabular}{lrrrrrrrr}
\toprule
 & \multicolumn{4}{c}{$\ell=1$} & \multicolumn{4}{c}{$\ell=0$} \\
\cmidrule(lr){2-5}\cmidrule(lr){6-9}
Procedure
& E.L. & Cov. & Avg. $|\mathcal C_n|$ & Sec.
& E.L. & Cov. & Avg. $|\mathcal C_n|$ & Sec. \\
\midrule
Sublevel traversal
& 0.31 & 0.97 & 2.47 & 0.02
& 0.35 & 1.00 & 13.39 & 0.10 \\
Full enumeration
& 0.34 & 0.97 & 2.53 & 0.14
& 0.39 & 1.00 & 13.72 & 0.26 \\
100 perturb-resolves
& 0.30 & 0.97 & 2.35 & 0.28
& 0.33 & 1.00 & 8.71 & 0.33 \\
\bottomrule
\end{tabular}
\begin{minipage}{0.96\linewidth}
\vspace{0.35em}
\footnotesize
\emph{Notes:} Results use 5{,}000 MC replications. Sec. is mean seconds per replication.
\end{minipage}
\end{table}
\subsection{Proof of Theorem \ref{th: main_inference}}\label{proof_main_inf}
\begin{proof}
Define $\Upsilon(A)=\{r \in \mathbb{R}^{|A|}:M_A'r=p\}$, and $\mathbf\Theta(A)=\{\theta\in\mathcal S(\mu_0):M_A\theta=c_A\}$. From the proof of Theorem \ref{root_n_pointwise},
$\hat A\in\mathbb A$ w.p.a.1.  On this event, let $r_n$ and
$\theta_n$ be the projections of $\hat r$ and $\hat\theta$ onto
$\Upsilon(\hat A)$ and $\mathbf\Theta(\hat A)$, respectively, and
$0$ otherwise. On $\{\hat A\in\mathbb A\}$, as $|\mathbb{A}| < \infty$,
\begin{align}\label{eq:representer_projection}
\|\hat r-r_n\|
&\leq\max_{A\in\mathbb A}\|(M_A')^\dagger\|
\bigl(\|p-\hat M_{\hat A}^{(1)\prime}\hat r\|+\|\hat M^{(1)}-M\|\,\|\hat r\|\bigr)
=o_p(1),
\end{align}
where $\dagger$ denotes the Moore-Penrose pseudoinverse. Endow the triplet space
$\mathcal T=\bigsqcup_{A\in\mathbb A}
\{A\}\times\Theta\times\mathbb R^{|A|}$ with the metric
\[
\rho^{\mathcal{T}}((A,\theta,r),(\tilde{A},\tilde\theta,\tilde r))=
\begin{cases}
\min\{1,\max(\|\theta-\tilde\theta\|,\|r-\tilde r\|)\},&A=\tilde{A},\\
2,&A\ne \tilde{A},
\end{cases}
\]
and write $\sigma(T)=\sigma(A,\theta,r,\Sigma)$.  For $C\ge1$, define a $\rho^{\mathcal{T}}-$compact set
\[
\aleph(C)=\bigsqcup_{A\in\mathbb A}\{A\}\times
\{(\theta,r)\in\Theta\times\mathbb R^{|A|}:
C^{-1}\le\|r\|\le C,\ \sigma(A,\theta,r,\Sigma)\ge C^{-1}\}.
\]
Let $Z_n^{(2)}=\sqrt{n_2}(\hat\mu^{(2)}-\mu_0)$ and, for
$z=\text{vec}(z^M,z^c)$ and $T=(A,\theta,r)$,
write $\mathbb{V}(z, T)=r'(z_A^c-z_A^M\theta)$ and
$H_n(T)=\mathbb{V}(Z_n^{(2)},T)/\sigma(T)$.  For every $C<\infty$ and
$x\in\mathbb R$,
\begin{equation}\label{eq:uniform_clt_compact}
\sup_{T\in\aleph(C)}
|\mathbb P[H_n(T)\le x]-\Phi(x)|\to0.
\end{equation}
Otherwise, a violating deterministic sequence has a convergent subsequence $T_n\in\aleph(C)$; Assumption 4 and continuity of $\mathbb{V}(\cdot,\cdot)$ then give $H_n(T_n)\xrightarrow{d} N(0,1)$,
a contradiction.  Independence of the folds, Lemma
\ref{lem:first_fold_reduction}.iii, and \eqref{eq:uniform_clt_compact}
therefore imply $H_n(\hat T)\xrightarrow{d} N(0,1)$.

On $\{\hat A\in\mathbb A\}$, $M_{\hat A}'r_n=p$ and
$r_n'c_{\hat A}=V(\mu_0)$, so
\begin{align}\label{eq:exact_decomposition}
\sqrt{n_2}(\breve V-V(\mu_0))
=\mathbb{V}(Z_n^{(2)}, \hat{T})
+\sqrt{n_2}(\hat r-r_n)'(c_{\hat A}-M_{\hat A}\hat\theta).
\end{align}
The last term in \eqref{eq:exact_decomposition} is $o_p(1)$ by \eqref{eq:representer_projection}, Lemma \ref{lem:first_fold_reduction} and Assumption 4; hence $\sqrt{n_2}(\breve V-V(\mu_0))/\sigma(\hat T) \xrightarrow{d}N(0,1).$ Finally, to extend to $\sigma(\hat T,\hat\Sigma_n)$, note that Assumption 5 and continuity of $\sigma(\cdot)$ imply that, for each $C>1$,
$\sup_{T\in\aleph(C)}
|\sigma(T,\hat\Sigma_n)-\sigma(T,\Sigma)|=o_p(1)$.
Lemma \ref{lem:first_fold_reduction} thus gives
$\sigma(\hat T, \hat{\Sigma}_n)/\sigma(\hat T)\to_p1$. Applying CMT concludes the proof of the Theorem. 

To establish \eqref{motivation},  note that $M_{\hat A}\theta_n=c_{\hat A}$ implies
\[
(\hat r-r_n)'(c_{\hat A}-M_{\hat A}\hat\theta)
=(p-\hat M_{\hat A}^{(1)\prime}\hat r)'(\hat\theta-\theta_n)
+\hat r'(\hat M_{\hat A}^{(1)}-M_{\hat A})(\hat\theta-\theta_n).
\]
By Lemma \ref{lem:first_fold_reduction}.i, the last term is $O_p(n^{-1})$.
\paragraph*{Technical Lemmas} 
\begin{lemma}\label{lem:representer_rate}
Under the conditions of Theorem \ref{th: main_inference}, apart from Assumption 6 and the conditions on $\hat{r}$, \eqref{reg_v} has a unique, first-fold-measurable minimizer a.s., satisfying
\[
\hat r=O_p(1),\qquad
\|p-\hat M_{\hat A}^{(1)\prime}\hat r\|=O_p(n^{-1/2}),\qquad
\mathrm d(\hat r,\Upsilon(\hat A))=O_p(n^{-1/2}).
\]
\end{lemma}
\begin{proof}The eigenvalue floor ensures strong convexity and coercivity of $\mathbf{Q}_n$, the criterion in \eqref{reg_v}, and guarantees a unique optimum a.s. For $A \in \mathbb{A}$, put $r^0(A)=(M_A')^\dagger p\in\Upsilon(A)$.  Finiteness of $\mathbb A$
gives $\max_A\|r^0(A)\|<\infty$ and
$\max_A\|(M_A')^\dagger\|<\infty$.  Assumption 5 gives
$\sup_{A \in \mathbb{A},\theta \in \Theta}\|\Omega_+(A,\theta,\hat\Sigma^{(1)})\|=O_p(1)$. From Step 1 of the proof of Theorem \ref{root_n_pointwise}, for some $\eta>0$, $B(\hat\theta,\eta)\subseteq\Theta_+$ w.p.a.1.  Hence, since $n_2C_\Theta^2\|(\hat M_{\hat A}^{(1)}-M_{\hat A})'
 r^0(\hat A)\|^2=O_p(1)$,
\[
\mathbf{Q}_n(\hat r)\le \mathbf{Q}_n(r^0(\hat A))=O_p(1),\qquad
\underline\sigma^2\|\hat r\|^2+n_2\eta^2
\|p-\hat M_{\hat A}^{(1)\prime}\hat r\|^2\le \mathbf{Q}_n(\hat r),
\]
The first two claims follow.  Finally, for $A \in \mathbb{A}$, $M_A'(M_A')^\dagger M_A' = M_A'$ and $p \in \mathcal{R}(M'_{A})$ imply that $r-(M_A')^\dagger(M_A'r-p)\in\Upsilon(A)$, so, using the first two claims,
\[
\mathrm d(\hat r,\Upsilon(\hat A))
\le \max_{A\in\mathbb A}\|(M_A')^\dagger\|\cdot
\bigl(\|p-\hat M_{\hat A}^{(1)\prime}\hat r\|
+\|\hat M^{(1)}-M\|\,\|\hat r\|\bigr)
=O_p(n^{-1/2}).
\]
This concludes the proof of the Lemma.
\end{proof}
\begin{lemma}\label{lem:first_fold_reduction}
Under the conditions of Theorem \ref{th: main_inference}, i) $\|\hat\theta-\theta_n\|=O_p(n^{-1/2})$;
ii) $T_n=(\hat A,\theta_n,r_n)$ is an optimal triplet w.p.a.1; iii) $\forall \varepsilon>0 ~ \exists C<\infty: \liminf_n\mathbb P[\hat T\in\aleph(C)]\ge1-\varepsilon$.
\end{lemma}
\begin{proof} We first establish i). Note that for $A \in \mathbb{A}$, $\mathbf{\Theta}(A)$ is a bounded polytope by Assumption 1, and $\mathbf{\Theta}(A) = \{\theta \in \mathbb{R}^d:M\theta \geq c, p'\theta \leq V(\mu_0), -M_A\theta \geq - c_A\}$. By Lemma \ref{lemma_polytope_anticoncentration} and Theorem \ref{root_n_pointwise}, using $\hat{A} \in \mathbb{A}$ w.p.a.1,
    \begin{align}
    \text{d}(\hat{\theta}, \mathbf{\Theta}(\hat{A}) ) \leq \max_{A \in \mathbb{A}}\kappa(\mathbf{\Theta}(A))^{-1}  \cdot(\iota'(c - M\hat{\theta})^+ ) + O_p(n^{-1/2}).
    \end{align}
    By \eqref{eq: bound_on_resid}, $\iota'(c - M\hat{\theta})^+ = O_p(1/\sqrt{n})$. ii) follows by definition. It remains to show iii). From i), \eqref{eq:representer_projection}, and $\hat A\in\mathbb A$ w.p.a.1, $\rho^{\mathcal T}(\hat T,T_n)=o_p(1)$ and $r_n=O_p(1)$.  Also, $\|r\|\ge\|p\|/\|M\|>0$ for every $r\in\Upsilon(A)$ and $A \in \mathbb{A}$.  Optimal triplets with bounded $\lVert r\rVert$ form a compact set, on which Assumption 6
and continuity of $\sigma$ give a strictly positive minimum of $\sigma(\cdot)$. Tightness
of $r_n$ and $\rho^{\mathcal{T}}(\hat T,T_n)=o_p(1)$ prove the last claim.
\end{proof}
\end{proof}

\subsection{Proof of Theorem \ref{thm:full_sample_uniform_inference}}
We will require the following Proposition.
\begin{prop}\label{theor_j_star}
    Under Assumption 1, the problem \eqref{lp_init} admits a primal-dual solution pair $ \theta^*, \lambda^*,$ s.t. for some $J^{*} \subseteq [q]$ with $|J^{*}| = d$, $M_{J^{*}}$ is invertible and:
    \begin{align*}
        &\theta^* = M_{J^{*}}^{-1}c_{J^{*}}, ~ \lambda^{*}_{J^{*}} = M_{J^{*}}^{-\prime}p \geq 0,  ~ \lambda^{*}_i = 0 ~ \text{if} ~ i \notin J^{*}.
    \end{align*}
\end{prop}
\begin{proof}
Consider the problem and its associated Lagrangian:
\begin{align*}
    \text{(P)}: \min_\theta ~ p'\theta \quad \text{s.t.}: M\theta \geq c, \quad \mathcal{L} \equiv p'\theta + \lambda'(c - M\theta).
\end{align*}
By boundedness of $\Theta_I$, (P) has a vertex solution giving a primal-dual solution pair $\theta^*,\lambda$, with $\lambda \geq 0$, at which $J \subseteq [q]$ is the set of all binding constraints, with $|J| \geq d$ such that $M_J \theta^* = c_J,$ and $\text{rk}(M_J) = d$. From $\lambda'(c - M\theta^*) = 0$, $\lambda_{[q]\setminus J} = 0$, so the FOC of (P) gives $p = M'\lambda = M_J'\lambda_J$. Thus, $p \in \text{Cone}(M_J')$. By the conical hull Caratheodory's Theorem, $\exists J^* \subseteq J$ such that $|J^{*}| = r \leq d$ and $p   = M'_{J^*} \tilde{\lambda}$ for some $\tilde{\lambda} \in \mathbb{R}^{r}_+$, and the columns of $M_{J^{*}}'$ are linearly independent. If $r < d$, we complement $J^{*}$ with $d - r$ vectors from $M_{J}'$, so that we obtain $\text{rk}(M_{J^{*}}) = d$, and extend $\tilde{\lambda}$ with zeroes. Finally, construct $\lambda^*_{J^*} = \tilde{\lambda}$ and $\lambda^*_j = 0$ for $j \notin J^*$. Then, $\theta^*, \lambda^*$ is the required primal-dual solution pair. 
\end{proof} 
We now prove the theorem.
\begin{proof}
Weyl's inequality, Assumption 8 and $w_n \to \infty$ w.p.a.1 uniformly over $\mathcal{P}_0$, imply that $\mathrm{sv}_d(\hat{M}_{J^*}) > w_n^{-1}$ w.p.a.1 uniformly over $\mathcal{P}_0$; we work on this event without making it explicit for ease of exposition.  Further, whenever $M_A$ is nonsingular, let $s_j^2(A)$, $\Xi(A)$, $\Psi(A)$, $\Gamma(A)$ denote the population counterparts of the corresponding sample quantities; we drop the argument when $A = J^*$. We also write $g^*(\mu) \equiv g(J^*, \mu)$. Finally, define $Z \sim N(0, \Sigma)$, $Z_n \equiv \sqrt{n}(\hat{\mu}_n - \mu_0)$, and $\hat{G} \equiv \sqrt{n}\hat{\Psi}^{-1}(g^*(\hat{\mu}_n) - g^*(\mu_0))$. Throughout this proof, $\mathcal{L}_\P(X)$ denotes the law of a random vector $X$ under $\P$.
   
\paragraph*{1. Linearizing $\hat{G}$.} By Lemma \ref{lem: unif_dm}, Assumption 8.ii, and as $s_j \geq \underline{\sigma}$ for $j \in [q]$,
    \begin{align}\label{dm_bound}
        \sqrt{n}\Psi^{-1}(g^*(\hat{\mu}_n) - g^*(\mu_0)) = \Psi^{-1}Dg^*(\mu_0)Z_n + O_{\mathcal{P}_0}(1/\sqrt{n}).
    \end{align}
    By Lemma \ref{lem: unif_dm} and because $\lVert \Psi^{-1} \rVert \leq \underline{\sigma}^{-1}$, we have that $\lVert \Psi^{-1} D g^*(\mu_0)\rVert < C_4$ for a constant $C_4 > 0$ that does not depend on $\P$. Thus, by Assumption 8.ii and the definition of $\text{d}_{BL_1}$,
    \begin{align}\label{gaus_lin}
        \limsup_{n \to \infty}\sup_{\P \in \mathcal{P}_0} \text{d}_{BL_1}(\mathcal{L}_\P(\Psi^{-1} D g^*(\mu_0) Z_n), \mathcal{L}_\P(\Psi^{-1} D g^*(\mu_0) Z)) =0.
    \end{align}
     Combining \eqref{dm_bound}, \eqref{gaus_lin}, Assumption 9 and using Lemma \ref{lem: unif_dm} to get a uniform lower bound on $\mathrm{sv}_q(\Psi^{-1})$ gives, by the usual uniform CMT argument,
    \begin{align}\label{uniform_gaus}
        \limsup_{n \to \infty} \sup_{\P \in \mathcal{P}_0} \text{d}_{BL_1}(\mathcal{L}_\P(\hat{G}), N(0,\Gamma(\P))) =0.
    \end{align}
    \paragraph*{2. Deterministic GMS set.}
    Define the deterministic local-to-binding set 
    \begin{align*}
        \mathcal{B}_n \equiv \left\{j \in [q]: \sqrt{n} g^*_j(\mu_0)/s_j \geq -{k}_n/2\right\},
    \end{align*}
    and events $E_{1n} \equiv \{\mathcal{B}_n \subseteq \hat{B}(J^*)\}$, and $E_{2n} \equiv \{\max_{j \in \hat{B}(J^*) \setminus \mathcal{B}_n} g^*_j(\hat{\mu}_n) \leq 0\}$. We show that
    \begin{align}
        \liminf_{n \to \infty} \inf_{\P \in \mathcal{P}_0} \P[E_{1n} \cap E_{2n}] = 1.
    \end{align}
    The part concerning $E_{1n}$ follows, because ${k}_n \uparrow \infty$, while for $j \in \mathcal{B}_n$, 
    \begin{align}\label{ineq_loctobind}
        \frac{\sqrt{n} g^*_j(\hat{\mu}_n)}{\hat{s}_j} &= \hat{G}_j + \frac{\sqrt{n} g_j^*(\mu_0)}{s_j}\frac{s_j}{\hat{s}_j} \geq O_{\mathcal{P}_0}(1) - \frac{3}{4}{k}_n,
    \end{align}
    where the last inequality follows by \eqref{uniform_gaus}, and noticing that $3/2\geq s_j/ \hat{s}_j \geq 1/2$ for $j \in [q]$ w.p.a.1 uniformly over $\mathcal{P}_0$, as $\max_j|s_j/\hat{s}_j-1| = o_{\mathcal{P}_0}(1)$ from Lemma \ref{lem: unif_dm} and Assumption 9. The part concerning $E_{2n}$ follows by considering $j \notin \mathcal{B}_n$ and reversing the sign in \eqref{ineq_loctobind}.

    \paragraph*{3. GMS level.}
    For $\beta \in (0; 1/2)$, define the rejection event at $J^*$
    \begin{align*}
        \mathcal{R}_n(\beta)\equiv \{\sqrt{n} \max_{j \in \hat{B}(J^*)} g_j^*(\hat{\mu}_n)/\hat{s}_j > \hat{q}_{1-\beta}(J^*)\}
    \end{align*}
    Under $E_{1n} \cap E_{2n}$, rejection may only stem from $ j \in \mathcal{B}_n$, i.e.
    \begin{align}\label{eq: det_local}
        \mathcal{R}_n(\beta) \cap E_{1n} \cap E_{2n} \subseteq \left\{\max_{j \in \mathcal{B}_n}\sqrt{n}g^*_j(\hat{\mu}_n)/\hat{s}_j  > \hat{q}_{1-\beta}(J^*, \mathcal{B}_n)\right\}.
    \end{align}
    Since $J^*$ is dual-primal feasible and $s_j, \hat{s}_j > 0$, from the first equality in \eqref{ineq_loctobind},
    \begin{align}\label{eq: prim_dual_feas}
        \hat{G} \geq \hat{\Psi}^{-1} \sqrt{n}g^*(\hat{\mu}_n).
    \end{align}
    Combining \eqref{eq: det_local} and \eqref{eq: prim_dual_feas} with Lemma \ref{lem:fixed_set_gaussian_cv} gives, for $\beta \in (0, 1/2)$,
    \begin{align}\label{fixed_beta_rej}
        \limsup_{n \to \infty} \sup_{\P \in \mathcal{P}_0} \P[\mathcal{R}_n(\beta)] \leq \beta.
    \end{align}
    For any $\beta \in (0,1/2)$, $\alpha_{f,n} < \beta$ for a large enough $n$. Thus, by monotonicity of $\mathcal{R}_n(\cdot)$ and sending $\beta \downarrow 0$ in \eqref{fixed_beta_rej},
    \begin{align}\label{eq: j_star_incl}
        \liminf_{n \to \infty} \inf_{\P \in \mathcal{P}_0} \P[J^* \in \mathcal{C}_n] = 1. 
    \end{align}
    \paragraph*{4. Uniform coverage.} By definition of $\widehat{\text{LCB}}_{1-\alpha}$,
\begin{align}\label{eq: coverage_decomp}
\mathbb P\!\left[
V(\mu_0)<\widehat{\mathrm{LCB}}_{1-\alpha}
\right]
\le
\mathbb P[J^*\notin\mathcal C_n] 
+
\mathbb P\!\left[
V(\mu_0)<
p'\hat\theta(J^*)-
z_{1-\alpha}\frac{\underline{\hat\sigma}(J^*)}{\sqrt n}
\right].
\end{align}
Combining \eqref{eq: j_star_incl}, \eqref{eq: coverage_decomp} and Lemma
\ref{lem:value_basis_gaussian} gives
\[
\liminf_{n\to\infty}\inf_{\mathbb P\in\mathcal P_0}
\mathbb P\!\left[
V(\mu_0)\ge \widehat{\mathrm{LCB}}_{1-\alpha}
\right]\ge1-\alpha.
\]
\paragraph*{5. Pointwise rate.} 
Since there are finitely many square bases, the requirement $\mathrm{sv}_d(\hat{M}_A) > w^{-1}_{n}$, Weyl's inequality, $\hat{\mu}_n = \mu_0 + O_p(n^{-1/2})$ and $w_n \to \infty$ w.p.a.1 with $w_n/\sqrt n=o_p(1)$ imply that all bases in $\hat{\mathcal{A}}$ are nonsingular w.p.a.1. For any square nonsingular $A$ that does not satisfy \eqref{eq:GI}, $g_j(A,\mu_0) > 0$ for some $j \in [q]$; its test statistic is thus positive w.p.a.1 and of order \(\sqrt n\). For a fixed-dimensional Gaussian maximum, $
\hat q_{1-\alpha_{f,n}}(A)
=O_p\!\left(\sqrt{\log(1/\alpha_{f,n})}\right)
=o_p(\sqrt n)$ as $\log \alpha_{f,n} = o(n)$. Thus, $A$ is rejected w.p.a.1. Finiteness of the collection gives simultaneous rejection. Because \eqref{eq:GI} is sufficient for a square nonsingular $A$ to induce $\theta(A) \in \mathcal{S}(\mu_0)$, and there are finitely many such bases, $\sqrt{n}(\hat{\mu}_n - \mu_0) \xrightarrow{d} N(0,\Sigma)$ and $\hat{\Sigma}_n \xrightarrow{p} \Sigma$, CMT and the Delta-method give
\begin{align*}
    \min_{A \in \hat{\mathcal{A}}} \left\{p'\hat{\theta}(A) - z_{1-\alpha}\frac{\underline{\hat\sigma}(A)}{\sqrt{n}}\right\}= V(\mu_0) + O_p(1/\sqrt{n}).
\end{align*}
This is extended to $\mathcal{C}_n$ by applying Theorem \ref{root_n_pointwise} and observing that $\hat{r} = O_p(1)$ by the same argument as in Lemma \ref{lem:representer_rate}, so $\underline{\hat\sigma}(\hat A)=O_p(1)$. 
\paragraph*{6. Convergence to oracle.} The argument in Step 5 under the additional conditions that $\mathcal{S}(\mu_0) = \{\theta^*\}$, and $|J(\theta^*; \mu_0)|=d$, implies that $\hat{\mathcal{A}} = \{J^*\}$ w.p.a.1. Moreover, from the proof of Theorem \ref{root_n_pointwise}, $\hat{A} = J^*$ w.p.a.1, as a unique $A = J^*$ satisfies \eqref{optimality} with $\theta^* \in \mathcal{S}(\mu_0)$ and \eqref{range_req}. Finally, from the proof of Lemma \ref{lem:representer_rate}, $\text{d}(\hat{r}, \Upsilon(\hat{A})) = O_p(n^{-1/2})$. Since $\text{rk}(M_{J^{*}}) = d$ and $J^* = \hat{A}$ w.p.a.1, this implies $\hat{r} = M^{-\prime}_{J^*}p + O_p(n^{-1/2})$, and the result follows by CMT.

\paragraph*{Technical Lemmas} 
    \begin{lemma}[Uniform Delta Method]\label{lem: unif_dm}
        Under Assumption 8, there exist positive constants $C_1, C_2, C_3$, such that, for any $h \in \mathbb{R}^m$, with $\lVert h \rVert \leq C_1$, and $\P \in \mathcal{P}_0$,
        \begin{align*}
            \max_{j \in [q]} |g^*_j(\mu_0 + h) - g^*_{j}(\mu_0) - Dg^*_j(\mu_0)[h]| \leq C_2\lVert h\rVert^2, \quad  \max_{j \in [q]}|Dg^*_j(\mu_0)[h]|\leq C_3 \lVert h \rVert,
        \end{align*}
        where $Dg^*_j(\mu_0)[h]$ is the linear directional derivative of $g^*_j$ at $\mu_0$ in the direction $h$.
    \end{lemma}
\begin{proof}
        Fix $\P \in \mathcal{P}_0$ and $h = \text{vec}(h^M,h^c) \in \mathbb{R}^m$. From \eqref{first_order}, submultiplicativity of $\lVert \cdot \rVert$, Assumption 8.i, 8.iii and Proposition \ref{theor_j_star}, we have
        \begin{align*}
            \max_{j \in [q]} |Dg^*_j(\mu_0)[h]| \leq \tilde{C}_3 (1 + (1+C_\Theta)\overline{\sigma}/\delta + \lVert p \rVert/\delta^2)\lVert h \rVert 
        \end{align*}
        for a universal constant $\tilde{C}_3$. Further differentiating \eqref{first_order}, using the same argument as above as well as Weyl's inequality, one observes that there exist constants $C_1, C_2 > 0$, such that,
        \begin{align}\label{residual_control}
            \max_{j\in[q]}\sup_{\lVert h_1\rVert\leq C_1}|D^2g^*_j(\mu_0+h_1)[h,h]|\leq C_2\lVert h\rVert^2,\quad h\in\mathbb R^m,~\P\in\mathcal P_0.
        \end{align}
        Consider Taylor expansion at $\mu_0$ and use \eqref{residual_control} to conclude the proof of the Lemma.
\end{proof}
\begin{lemma}\label{lem:fixed_set_gaussian_cv}
    Under Assumptions 8 and 9, for every fixed $\beta \in (0, 1/2)$, 
    \begin{align*}
        \limsup_{n \to \infty} \sup_{\P \in \mathcal{P}_0} \sup_{B \subseteq [q]} \P[\max_{j \in B} \hat{G}_{j} > \hat{q}_{1-\beta}(J^*, B)] \leq \beta. 
    \end{align*}
\end{lemma}
\begin{proof}
For $\P \in \mathcal{P}_0$, consider $\mathbb{G} \sim N(0, \Gamma(\P))$. By Assumption 9 and due to smoothness of $D g^*(\mu)$ established in Lemma \ref{lem: unif_dm}, we have
    \begin{align}\label{xi_const}
        \lVert \hat{\Xi} - \Xi\rVert = o_{\mathcal{P}_0}(1).
    \end{align}
     Because $x\mapsto[\max\{x,\underline\sigma^2\}]^{-1/2}$
is uniformly Lipschitz on $[0,\infty)$,  for $C_5 > 0$ that does not depend on $\P \in \mathcal{P}_0$,
    \begin{align}\label{psi_lip}
        \lVert \hat{\Psi}^{-1} - \Psi^{-1} \rVert \leq C_5  \lVert \hat{\Xi} - \Xi\rVert.
    \end{align}
    Combining \eqref{xi_const} and \eqref{psi_lip} gives $\lVert \hat{\Gamma} - \Gamma\rVert = o_{\mathcal{P}_0}(1),$ implying that
    \begin{align}\label{cond_conv}
        \text{d}_{BL_1}(N(0, \hat{\Gamma})|\mathcal{W}_n, N(0, \Gamma(\P))) = o_{\mathcal{P}_0}(1).
    \end{align}
    For $G \in \mathbb{R}^q$ and $B \in 2^{[q]}\setminus \{\emptyset\}$, $G \to \mathcal{M}_B G \equiv \max_{j \in B} G_j$ is $1-$Lipschitz. From \eqref{uniform_gaus}, \eqref{cond_conv}, 
    \begin{align}\label{gaussians_conv}
        \max_{B \subseteq [q], B \ne \emptyset}\text{d}_{BL_1}(\mathcal{L}_{\P|\mathcal{W}_n}(\mathcal{M}_B\hat{\mathbb{G}}), \mathcal{L}_{\P}(\mathcal{M}_B\mathbb{G})) =o_{\mathcal{P}_0}(1), \\ \label{g_to_gauss}
        \lim\sup_{n \to \infty} \sup_{\P \in \mathcal{P}_0}   \max_{B \subseteq [q], B \ne \emptyset}\text{d}_{BL_1}(\mathcal{L}_{\P}(\mathcal{M}_B\hat{G}), \mathcal{L}_{\P}(\mathcal{M}_B\mathbb{G})) =0.
    \end{align}
    We claim that for $0< \eta < z_{1-\beta}/4, ~ t \geq z_{1-\beta}/2$, for $C_5$ that depends only on $\beta$ and $q$,
    \begin{align}\label{anticoncentration_maxima}
        \max_{B \subseteq [q], B \ne \emptyset} \P[|\mathcal{M}_B \mathbb{G} - t|\leq \eta] \leq C_5 \eta. 
    \end{align}
    To see that, for a non-empty $B \subseteq [q]$, use the union bound $\P[|\mathcal{M}_B \mathbb{G} - t|\leq \eta] \leq \sum_{j \in [q]} \P[|\mathbb{G}_j - t|\leq \eta].$ If $\Gamma_{jj} = \text{Var}(\mathbb{G}_j) = 0$, the corresponding probability is $0$. Otherwise, as $\Gamma_{jj} \leq 1$ by definition, $\P[|\mathbb{G}_j - t|\leq \eta] \leq  \eta \cdot\max_{0< \varpi \leq 1}2\varpi^{-1}\phi(z_{1-\beta}/(4\varpi)) < \infty.$
    
Put $z=z_{1-\beta}>0$, and, for $B\ne \emptyset$, let $F_{\P,B}$ and
$\hat F_B$ denote the CDFs of $\mathcal M_B\mathbb G$ and
$\mathcal M_B\hat{\mathbb G}\mid\mathcal W_n$, respectively.
By bounded-Lipschitz smoothing and \eqref{anticoncentration_maxima},
\eqref{gaussians_conv}--\eqref{g_to_gauss} give
\begin{align*}
 \mathfrak{F}_n:=\max_{\emptyset\ne B\subseteq[q]}
 \sup_{t\ge z}|\hat F_B(t)-F_{\P,B}(t)|
 &=o_{\mathcal P_0}(1),\\
 \sup_{\P\in\mathcal P_0}\max_{\emptyset\ne B\subseteq[q]}
 \sup_{t\ge z}
 |\P[\mathcal M_B\hat G\le t]-F_{\P,B}(t)|&=o(1).
\end{align*}
Fix $\varepsilon\in(0,1-\beta)$ and put
$b_{\P,B} \equiv \max\{z,F_{\P,B}^{-1}(1-\beta-\varepsilon)\}$. Gaussian maxima have no
atoms on $(0,\infty)$, so
$\hat F_B(\hat q_{1-\beta}(J^*,B))\ge1-\beta$.
On $\{\mathfrak{F}_n\le\varepsilon\}$, we get
$F_{\P,B}(\hat q_{1-\beta}(J^*,B))\ge1-\beta-\varepsilon$
and hence $\hat q_{1-\beta}(J^*,B)\ge b_{\P,B}$.
Since $b_{\P,B}$ is deterministic and at least $z$, uniformly over
$\P \in \mathcal{P}_0$ and nonempty $B$,
\[
 \P[\mathcal M_B\hat G>\hat q_{1-\beta}(J^*,B)]
 \le\P[\mathfrak{F}_n>\varepsilon]+1-F_{\P,B}(b_{\P,B})+o(1)
 \le\beta+\varepsilon+o(1).
\]
Take the uniform limsup and then let $\varepsilon\downarrow0$.
For $B=\emptyset$, rejection is impossible under the stated conventions. This concludes the proof of the Lemma.
\end{proof}
\begin{lemma}[Known optimal basis value bound]
\label{lem:value_basis_gaussian}
Under Assumptions 8 and 9,
\[
\limsup_{n\to\infty}\sup_{\mathbb P\in\mathcal P_0}
\mathbb P\!\left[
V(\mu_0)<
p'\hat\theta(J^*)-
z_{1-\alpha}\frac{\underline{\hat\sigma}(J^*)}{\sqrt n}
\right]\le\alpha .
\]
\end{lemma}
\begin{proof}
Since $J^*$ is optimal,
$p'\theta(J^*)=V(\mu_0)$. For $\mu=\text{vec}(\tilde M,\tilde c)$ with $\tilde M_{J^*}$ nonsingular, define $r_{J^*}(\mu)=p'\tilde M_{J^*}^{-1}\tilde c_{J^*}$. For $h=\text{vec}(h^M,h^c)$, we have 
\begin{align}\label{dd_opt_basis}
    Dr_{J^*}(\mu_0)[h] = \lambda(J^*)'(h^c_{J^*} - h^M_{J^*} \theta(J^*)),
\end{align}
and since $\mathrm{sv}_d(M_{J^*})\ge\delta$, by the same uniform delta method arguments as in the proof of Lemma \ref{lem: unif_dm}, Assumptions  \ref{ass:admissible} and \ref{ass:var_uniform} give the uniform Gaussian approximation
\begin{align*}
    \limsup_{n \to \infty}\sup_{\P \in \mathcal{P}_0}\text{d}_{BL_1}(\mathcal L_\P(\sqrt{n}(p'\hat{\theta}(J^*)-V(\mu_0))),N(0,\sigma^2(J^*)))=0,
\end{align*}
where, from \eqref{dd_opt_basis}, $\sigma^2(J^*) = \lambda(J^*)'\Omega(J^*, \theta(J^*),\Sigma)\lambda(J^*)$. Note that, on $\{J^* = \hat{A}\}$,  $\hat{\lambda}(J^*)$ is feasible in the full-sample analogue of \eqref{reg_v} and has an $O_{\mathcal{P}_0}(1)$ loss. By optimality of $\hat{r}$ and $p = \hat{M}'_{J^*}\hat{\lambda}(J^*)$, one gets 
\begin{align*}
    \mathds{1}\{J^* = \hat{A}\} n\max_{\theta \in \Theta_+} ((\hat{\lambda}(J^*) - \hat{r})'(\hat{c}_{J^*} - \hat{M}_{J^*}\theta))^2 \leq O_{\mathcal{P}_0}(1).
\end{align*}
Restrict to an event $\{\text{sv}_d(\hat{M}_{J^*}) > \delta/2\}$, which holds w.p.a.1, and use $\theta(J^*) \in\text{Int}(\Theta_+)$ to get $\mathds{1}\{J^* = \hat{A}\} \lVert \hat{r} - \hat{\lambda}(J^*)\rVert = o_{\mathcal{P}_0}(1)$. Using this and observing that $\Omega(J^*, \hat{\theta}(J^*),\hat{\Sigma}_n) - \Omega(J^*, \theta(J^*),\Sigma) = o_{\mathcal{P}_0}(1)$ and $\hat{\lambda}(J^*) =  \lambda(J^*) + o_{\mathcal{P}_0}(1)$ gives $\underline{\hat\sigma}(J^*)=\max\{\sigma(J^*),\underline\sigma\}+o_{\mathcal P_0}(1)$, so the one-sided bound controls asymptotic coverage, and is exact at a fixed DGP when $\sigma(J^*)\ge\underline\sigma$ by standard arguments. This concludes the proof of the Lemma.
\end{proof}
\end{proof}
\subsection{Polytope geometry}
\subsubsection{Proof of Lemma \ref{lemma_polytope_anticoncentration}}
\begin{proof}
    The case $\theta \in \tilde{\Theta}$ is trivial. Fix $\theta$ such that $\text{d}(\theta,\tilde{\Theta}) > 0$. The unique projection $y^*$ of $\theta$ onto $\tilde{\Theta}$ solves $\min_{y \in \tilde{\Theta}} \frac{1}{2}(y - \theta)'(y-\theta)$. The Lagrangian $\mathcal{L} = \frac{1}{2}(y-\theta)'(y-\theta) + \lambda'(\tilde{c} - \tilde{M}y),$ gives the FOCs $\tilde{M}y \geq \tilde{c}, \lambda \geq 0$ and
    \begin{align}\label{foc_1}
        y- \theta - \tilde{M}'\lambda &= 0,\\\label{foc_2}
        \lambda_j(\tilde{M}_jy  - \tilde{c}_j) &= 0,~  j \in [q].
    \end{align}
     Let $J \equiv \{j \in [q]| \tilde{M}_jy^* = \tilde{c}_j\}$. The smallest face to which $y^*$ belongs is
    \begin{align*}
        f^* = \bigcap_{f \text{ is a face of $\tilde{\Theta}$}: ~ y^* \in f}f,
    \end{align*}
    and $f^*$ is characterized by $J$. Suppose $f^*$ has dimension $k^* \in \{0\}\cup [d-1]$, with $\text{rk}(\tilde{M}_{J}) = d - k^*$ (see Chapters 2,3 in \citet{grunbaum1967convex}). By the same argument as in the proof of Proposition \ref{theor_j_star}, there exist $J^*\subseteq J$ and the KKT vector $\lambda^* \geq 0$ satisfying \eqref{foc_1}, \eqref{foc_2}, such that $|J^*| = d - k^*,$ matrix $\tilde{M}'_{J^*}$ has full column-rank, and $\lambda^*_{[q]\setminus J^*} = 0$. From \eqref{foc_2}, $\lambda^{*\prime}\tilde{c} = \lambda^{*\prime}\tilde{M}y^*$, which implies
    \begin{align}\label{first_step_anticonc}
        \lambda^{*\prime}(\tilde{c}-\tilde{M}\theta) = \lambda^{*\prime}\tilde{M}(y^*-\theta) = \text{d}(\theta,\tilde{\Theta})^2,
    \end{align}
    where the last equality follows from \eqref{foc_1} and $\lVert y^* - \theta \rVert = \text{d}(\theta, \tilde{\Theta})$. By the spectral property of singular values and Definition \ref{def_polcond},
    \begin{align}\label{up_bound_norm_lam}
        \lVert \tilde{M}'\lambda^*\rVert =  \lVert \tilde{M}'_{J^*}\lambda_{J^*}^*\rVert \geq \lVert \lambda_{J^*}^*\rVert \mathrm{sv}_{d-k^*}(\tilde{M}_{J^*}) \geq \lVert \lambda_{J^*}^*\rVert \kappa(\tilde{\Theta}).
    \end{align}
    Combining \eqref{foc_1} and \eqref{up_bound_norm_lam} and using $\lVert \cdot\rVert \geq \lVert \cdot \rVert_\infty$ in $\mathbb{R}^{d-k^*}$ , we get
    \begin{align}\label{inf_norm_bound}
        \lVert \lambda_{J^*}^*\rVert_\infty \leq \text{d}(\theta,\tilde{\Theta}) \kappa^{-1}(\tilde{\Theta}).
    \end{align}
    Let us rewrite 
    \begin{align}\label{eq174}
        \lambda^{*\prime}(\tilde{c}-\tilde{M}\theta) \leq  \lambda^{*\prime}(\tilde{c}-\tilde{M}\theta)^+ \leq  \lVert \lambda^*\rVert_\infty \iota'(\tilde{c}-\tilde{M}\theta)^+.
    \end{align}
    Combining \eqref{first_step_anticonc}, \eqref{inf_norm_bound}, \eqref{eq174} and rearranging establishes the inequality in the lemma. 
\end{proof}
\subsubsection{Proof that Definition \ref{def_polcond} is well-defined}\label{ap_def}
\begin{proof}
    The inequality is immediate, so it suffices to show the first equality. For a $k-$face $f$ with $k \in \{0\} \cup [d-1]$, let $I(f) \equiv \{j \in [q]: \tilde{M}_j \theta = \tilde{c}_j ~ \forall \theta \in f\}$. If there are no faces of non-trivial dimension, the proof is complete. Otherwise, fix any $
    k-$face $f$ with $k\in [d-1]$. There exists a vertex of $\tilde\Theta$, $f^*$, such that $f^* \subseteq f$ (see p. 31 in \citet{grunbaum1967convex}). It follows that $I(f) \subseteq I(f^*)$, so, for any $B \subseteq I(f)$ with $\text{rk}(\tilde{M}_B) = |B|=d - k$, we have $B \subseteq I(f^*)$. Since $f^*$ is a vertex, $\text{rk}(\tilde{M}_{I(f^*)}) = d$, so there exists a set $\tilde{B}$:  $B \cup \tilde{B} \subseteq I(f^*)$, $\tilde{B} \cap B = \emptyset$, $|\tilde{B}| = k$ and $\text{rk}(\tilde{M}_{B \cup \tilde{B}}) = d$. Recall that
    \begin{align}\label{opt1}
        \mathrm{sv}_{d-k}(\tilde{M}_{B}) = \min_{r\in \mathbb{R}^{|B|}, \lVert r \rVert = 1} \lVert \tilde{M}'_{B}r \rVert,  \text{ and } \mathrm{sv}_{d}(\tilde{M}_{B\cup \tilde{B}}) =  \min_{r\in \mathbb{R}^{|B| + k}, \lVert r \rVert = 1} \lVert \tilde{M}'_{B \cup \tilde{B}}r\rVert.
    \end{align}
    Comparing the optimization problems in \eqref{opt1}, one concludes that $\mathrm{sv}_{d}(\tilde{M}_{B\cup \tilde{B}}) \leq \mathrm{sv}_{d-k}(\tilde{M}_B).$ As $k$, $f$ and $B$ were arbitrary, using the definition of $\tilde{\kappa}$ establishes the claim. 
\end{proof}
\subsubsection{Proof of Proposition \ref{delta_properites}}
\begin{proof}
    Part 1, and non-strict inclusions in 2, 3 follow from the definitions. To prove strictness in 2, note that, since $\mathrm{Int}(\Theta) \ne \emptyset$, $\Theta$ includes a shifted by $\overline{\theta} \in \mathbb{R}^d$ and scaled by $\varepsilon > 0$ $d$-simplex $\varepsilon\mathbb{S}_d +\overline{\theta}$, and all rotations/reflections $\varepsilon U\mathbb{S}_d +\overline{\theta}$ for unitary $U \in \mathbb{R}^{d \times d}$. First, suppose $p = \iota$. We have $\tilde{\Theta} = \{\theta :M\theta \geq c\} \subseteq \varepsilon\mathbb{S}_d +\overline{\theta}\subseteq \Theta$, with $M = (I_d ~ \iota_d ~ -\iota_d)' \in \mathbb{R}^{d+2\times d}$, $c = (\overline{\theta}' ~ \varepsilon + \iota'\overline{\theta} ~ -(\varepsilon+\iota'\overline{\theta}))' \in \mathbb{R}^{d+2}$ if $q = d+2$. If $q > d+2$, append $q - d - 2$ strictly slack over $\Theta$ inequalities. $\tilde{\Theta} \subseteq \Theta$ is non-empty, so for some $\P \in \mathcal{P}$ $\Theta_I(\mu_0(\P)) = \tilde{\Theta}$ by Assumption 7. It is immediate to see that $\P \notin \mathcal{P}^{\text{LICQ},0} \cup \mathcal{P}^{\text{SC}}$. NFF also fails, so $\P \notin \mathcal{P}^{\text{NFF}}$. This proves 2 for $p = \iota$. For a general fixed $p \ne 0$, the result follows by rotating/reflecting $\tilde{\Theta} - \overline{\theta}$ to align $p$ with $\iota$. For any fixed $\delta >0$, strictness in part 3 follows by rescaling $M, c$ in the above construction until $\kappa(\Theta_I(\mu_0(\P))) > \delta$.
\end{proof}
\subsection{Extension to the unknown $\Theta$ case}\label{unknown_Theta}
We thank an anonymous referee for pointing out that our main results extend to the case when the penalized value is defined as $V(\mu;w) \equiv \inf_{\theta \in \mathbb{R}^d} L(\theta;\mu, w) \in \{-\infty\} \cup \mathbb{R}$, i.e. does not assume the knowledge of $\Theta$, and the equivalence in \eqref{eq: lp_reform} holds a.s. with $\overline{\Theta} = \mathbb{R}^d$. We formalize it in this subsection. Throughout the subsection, we rely on the unconstrained definition, and we do not suppose that $\Theta$ is encoded in the constraints $M\theta \geq c$. $p$ is fixed, other definitions and Assumptions in the main text are carried over verbatim, with the only difference that $\Theta$ is not known to the researcher.
\begin{lemma}[Uniform localization]\label{unif_loc}
    Suppose a non-empty $\mathcal{P}^* \subseteq \mathcal{P}$ is s.t. $\exists \mathfrak{h} \in \mathbb{R}_{++}$, for which $\mathfrak{h}\mathrm{d}(\theta, \Theta_I(\P)) \leq \iota'(c(\P)-M(\P)\theta)^+ ~\forall \theta \in \mathbb{R}^d, \P \in \mathcal{P}^*$. If $\lVert \hat{\mu}_n - \mu_0\rVert = O_{\mathcal{P}^*}(1/\sqrt{n})$ and $w_n \to \infty$ w.p.a.1 uniformly over $\mathcal{P}^*$, then for any $u_n \downarrow 0$ with $u_n \sqrt{n} \to \infty$, 
    \begin{align*}
        \liminf_{n \to \infty}\inf_{\P \in \mathcal{P}^*}\P[\mathcal{S}(\hat{\mu}_n; w_n) \subseteq \Theta_I^{u_n} \cap \mathcal{S}(\hat{\mu}_n; w_n) \ne \emptyset]  = 1.
    \end{align*}
\end{lemma}
\begin{proof}
    Fix $\theta \in \mathbb{R}^d$. Let $\Delta_M \equiv \sqrt{q}\lVert \hat{M}_n - M\rVert$, $\Delta_{\mu} \equiv \sqrt{q} \lVert \hat{c}_n - c\rVert + C_\Theta \Delta_M$. We have
    \begin{align}\label{pen_lb}
        \iota'(\hat{c}_n - \hat{M}_n \theta)^+ & \geq  \mathfrak{h}\mathrm{d}(\theta, \Theta_I) -\sqrt{q}\lVert \hat{c}_n - c\rVert -  \Delta_M\lVert \theta \rVert \\ \notag
        &\geq (\mathfrak{h} - \Delta_M) \text{d}(\theta, \Theta_I)  - \Delta_\mu.
    \end{align}
    Using \eqref{pen_lb} and $p'\theta \geq V(\mu_0) - \lVert p \rVert \text{d}(\theta, \Theta_I)$,
    \begin{align}\label{lb_on_l_ur}
        L(\theta; \hat{\mu}_n, w_n) \geq V(\mu_0) + w_n(\mathfrak{h} - \Delta_M - \lVert p\rVert/w_n) \text{d}(\theta, \Theta_I) -  w_n\Delta_\mu
    \end{align}
    Because $\text{d}(\theta, \Theta_I) \geq  \lVert \theta \rVert -C_\Theta$ and $V(\hat{\mu}_n; w_n) \leq V(\mu_0) + w_n\Delta_\mu$, \eqref{lb_on_l_ur} implies that a penalty solution exists w.p.a.1 uniformly over $\mathcal{P}^*$. Any such solution $\hat{\theta}$ should then satisfy $(\mathfrak{h} - \Delta_M - \lVert p\rVert/w_n) \text{d}(\hat{\theta}, \Theta_I) \leq 2\Delta_\mu$, from which the conclusion follows.
\end{proof}
The split-sample representer in \eqref{reg_v} and its full-sample analogue are infeasible if $\Theta$ is unknown. Thus, in the result below we do not require $\hat{r}(\hat{A}) = \hat{r}$ for Theorem \ref{thm:full_sample_uniform_inference}.
\begin{prop}
    With the conventions above, Theorems \ref{root_n_pointwise}, \ref{th: main_inference} hold absent the requirements that $\mathcal{S}(\mu_0) \subseteq \mathrm{Int}(\Theta)$ or that $\Theta$ is a polytope, and Theorem \ref{uniform_penalty_theorem} holds. Moreover, if in Theorem \ref{thm:full_sample_uniform_inference}, $\mathds{1}\{J^* = \hat{A}\} \lVert \hat{r}(\hat{A}) - \hat{\lambda}(J^*)\rVert = o_{\mathcal{P}_0}(1)$, and the same condition holds pointwise for any $\P \in \mathcal{P}$ along with $\hat{r}(\hat{A}) = O_{\P}(1)$, then Theorem \ref{thm:full_sample_uniform_inference} holds absent the requirements that $\mathcal{S}(\mu_0) \subseteq \mathrm{Int}(\Theta)$ or that $\Theta$ is a polytope.
\end{prop}
\begin{proof}
    The conclusions of Theorems \ref{root_n_pointwise} and \ref{th: main_inference} follow immediately by considering $\tilde{\Theta} = \text{Closure}(\Theta^{u})$ for some $u > 0$, using Lemma \ref{unif_loc} to see that $\mathcal{S}(\hat{\mu}_n; w_n) \subseteq \text{Int}(\tilde{\Theta})$ occurs w.p.a.1, and otherwise following the proofs under $\mathcal{S}(\mu_0) \subseteq \text{Int}(\tilde{\Theta})$. The proof of \eqref{uniform_non_cons} is unchanged once one localizes to $\tilde{\Theta}$ on an event whose probability tends to $1$ uniformly by Lemma \ref{unif_loc}. The conclusion in \eqref{uniform_validity} follows immediately by recalling that in the main text it is also valid for $V(\hat{\mu}_n;w_n)$, while relaxing the restriction to $\Theta$ reduces the penalized value, and $\hat{V}(\hat{\mu}_n; w_n) \leq V(\hat{\mu}_n;w_n)$ both with and without the restriction to $\Theta$, where both sides may be $-\infty$. Finally, Theorem \ref{thm:full_sample_uniform_inference} depends on $\Theta$ only through the definition of $\hat{A}$ and $\hat{r}(\hat{A})$. The conditions of the Proposition ensure that when $\hat{A}$, under any definition, coincides with $J^*$, $\hat{r}(\hat{A})$ consistently estimates its dual solution, thereby giving the uniform validity claim. The root-n convergence claim follows from Theorem \ref{root_n_pointwise} and $\hat{r}(\hat{A}) = O_{p}(1)$, while the oracle claim follows from Theorem \ref{root_n_pointwise} and $\mathds{1}\{J^* = \hat{A}\} \lVert \hat{r}(\hat{A}) - \hat{\lambda}(J^*)\rVert = o_{p}(1)$, while keeping the rest of the proof of Theorem \ref{thm:full_sample_uniform_inference} unchanged. 
\end{proof}
\subsection{The polytope case of Theorem \ref{root_n_pointwise}}
\label{app:polytope_theta}

We give the modifications to Appendix \ref{ap_proof_th22} for case ii).
Let $A_\Theta\subseteq[q]$ index the deterministic inequalities
defining $\Theta=\{\theta:M_{A_\Theta}\theta\ge c_{A_\Theta}\}$,
so $\hat M_{A_\Theta}=M_{A_\Theta}$ and
$\hat c_{A_\Theta}=c_{A_\Theta}$ by the convention following
Assumption 1. We work on $\{w_n>0\}$, whose probability tends to $1$.

Step 1 is unchanged. In Step 2, retain $\overline\Theta=\Theta$
in \eqref{eq: lp_reform}. This LP has a nonempty compact optimal
set, since $\theta\in\Theta$ and optimality requires
$a=(\hat c_n-\hat M_n\theta)^+$. At a vertex $(\theta_n^*,a)$,
put $\hat A=J(\theta_n^*;\hat\mu_n)$ and
$K=A_\Theta\cap\hat A$. Every binding hard constraint
$M_j\theta=c_j$, $j\in K$, follows by subtracting the binding
equality $a_j=0$ from $a_j+\hat M_j\theta=\hat c_j$.
Thus, the binding penalty constraints alone have rank $d+q$,
and the row operations in Step 2 again give
$\mathrm{rk}(\hat M_{\hat A})=d$.

In Step 3, the argument excluding $E_n$ is unchanged.
Stationarity in \eqref{pos_lin_hul} changes to
\begin{align}\label{stat_unb}
 p=\hat M_n'\tau+M_K'\nu,\qquad
 0\le\tau\le w_n\iota,\quad \nu\ge0.
\end{align}
The remaining KKT conditions still give $\tau_j=w_n$ on $A^*$
and $\tau_j=0$ outside $A^*\cup\hat A$.
On $F_n$, $M_Kh(\hat A)=0$ because $K\subseteq\hat A$.
Hence, multiplying both sides of the equality in \eqref{stat_unb} by $h(\hat A)$ gives, on $F_n$, $-1=w_n\iota'M_{A^*}h(\hat A)+O_p(w_n/\sqrt n)$, where the $O_p(w_n/\sqrt{n})$ does not feature any of the $K$ inequalities -- hence no bound on $\nu$ is needed. The same finite-set argument yields
$\mathbb P[F_n]\to0$, and Step 4 applies unchanged to prove the
theorem. 
\subsection{Bootstrapping the plug-in fails even under SC}\label{subsection_bad_inference}
In this subsection and Section \ref{ap_alt_inf}, we return to the full parameterization $\mu=\text{vec}(\tilde p,\tilde M,\tilde c)$, $\mu_0=\text{vec}(p,M,c)$, with $m=d+qd+q$. Accordingly, Assumptions 1, 3 and 4 concern the full coefficient vector, with $\mathbb Z=\text{vec}(\mathbb Z^p,\mathbb Z^M,\mathbb Z^c)$.
We examine the properties of the approach that combines bootstrap on $\hat{\mu}_n$ with the plug-in estimator $V(\hat{\mu}_n)$. This method is widely used in empirical literature (\citet{blundell2007changes}, \citet{kreider2012identifying}, \citet{pepper}, \citet{siddique}, \citet{de2017effect}, and  \citet{mentalhealth}). In light of \citet[Theorem 4]{robinson1975stability} and Example \ref{example_incons}, this approach is inapplicable when SC fails. In practice, researchers attempting to apply it to an LP with a small or empty interior of $\Theta_I$ may encounter frequent LP infeasibility in the bootstrap draws. 

In some cases, SC may be established. This is true, for example, if the bound of interest can be expressed as an intersection bound $V(\mu_0) = \max \{c_1,c_2, \dots, c_q\} = \min_{t \in \mathbb{R}} t ~~ \text{s.t. } t \geq c_i, ~ i \in [q]$, where $\mu_0=\text{vec}(1,\iota,c)$ (as in \citet{CLR}). Yet, even if SC holds, we demonstrate that bootstrap inference based on the plug-in estimator is not valid, unless further regularity conditions hold. This observation can be viewed as a generalization of the parameter-on-the-boundary problem of \citet{andrews1999estimation,andrews2000}.

\begin{definition}
    Let $\mathbb{D}$ and $\mathbb{E}$ be Banach spaces, and $f: \mathbb{D}_f \subseteq \mathbb{D} \to \mathbb{E}$. The map $f$ is said to be Hadamard directionally differentiable (H.d.d.) at $\upsilon \in \mathbb{D}_f$ tangentially to $\mathbb{D}_0 \subseteq \mathbb{D}$, if there is a continuous map $f'_\upsilon: \mathbb{D}_0 \to \mathbb{E}$, such that
    \begin{align*}
        \underset{n \to \infty}{\lim}\left \lVert \frac{f(\upsilon + t_n h_n) - f(\upsilon)}{t_n} - f'_{\upsilon}(h)\right \rVert_{\mathbb{E}} = 0,
    \end{align*}
    for all sequences $\{h_n\} \subset \mathbb{D}$ and $\{t_n\} \subset \mathbb{R}_+$ such that $t_n \to 0^+$, $h_n \to h \in \mathbb{D}_0$ as $n \to \infty$ and $\upsilon + t_n h_n \in \mathbb{D}_f$ for all $n$. If, moreover, $f'_\upsilon(h)$ is linear in $h$, the map $f$ is said to be fully Hadamard differentiable at $\upsilon$.
\end{definition}

For simplicity of exposition, we abstract from the case of `true equalities' in $\Theta_I$ and zero rows in $M$. The results extend trivially to this case if SC is defined accordingly.
\begin{lemma}[\citet{duan2020hadamard}]\label{hdd_of_lp}
    Under Assumption 1 and SC, $V(\cdot)$ is Hadamard directionally differentiable at $\mu_0$. The directional derivative is given by
    \begin{align}\label{hdd_lp}
        V'_{\mu_0}(h)=\underset{\theta\in\mathcal S(\mu_0)}{\inf}\;\underset{\lambda\in\Lambda(\mu_0)}{\sup}\left\{(h^p)'\theta+\lambda'(h^c-h^M\theta)\right\},
    \end{align}
    where $h=\text{vec}(h^p,h^M,h^c)\in\mathbb R^m$, with $h^p\in\mathbb R^d$, $h^M\in\mathbb R^{q\times d}$ and $h^c\in\mathbb R^q$, is the coefficient perturbation.
\end{lemma}

Hadamard directional differentiability of $V(\cdot)$ is sufficient for convergence in law.
\begin{prop}\label{dconv_of_lp}
    Under SC and Assumptions 1, 3, 4, it follows that
    \begin{align*}
        \sqrt{n}(V(\hat{\mu}_n) - V(\mu_0)) \xrightarrow{L} V'_{\mu_0}(\mathbb Z)
    \end{align*}
\end{prop}
\begin{proof}
    \citet{santos2019} Theorem 2.1. combined with Lemma \ref{hdd_of_lp}.  
\end{proof}
Gaussianity of $V'_{\mu_0}(\mathbb Z)$ is necessary for bootstrap consistency (\citet{santos2019}). Consequently, the empirical literature using bootstrap with the plug-in estimator has implicitly relied on this assumption. However, $V'_{\mu_0}(\mathbb Z)$ is not normal unless full Hadamard differentiability holds, i.e. $V'_{\mu_0}(h)$ is linear in $h$. As \eqref{hdd_lp} suggests, this is not generally the case. Theorem 3.1 in \citet{santos2019} establishes that bootstrap is inconsistent for the distribution when $V'_{\mu_0}(h)$ fails to be linear. The typically applied plug-in and bootstrap combination is then only valid under further restrictive assumptions\footnote{The full-support condition in Proposition \ref{prop_4} is imposed for expositional purposes. Sufficiency of conditions i, ii holds generally, whereas necessity obtains whenever the derivative $V'_{\mu_0}(h)$ is not linear over $\text{Supp}(\mathbb Z)$ when $\mathcal{S}(\mu_0), \Lambda(\mu_0)$ are not singletons, see Lemma \ref{lemma_lin_fun} in Appendix.}.
\begin{definition}[SMFCQ]
        Strict Mangasarian-Fromovitz Constraint Qualification (SMFCQ) holds at $\mu_0$ if $\exists \theta \in \mathcal{S}(\mu_0)$ and the corresponding $\lambda \in \Lambda(\mu_0)$, s.t.: i) $M_{\text{Supp}(\lambda)}$ has full row rank, and ii)$\exists z \in \mathbb{R}^d$ s.t. $M_{\text{Supp}(\lambda)}z = 0$, and $M_{J(\theta;\mu_0)\setminus\text{Supp}(\lambda)} z < 0$.
    \end{definition}
SMFCQ is equivalent to $|\Lambda(\mu_0)| = 1$ (see Proposition 1.1 in \citet{kyparisis1985uniqueness} and Corollary 2.X in \citet{mangasarian1978uniqueness}), which is necessary for bootstrap consistency. However, SMFCQ depends on both the vector $p$ and on $\Theta_I$. Proposition \ref{smfcq_prop} shows that \( |\Lambda(\cdot)| = 1 \) uniformly over all objective functions minimized at some \( \theta \in \Theta_I \) if and only if LICQ holds at \( \theta \). Both SMFCQ and LICQ are high-level and may be restrictive in models with potentially overidentifying constraints, such as in Section \ref{mc_section}.
\begin{prop}\label{smfcq_prop}
    Consider $\mu_0=\text{vec}(p,M,c)$ satisfying Assumption 1 and fix $\overline\theta\in\Theta_I$. Define $\mathcal U(\overline\theta)\equiv\{\mu=\text{vec}(\tilde p,M,c):\tilde p\in\mathbb R^d,~\overline\theta\in\mathcal S(\mu)\}$. Then,
    \begin{align*}
        |\Lambda(\mu)| = 1 ~ \forall \mu \in \mathcal U(\overline\theta) \iff \rk M_{J(\overline{\theta};\mu_0)} = |J(\overline{\theta};\mu_0)|, \text{i.e. LICQ holds at $\overline{\theta}$.}
    \end{align*}
    The claim above is also true if $\mu_0$ satisfies both Assumption 1 and SC. 
\end{prop}
\begin{proof}
    The implication $\impliedby$ is well-known (e.g. \citet{gafarov2024simple}). For $\implies$, set the class of functions $\mathcal{F}$ in \citet{WACHSMUTH201378} to be the class of linear functions parametrized by $\mathcal U(\overline\theta)$, and observe that for $\tilde{p} = \sum_{i \in J(\overline{\theta};\mu_0)} M_i'$, we have the corresponding $\mu \in \mathcal U(\overline\theta)$. The rest of the argument is identical to the proof of Theorem 2 in \citet{WACHSMUTH201378}.
\end{proof}
\begin{prop}\label{prop_4}
    If SC and Assumptions 1, 3, 4 hold, $\text{Supp}(\mathbb Z) = \mathbb{R}^m$ and a bootstrap coefficient vector $\hat\mu_n^*$ satisfies Assumption 3 in \citet{santos2019}, bootstrap is consistent in the sense that
    \begin{align*}
        \underset{f \in BL_1}{\sup} \left|\mathbb{E}[f(\sqrt{n}(V(\hat\mu_n^*) - V(\hat{\mu}_n)))|\mathcal{W}_n] - \mathbb{E}[f(V'_{\mu_0}(\mathbb Z))]\right| = o_p(1), \
    \end{align*}    
    where $BL_1 \equiv \{f:\mathbb{R}\to\mathbb{R} \text{ s.t. }|f(a)|\leq 1, |f(a) - f(b)|~  \leq |a - b| ~ \forall a, b \in \mathbb{R}\}$, if and only if i) NFF and ii) SMFCQ also hold at $\mu_0$. 
\end{prop}
\begin{proof}
    We will use the following lemma:
    \begin{lemma}\label{lemma_lin_fun}
        For a non-empty and bounded $X \subseteq \mathbb{R}^{d_y}$ and any $y \in \mathbb{R}^{d_y}$, let $G(y) \equiv \inf_{x \in X} y'x$. $G$ is linear iff $|X| = 1$.
    \end{lemma}
\begin{proof}
If $X=\{x\}$, then $G(y)=x'y$. Conversely, if $G$ is linear, write $G(y)=a'y$. For every $x\in X$, we have $a'y\le x'y$ for
all $y$. Taking $y=a-x$ gives $\|a-x\|^2 =0$.
\end{proof}
    Under SC, i) and ii) are equivalent to $|\mathcal{S}(\mu_0)| = |\Lambda(\mu_0)| = 1$. Note that i) is equivalent to $|\mathcal{S}(\mu_0)| = 1$ by definition. Suppose conditions i), ii) hold. Then, Proposition 1.1 in \citet{kyparisis1985uniqueness} establishes $|\Lambda(\mu_0)| = 1$. Conversely, suppose $|\mathcal{S}(\mu_0)| = |\Lambda(\mu_0)| = 1$. Then, Proposition 1.1 in \citet{kyparisis1985uniqueness} implies that SMFCQ must hold. Thus, by Theorem 3.1 in \citet{santos2019}, to prove the proposition it suffices to show that, under SC, $V'_{\mu_0}(h)$ is linear iff $|\mathcal{S}(\mu_0)| = |\Lambda(\mu_0)| = 1$. Sufficiency is immediate. To show necessity, denote the direction of the increment of $\mu$ by $h=\text{vec}(h^p,h^M,h^c)$. To see that $|\mathcal{S}(\mu_0)|=1$, consider $h=\text{vec}(h^p,0_{q\times d},0_q)$, giving $ V'_{\mu_0}(h) = \inf_{\theta \in \mathcal{S}(\mu_0)} (h^p)'\theta$ and apply Lemma \ref{lemma_lin_fun}. To see that $|\Lambda(\mu_0)|=1$, fix $h=\text{vec}(0_d,0_{q\times d},h^c)$, giving $V'_{\mu_0}(h) = \sup_{\lambda \in \Lambda(\mu_0)} \lambda'h^c$ and apply Lemma \ref{lemma_lin_fun}. This concludes the proof of the proposition.
\end{proof}
\subsection{Alternative approach to inference}\label{ap_alt_inf}
\subsubsection{LP under Slater's condition}\label{ap_inf_lp_sc}
One way to consistently estimate the plug-in estimator’s limiting distribution under SC is to employ the procedure developed in \citet{hong2015numerical}. Let
\begin{align}
    V'_n(\mathbb{Z}^*_n) \equiv \frac{V(\hat{\mu}_n + \epsilon_n \mathbb{Z}^*_n) - V(\hat{\mu}_n)}{\epsilon_n}
\end{align}
For $\epsilon_n \to 0$ with $\sqrt n\,\epsilon_n\to\infty$, we have the following proposition:
\begin{prop}\label{prop_f1}
    If SC and Assumptions 1, 3, 4 hold, and the bootstrap process $\mathbb{Z}^*_n$ satisfies the measurability and consistency conditions in \citet{hong2015numerical}, then
    \begin{align}
        \underset{f \in BL_1(\mathbb{R})}{\sup} |\mathbb{E}[f(V'_n(\mathbb{Z}^{*}_n))|\mathcal W_n] - \mathbb{E}[f(V'_{\mu_0}(\mathbb Z))]| = o_p(1)
    \end{align}
\end{prop}

Assumption SC may often be too strong, as argued in the main text. An even more serious problem in practice is that, even if an open ball is contained in $\Theta_I$ at $\mu_0$, the radius of that ball is not inconsequential in finite samples. A thinner identified set leads the bootstrap iterations of the N.D.M. to fail more often, as the constraint set turns empty at perturbed parameter values. Dropping the failed iterations introduces an unknown bias to the estimates, and is not advised. One solution would be to use $\tilde{\Theta}_I = \{\theta \in \Theta: M\theta \geq c - {k}\}$. Indeed, as long as the true system is feasible, decreasing the RHS of $M\theta \geq c$ renders the SC true, and the procedure in this section becomes applicable. The bias could be controlled via
\begin{align*}
    \underset{\Theta_I}{\min} ~ p'\theta - \lVert p \rVert d_H(\Theta_I, \tilde{\Theta}_I) \leq \underset{\tilde{\Theta}_I}{\min} ~ p'\theta \leq \underset{\Theta_I}{\min} ~ p'\theta,
\end{align*}
and using Lipschitz continuity of systems of linear inequalities, $d_H(\Theta_I, \tilde{\Theta}_I) \leq C\lVert {k} \rVert$ for some $C > 0$ depending on $\mu_0$, where the vector ${k} > 0$ is the RHS-expansion. This estimator, however, is still problematic both because it is rate-conservative, and because it relies on an arbitrarily selected ${k}$. Since a larger ${k}$ leads to a more conservative lower bound, in applied work the researcher would be tempted to select the minimal ${k}$ that ensures the bootstrap iterations do not fail. The statistical properties of that approach are unclear. 

\subsubsection{Inference for the biased penalty}\label{ap_bias_pen}
\begin{prop}\label{prop_penalty_hdd}
    The penalty function estimator, $V(\mu; w)$ is Hadamard directionally differentiable in $\mu$ at $\mu_0$ if either i) $\Theta$ is a polytope with $\text{Int}(\Theta) \ne \emptyset$, or ii) $\mathcal{S}(\mu_0; w) \subseteq \text{Int}(\Theta)$. The H.d.d. is given by:
    \begin{align}\label{hdd_penalty}
        V'_{\mu_0}(h; w)=\underset{\theta\in\mathcal S(\mu_0; w)}{\inf}\;\underset{\lambda\in\Lambda_w(\mu_0)}{\sup}\left\{(h^p)'\theta+\sum_{j=1}^{2^q}\lambda_j\sum_{i\in\Pi_j}w_i\bigl(h_i^c-h_i^M\theta\bigr)\right\},
    \end{align}
    where $h=\text{vec}(h^p,h^M,h^c)\in\mathbb R^m$, $\{\Pi_j\}_{j=1}^{2^q}$ enumerates $2^{[q]}$, and the upper-hemicontinuous correspondence $\Lambda_w:\mathbb R^m\to 2^{\mathbb R_+^{2^q}}$ is defined in the proof. Here $h_i^M$ is row $i$ of $h^M$.
\end{prop}
\begin{proof}
Throughout this proof $w$ is fixed. By Berge's maximum theorem, $V(\cdot; w)$ is continuous, and $\mathcal{S}(\cdot; w)$ is uhc.  Write $\chi_j(\theta;\mu,w)\equiv\tilde p'\theta+\sum_{i\in\Pi_j}w_i(\tilde c_i-\tilde M_i\theta)$, so $L(\theta;\mu,w)=\max_{j\in[2^q]}\chi_j(\theta;\mu,w)$. We have
    \begin{align}\label{penalty_lp}
        V(\mu; w) = \underset{t \in [\underline{t}, \overline{t}],\theta \in \Theta}{\min} ~ t \quad \text{s.t.: } \chi_j(\theta;\mu,w) \leq t, j \in [2^q] ,
    \end{align}
   for some $\underline{t}, \overline{t}$ s.t. $V(\mu_0;w) \in (\underline{t},\overline{t})$, given $\mu$ is close to $\mu_0$. At $\mu_0$, the feasible set of \eqref{penalty_lp} is compact, non-empty and contains an open set by continuity of $(\chi_j(\cdot;\mu_0,w))_j$ and as by either i), or ii), $\text{Int}(\Theta) \ne \emptyset$. Under i), this problem is an LP, while under ii) it is equivalent to an LP in an open neighborhood of $\mu_0$, since $\mathcal{S}(\cdot; w)$ is uhc. Thus, the problem at $\mu_0$ and around it, for H.d.d. calculation, is an LP satisfying SC. Suppose $\check\Lambda_w(\mu_0)$ is the set of Lagrange multipliers of \eqref{penalty_lp} at $\mu = \mu_0$, and $\Lambda_w(\mu_0)$ is its projection on the coordinates corresponding to the constraints $\chi_j(\theta;\mu_0,w) \leq t$ for all $j \in [2^q]$. Equation \eqref{hdd_penalty} is then established by direct computation and using Lemma \ref{hdd_of_lp}. 
\end{proof}
\begin{remark}
    By Lemma \ref{lemma_penalty}, Assumptions 1, 2 give ii) in the proposition if $\Theta_I \subseteq \text{Int}(\Theta)$.
\end{remark}
Under Assumption 2, exact inference can be obtained via Proposition \ref{prop_f1}. If Assumption 2 fails, the constructed CIs are valid for a lower bound, albeit potentially conservative. 

\subsection{Impossibility of uniform estimation}
\label{app:uniform_impossibility}

We formalize the impossibility of uniform estimation over the unrestricted class, and the insufficiency of imposing SC, LICQ, and NFF only pointwise. We do so for a fixed $p\ne 0$ and otherwise a full-rank Gaussian statistical experiment with $M$ and $c$ coefficients as its mean. For probability measures $\mathbb P,\mathbb Q$ on a common measurable space, $d_{\mathrm{TV}}(\mathbb P,\mathbb Q)$ denotes the total variation distance. Infima over estimators are restricted to measurable estimators based on an i.i.d. sample. $\P^{\otimes n}$ is a product measure, according to which an i.i.d. sample of size $n$ from $\P$ is distributed. The following Lemma is well-known and proven here for completeness. Its essence is Le Cam's binary testing argument (see e.g. \citet[Chapter 15]{wainwright2019high}).

\begin{lemma}[A necessary condition for uniform estimation]
\label{lem:tv_uniform_estimation}
Suppose that $T:\mathcal Q\to\mathbb R$ and there exists a sequence
$\mathbb P_k,\mathbb Q_k\in\mathcal Q$ and a fixed $\Delta>0$ such that, as $k \to \infty$,
\[
 d_{\mathrm{TV}}(\mathbb P_k,\mathbb Q_k)\longrightarrow0,
 \qquad
 |T(\mathbb P_k)-T(\mathbb Q_k)|\ge\Delta.
\]
For every $n\ge1$ and $0<\varepsilon<\Delta/2$,
\begin{equation}
\label{eq:tv_impossibility_bound}
 \inf_{\widehat T_n}\sup_{\mathbb P\in\mathcal Q}
 \mathbb P^{\otimes n}
 \bigl(|\widehat T_n-T(\mathbb P)|>\varepsilon\bigr)
 \ge\frac12,
\end{equation}
\end{lemma}

\begin{proof}
Fix $n$, $\widehat T_n$, and $k$, and let
$A_k=\{|\widehat T_n-T(\mathbb P_k)|\le\varepsilon\}$.
Since $2\varepsilon<\Delta$,
$A_k\subseteq\{|\widehat T_n-T(\mathbb Q_k)|>\varepsilon\}$.
Therefore,
\begin{align*}
 &\mathbb P_k^{\otimes n}
     \bigl[|\widehat T_n-T(\mathbb P_k)|>\varepsilon\bigr]
 +\mathbb Q_k^{\otimes n}
     \bigl[|\widehat T_n-T(\mathbb Q_k)|>\varepsilon\bigr] \ge\mathbb P_k^{\otimes n}[A_k']
             +\mathbb Q_k^{\otimes n}[A_k]\\
 &\quad\ge1-d_{\mathrm{TV}}
                  (\mathbb P_k^{\otimes n},\mathbb Q_k^{\otimes n})
 \ge1-n\,d_{\mathrm{TV}}(\mathbb P_k,\mathbb Q_k),
\end{align*}
where the second transition holds by definition of $\text{d}_{TV}$ and the final uses a well-known bound on the TV of product measures. The supremum in \eqref{eq:tv_impossibility_bound} is at least half
the last expression. Letting $k\to\infty$, and then taking the infimum over
$\widehat T_n$, proves the claim.
\end{proof}

Our goal is to combine the above lemma with a suitable family of LPs whose coefficients lie in the unrestricted coefficient class. Namely, fix $p\ne0$, a compact $\Theta\subseteq\mathbb R^d$ with nonempty interior,
and $q\ge d+2$. Let
\[
 \mathcal M_\Theta
 =\{\mu=\operatorname{vec}(M,c)\in\mathbb R^m:
          \varnothing\ne\Theta_I(\mu)\subseteq\Theta\}.
\]
This is the image of the unrestricted class of measures in Assumption 7 in the coefficient space. In it, we have the LP family constructed in line with Example \ref{example_incons} with the following properties. 

\begin{lemma}
\label{lem:impossibility_lp_path}
There exist an affine map $b\mapsto\mu(b)\in\mathcal M_\Theta$ for
$|b|<1/2$, and constants $v_\circ\in\mathbb R$, $\Delta>0$, such that
\begin{equation}
\label{eq:impossibility_value_path}
 V(\mu(b))=v_\circ-\Delta\mathbf 1\{b\ge0\}.
\end{equation}
For every $b\ne0$, SC and NFF hold, and exactly $d$ linearly independent
inequalities bind at each vertex of $\Theta_I(\mu(b))$.
\end{lemma}

\begin{proof}
We begin by considering $q = d+2$. First, suppose $d\ge2$. Write
$\theta=(x,y,z')'$, with $z\in\mathbb R^{d-2}$, and consider the LP
\begin{equation}
\label{eq:impossibility_lp_family}
 \min_{x,y,z} x
 \quad\text{subject to}\quad
 \begin{cases}
  y\ge(1+b)x,\\
  y\le x-\iota'z,\\
  z\ge0,\\
  y\ge-1,\\
  x\le1.
 \end{cases}
\end{equation}
When $d=2$, the $z$ terms are omitted. There are $d+2$ inequalities, and their coefficients are affine
in $b$. Note that the feasible sets are nonempty and contained in $[-1,1]^d$. To study the geometry of the feasible polytope, define
\[
 u=y-(1+b)x,\qquad v=x-y-\iota'z,\qquad \vartheta =(u,v,z')',
\]
For $b\ne0$, this is an invertible linear change of
coordinates, with $x=-\iota'\vartheta/b$ and $y=(1+b)x+u$. The polytope in $\vartheta$ coordinates for $b \ne 0$ is hence affinely isomorphic to the original one. The first $d$ inequalities are just $\vartheta \geq 0$. 

If $b < 0$, the inequality $y \geq - 1$ is slack everywhere on the feasible region, $x \leq 1$ is $\iota'\vartheta \leq -b$, so the polytope is a nondegenerate $d$-simplex
\[
 \{\vartheta\ge0:\iota'\vartheta\le-b\}.
\]
For it, SC holds and exactly $d$ linearly independent inequalities bind at each vertex. Moreover, the unique minimizer of $x=-\iota' \vartheta/b$ is $\vartheta = 0$, so NFF holds. Thus the unique solution of
\eqref{eq:impossibility_lp_family} is $(0,0,0)$, with value $0$.

If $b>0$, $x \leq 1$ is strictly slack, and $y \geq - 1$ gives $u + (1+b)(v+\iota'z) \leq b$. The feasible region is again a nondegenerate $d$-simplex
\[
 \{\vartheta\ge0:u+(1+b)(v+\iota'z)\le b\}.
\]
Once again, SC holds and exactly $d$ linearly independent inequalities bind at each vertex. Minimizing $x$ is equivalent to maximizing $\iota'\vartheta$. Since the coefficient of $u$ in the above constraint is
$1$, strictly below the coefficients $1+b$ of all other coordinates, the
unique optimizer is $u=b$, $v=0$, $z=0$, establishing NFF. It gives $(x,y,z)=(-1,-1,0)$ and value $-1$ in \eqref{eq:impossibility_lp_family}.

Finally, at $b=0$ it is immediate to see that $y = x$ and $z = 0$ over the feasible set. The latter becomes a line segment $\{(x,x,0):-1\le x\le1\}$, giving value $-1$ in \eqref{eq:impossibility_lp_family}.

For $d=1$, use instead the three inequalities
$-bx\ge0$, $x\ge-1$, and $-x\ge-1$, with objective $x$.
The feasible intervals are $[0,1]$ for $b<0$, $[-1,0]$ for $b>0$, and
$[-1,1]$ for $b=0$. The same value formula holds, and all three
pointwise conditions hold for $b\ne0$.

To extend the result to a general $p$ and $\Theta$, choose $\theta^\circ\in\operatorname{Int}(\Theta)$, and align $p$ with $e_1$ via an orthogonal
matrix $R$ with $R'p=\|p\|e_1$. Fix $r>0$ small enough that
$\theta^\circ+rR[-1,1]^d\subseteq\operatorname{Int}(\Theta)$.
The bijective affine transformation
$\theta=\theta^\circ+rR\xi$ can then be applied to the problems constructed above, preserving all geometric statements. The transformed coefficients remain affine in $b$, and the objective is
$p'\theta^\circ+r\|p\|\xi_1$. Thus
\eqref{eq:impossibility_value_path} holds with
$v_\circ=p'\theta^\circ$ and $\Delta=r\|p\|>0$.
For $q>d+2$, append fixed inequalities that are strictly slack throughout
$\Theta$. This leaves every conclusion unchanged.
\end{proof}

Our impossibility result will use the simple Gaussian experiment
\begin{equation}
\label{eq:gaussian_coefficient_experiment}
 W_1,\ldots,W_n\ \text{i.i.d.}\ \mathbb P_\mu=N_m(\mu,I_m),
 \qquad \mu\in\mathcal M_\Theta,
 \qquad
 \mathcal P_G=\{\mathbb P_\mu:\mu\in\mathcal M_\Theta\},
\end{equation}
and set $\mu_0(\mathbb P_\mu)=\mu$. This functional has exactly
the image in Assumption 7. The LICQ, SC and NFF subclasses of $\mathcal P_G$  are defined as in the main text. The following proposition formalizes our impossibility result. 
\begin{prop}[Impossibility result]
\label{prop:unrestricted_impossibility}
In experiment \eqref{eq:gaussian_coefficient_experiment}, there is a
$\Delta>0$ such that, for every $n\ge1$ and $0<\varepsilon<\Delta/2$,
\begin{align}
\label{eq:regular_impossibility}
 \inf_{\tilde V_n}\sup_{\mathbb P\in\mathcal P_G}
 \mathbb P^{\otimes n}
 \bigl(|\tilde V_n-V(\mu_0(\mathbb P))|>\varepsilon\bigr) \geq \\ \notag
 \inf_{\tilde V_n}\sup_{\mathbb P\in\mathcal P_G^{LICQ;0}\cap\mathcal P_G^{SC}\cap\mathcal P_G^{NFF}}
 \mathbb P^{\otimes n}
 \bigl(|\tilde V_n-V(\mu_0(\mathbb P))|>\varepsilon\bigr)
 \ge\frac12.
\end{align}
\end{prop}

\begin{proof}
For every $0<t<1/2$, Lemma \ref{lem:impossibility_lp_path} gives
$$\mathbb P_{\mu(t)},\mathbb P_{\mu(-t)}\in\mathcal P_{G,\mathrm{reg}} \equiv \mathcal P_G^{LICQ;0}\cap\mathcal P_G^{SC}\cap\mathcal P_G^{NFF} \subseteq \mathcal{P}_G.$$
Their values differ by $\Delta >0$, and $d_{\mathrm{TV}}(\mathbb P_{\mu(t)},\mathbb P_{\mu(-t)}) \to 0$ as $t \downarrow 0$ since $\mu(\cdot)$ is affine and Gaussian location laws are continuous
in total variation. Apply Lemma \ref{lem:tv_uniform_estimation} to conclude. 
\end{proof}


\bibliographystyle{ecta} 
\bibliography{literature}  

@article{MP2000,
 ISSN = {00129682, 14680262},
 URL = {http://www.jstor.org/stable/2999533},
 author = {Charles F. Manski and John V. Pepper},
 journal = {Econometrica},
 number = {4},
 pages = {997--1010},
 publisher = {[Wiley, Econometric Society]},
 title = {Monotone Instrumental Variables: With an Application to the Returns to Schooling},
 urldate = {2023-11-10},
 volume = {68},
 year = {2000}
}

@article{HENRY2012256,
title = {Set coverage and robust policy},
journal = {Economics Letters},
volume = {115},
number = {2},
pages = {256-257},
year = {2012},
issn = {0165-1765},
doi = {https://doi.org/10.1016/j.econlet.2011.12.071},
url = {https://www.sciencedirect.com/science/article/pii/S0165176511005945},
author = {Marc Henry and Alexei Onatski}
}

@article{freyberger2015identification,
  title={Identification and shape restrictions in nonparametric instrumental variables estimation},
  author={Freyberger, Joachim and Horowitz, Joel L},
  journal={Journal of Econometrics},
  volume={189},
  number={1},
  pages={41--53},
  year={2015},
  publisher={Elsevier}
}

@article{wollmann2018trucks,
  title={Trucks without bailouts: Equilibrium product characteristics for commercial vehicles},
  author={Wollmann, Thomas G},
  journal={American Economic Review},
  volume={108},
  number={6},
  pages={1364--1406},
  year={2018},
  publisher={American Economic Association 2014 Broadway, Suite 305, Nashville, TN 37203}
}

@article{holmes2011diffusion,
author = {Holmes, Thomas J.},
title = {The Diffusion of Wal-Mart and Economies of Density},
journal = {Econometrica},
volume = {79},
number = {1},
pages = {253-302},
doi = {https://doi.org/10.3982/ECTA7699},
url = {https://onlinelibrary.wiley.com/doi/abs/10.3982/ECTA7699},
eprint = {https://onlinelibrary.wiley.com/doi/pdf/10.3982/ECTA7699},
year = {2011}
}

@article{andrews2010inference,
  title={Inference for parameters defined by moment inequalities using generalized moment selection},
  author={Andrews, Donald WK and Soares, Gustavo},
  journal={Econometrica},
  volume={78},
  number={1},
  pages={119--157},
  year={2010},
  publisher={Wiley Online Library}
}

@book{Vaart_1998, place={Cambridge}, series={Cambridge Series in Statistical and Probabilistic Mathematics}, title={Asymptotic Statistics}, publisher={Cambridge University Press}, author={van der Vaart, A. W.}, year={1998}, collection={Cambridge Series in Statistical and Probabilistic Mathematics}}

@article{kaido2019confidence,
  title={Confidence intervals for projections of partially identified parameters},
  author={Kaido, Hiroaki and Molinari, Francesca and Stoye, J{\"o}rg},
  journal={Econometrica},
  volume={87},
  number={4},
  pages={1397--1432},
  year={2019},
  publisher={Wiley Online Library}
}

@article{staiger1997instrumental,
  author  = {Staiger, Douglas and Stock, James H.},
  title   = {Instrumental Variables Regression with Weak Instruments},
  journal = {Econometrica},
  year    = {1997},
  volume  = {65},
  number  = {3},
  pages   = {557--586}
}

@article{ddml,
    author = {Chernozhukov, Victor and Chetverikov, Denis and Demirer, Mert and Duflo, Esther and Hansen, Christian and Newey, Whitney and Robins, James},
    title = {Double/debiased machine learning for treatment and structural parameters},
    journal = {The Econometrics Journal},
    volume = {21},
    number = {1},
    pages = {C1-C68},
    year = {2018},
    month = {02},
    issn = {1368-4221},
    doi = {10.1111/ectj.12097},
    url = {https://doi.org/10.1111/ectj.12097},
    eprint = {https://academic.oup.com/ectj/article-pdf/21/1/C1/27684918/ectj00c1.pdf},
}

@article{bugni2017inference,
  title={Inference for subvectors and other functions of partially identified parameters in moment inequality models},
  author={Bugni, Federico A and Canay, Ivan A and Shi, Xiaoxia},
  journal={Quantitative Economics},
  volume={8},
  number={1},
  pages={1--38},
  year={2017},
  publisher={Wiley Online Library}
}

@article{cox2023simple,
  title={Simple adaptive size-exact testing for full-vector and subvector inference in moment inequality models},
  author={Cox, Gregory and Shi, Xiaoxia},
  journal={The Review of Economic Studies},
  volume={90},
  number={1},
  pages={201--228},
  year={2023},
  publisher={Oxford University Press}
}

@article{deb2023revealed,
  title={Revealed price preference: theory and empirical analysis},
  author={Deb, Rahul and Kitamura, Yuichi and Quah, John KH and Stoye, J{\"o}rg},
  journal={The Review of Economic Studies},
  volume={90},
  number={2},
  pages={707--743},
  year={2023},
  publisher={Oxford University Press US}
}

@techreport{bai2022testing,
  title={On testing systems of linear inequalities with known coefficients},
  author={Bai, Yuehao and Santos, Andres and Shaikh, Azeem M},
  year={2022},
  institution={Working Paper}
}

@article{fang2023inference,
  title={Inference for Large-Scale Linear Systems With Known Coefficients},
  author={Fang, Zheng and Santos, Andres and Shaikh, Azeem M and Torgovitsky, Alexander},
  journal={Econometrica},
  volume={91},
  number={1},
  pages={299--327},
  year={2023},
  publisher={Wiley Online Library}
}

@article{manski2002inference,
  title={Inference on regressions with interval data on a regressor or outcome},
  author={Manski, Charles F and Tamer, Elie},
  journal={Econometrica},
  volume={70},
  number={2},
  pages={519--546},
  year={2002},
  publisher={Wiley Online Library}
}

@article{CLR,
author = {Chernozhukov, Victor and Lee, Sokbae and Rosen, Adam M.},
title = {Intersection Bounds: Estimation and Inference},
journal = {Econometrica},
volume = {81},
number = {2},
pages = {667-737},
doi = {https://doi.org/10.3982/ECTA8718},
url = {https://onlinelibrary.wiley.com/doi/abs/10.3982/ECTA8718},
eprint = {https://onlinelibrary.wiley.com/doi/pdf/10.3982/ECTA8718},
year = {2013}
}

@article{CHT,
 ISSN = {00129682, 14680262},
 URL = {http://www.jstor.org/stable/4502031},
 author = {Victor Chernozhukov and Han Hong and Elie Tamer},
 journal = {Econometrica},
 number = {5},
 pages = {1243--1284},
 publisher = {[Wiley, The Econometric Society]},
 title = {Estimation and Confidence Regions for Parameter Sets in Econometric Models},
 urldate = {2023-11-13},
 volume = {75},
 year = {2007}
}

@article{honoreadriana2006,
  title={Bounds in competing risks models and the war on cancer},
  author={Honor{\'e}, Bo E and Lleras-Muney, Adriana},
  journal={Econometrica},
  volume={74},
  number={6},
  pages={1675--1698},
  year={2006},
  publisher={Wiley Online Library}
}

@article{santos,
  title={Using instrumental variables for inference about policy relevant treatment parameters},
  author={Mogstad, Magne and Santos, Andres and Torgovitsky, Alexander},
  journal={Econometrica},
  volume={86},
  number={5},
  pages={1589--1619},
  year={2018},
  publisher={Wiley Online Library}
}

@misc{semenova2023adaptive,
  author        = {Khandamiryan, Gevorg and Semenova, Vira},
  title         = {Adaptive Estimation of Aggregated Values of
                   Conditional Linear Programs},
  year          = {2026},
  eprint        = {2606.08359},
  archivePrefix = {arXiv},
  primaryClass  = {econ.EM},
  note          = {Accepted, The Review of Economics and Statistics},
  url           = {https://arxiv.org/abs/2606.08359}
}

@article{bhattacharya2009inferring,
  title={Inferring optimal peer assignment from experimental data},
  author={Bhattacharya, Debopam},
  journal={Journal of the American Statistical Association},
  volume={104},
  number={486},
  pages={486--500},
  year={2009},
  publisher={Taylor \& Francis}
}

@article{pepper,
title = {The impact of the National School Lunch Program on child health: A nonparametric bounds analysis},
journal = {Journal of Econometrics},
volume = {166},
number = {1},
pages = {79-91},
year = {2012},
note = {Annals Issue on ``Identification and Decisions'', in Honor of Chuck Manski's 60th Birthday},
issn = {0304-4076},
doi = {https://doi.org/10.1016/j.jeconom.2011.06.007},
url = {https://www.sciencedirect.com/science/article/pii/S0304407611001205},
author = {Craig Gundersen and Brent Kreider and John Pepper}
}

@article{hoffman1952approximate,
  author  = {Hoffman, Alan J.},
  title   = {On Approximate Solutions of Systems of Linear Inequalities},
  journal = {Journal of Research of the National Bureau of Standards},
  year    = {1952},
  volume  = {49},
  number  = {4},
  pages   = {263--265},
  month   = oct,
  doi     = {10.6028/jres.049.027}
}

@article{robinson1975stability,
  title={Stability theory for systems of inequalities. Part I: Linear systems},
  author={Robinson, Stephen M},
  journal={SIAM Journal on Numerical Analysis},
  volume={12},
  number={5},
  pages={754--769},
  year={1975},
  publisher={SIAM}
}

@article{houde2023nexus,
  title={Nexus Tax Laws and Economies of Density in E-Commerce: A Study of Amazon's Fulfillment Center Network},
  author={Houde, Jean-Fran{\c{c}}ois and Newberry, Peter and Seim, Katja},
  journal={Econometrica},
  volume={91},
  number={1},
  pages={147--190},
  year={2023},
  publisher={Wiley Online Library}
}

@article{kaido2014asymptotically,
  title={Asymptotically efficient estimation of models defined by convex moment inequalities},
  author={Kaido, Hiroaki and Santos, Andres},
  journal={Econometrica},
  volume={82},
  number={1},
  pages={387--413},
  year={2014},
  publisher={Wiley Online Library}
}

@article{bc2013,
author = {Alexandre Belloni and Victor Chernozhukov},
title = {{Least squares after model selection in high-dimensional sparse models}},
volume = {19},
journal = {Bernoulli},
number = {2},
publisher = {Bernoulli Society for Mathematical Statistics and Probability},
pages = {521 -- 547},
year = {2013},
doi = {10.3150/11-BEJ410},
URL = {https://doi.org/10.3150/11-BEJ410}
}

@article{ho2009,
Author = {Ho, Katherine},
Title = {Insurer-Provider Networks in the Medical Care Market},
Journal = {American Economic Review},
Volume = {99},
Number = {1},
Year = {2009},
Month = {March},
Pages = {393–430},
DOI = {10.1257/aer.99.1.393},
URL = {https://www.aeaweb.org/articles?id=10.1257/aer.99.1.393}}

@article{mentalhealth,
author = {Cygan-Rehm, Kamila and Kuehnle, Daniel and Oberfichtner, Michael},
title = {Bounding the causal effect of unemployment on mental health: Nonparametric evidence from four countries},
journal = {Health Economics},
volume = {26},
number = {12},
pages = {1844-1861},
doi = {https://doi.org/10.1002/hec.3510},
url = {https://onlinelibrary.wiley.com/doi/abs/10.1002/hec.3510},
eprint = {https://onlinelibrary.wiley.com/doi/pdf/10.1002/hec.3510},
year = {2017}
}

@article{kreider2012identifying,
  title={Identifying the effects of SNAP (food stamps) on child health outcomes when participation is endogenous and misreported},
  author={Kreider, Brent and Pepper, John V and Gundersen, Craig and Jolliffe, Dean},
  journal={Journal of the American Statistical Association},
  volume={107},
  number={499},
  pages={958--975},
  year={2012},
  publisher={Taylor \& Francis}
}

@article{siddique,
author = {Zahra Siddique},
title = {Partially Identified Treatment Effects Under Imperfect Compliance: The Case of Domestic Violence},
journal = {Journal of the American Statistical Association},
volume = {108},
number = {502},
pages = {504-513},
year = {2013},
publisher = {Taylor & Francis},
doi = {10.1080/01621459.2013.779836},
URL = { 
        https://doi.org/10.1080/01621459.2013.779836
},
eprint = { 
        https://doi.org/10.1080/01621459.2013.779836
}
}

@article{santos2019,
    author = {Fang, Zheng and Santos, Andres},
    title = "{Inference on Directionally Differentiable Functions}",
    journal = {The Review of Economic Studies},
    volume = {86},
    number = {1},
    pages = {377-412},
    year = {2019},
    month = {1},
    issn = {0034-6527},
    doi = {10.1093/restud/rdy049},
    url = {https://doi.org/10.1093/restud/rdy049},
    eprint = {https://academic.oup.com/restud/article-pdf/86/1/377/27285299/rdy049.pdf},
}

@article{chorussel2023,
author = {JoonHwan Cho and Thomas M. Russell},
title = {Simple Inference on Functionals of Set-Identified Parameters Defined by Linear Moments},
journal = {Journal of Business \& Economic Statistics},
volume={42}, 
number={2}, 
pages={563–578},
year={2024},
publisher = {Taylor & Francis},
doi = {10.1080/07350015.2023.2203768},
URL = { 
    https://doi.org/10.1080/07350015.2023.2203768
},
eprint = {    https://doi.org/10.1080/07350015.2023.2203768
}
}

@article{hong2015numerical,
  author = {Hong, Han and Li, Jessie},
  title = {The numerical delta method},
  journal = {Journal of Econometrics},
  volume = {206}, number = {2}, pages = {379--394},
  year = {2018}, doi = {10.1016/j.jeconom.2018.06.007}
}

@article{WACHSMUTH201378,
title = {On LICQ and the uniqueness of Lagrange multipliers},
journal = {Operations Research Letters},
volume = {41},
number = {1},
pages = {78-80},
year = {2013},
issn = {0167-6377},
doi = {https://doi.org/10.1016/j.orl.2012.11.009},
url = {https://www.sciencedirect.com/science/article/pii/S0167637712001459},
author = {Gerd Wachsmuth}
}

@article{duan2020hadamard,
  title={Hadamard directional differentiability of the optimal value of a linear second-order conic programming problem},
  author={Duan, Qingsong and Xu, Mengwei and Zhang, Liwei and Zhang, Sainan},
  journal={Journal of Industrial and Management Optimization},
  volume={17},
  number={6},
  pages={3085--3098},
  year={2021},
  doi={10.3934/jimo.2020108},
  publisher={Journal of Industrial and Management Optimization}
}

@techreport{mangasarian1978uniqueness,
  title={Uniqueness of solution in linear programming},
  author={Mangasarian, Olvi},
  year={1978},
  institution={University of Wisconsin-Madison Department of Computer Sciences}
}

@article{kyparisis1985uniqueness,
  title={On uniqueness of Kuhn-Tucker multipliers in nonlinear programming},
  author={Kyparisis, Jerzy},
  journal={Mathematical Programming},
  volume={32},
  pages={242--246},
  year={1985},
  publisher={Springer}
}

@article{blundell2007changes,
  title={Changes in the distribution of male and female wages accounting for employment composition using bounds},
  author={Blundell, Richard and Gosling, Amanda and Ichimura, Hidehiko and Meghir, Costas},
  journal={Econometrica},
  volume={75},
  number={2},
  pages={323--363},
  year={2007},
  publisher={Wiley Online Library}
}

@article{de2017effect,
  title={The Effect of Additional Funds for Low-ability Pupils: A Non-parametric Bounds Analysis},
  author={De Haan, Monique},
  journal={The Economic Journal},
  volume={127},
  number={599},
  pages={177--198},
  year={2017},
  publisher={Oxford University Press Oxford, UK}
}

@article{gafarov2024simple,
title = {Simple subvector inference on sharp identified set in affine models},
journal = {Journal of Econometrics},
volume = {249},
pages = {105952},
year = {2025},
issn = {0304-4076},
doi = {https://doi.org/10.1016/j.jeconom.2025.105952},
url = {https://www.sciencedirect.com/science/article/pii/S0304407625000065},
author = {Bulat Gafarov}
}

@article{tao2010random,
  title={Random matrices: The distribution of the smallest singular values},
  author={Tao, Terence and Vu, Van},
  journal={Geometric And Functional Analysis},
  volume={20},
  pages={260--297},
  year={2010},
  publisher={Springer}
}

@book{wainwright2019high,
  title={High-dimensional statistics: A non-asymptotic viewpoint},
  author={Wainwright, Martin J},
  volume={48},
  year={2019},
  publisher={Cambridge university press}
}

@article{li1993sharp,
  title={The sharp Lipschitz constants for feasible and optimal solutions of a perturbed linear program},
  author={Li, Wu},
  journal={Linear algebra and its applications},
  volume={187},
  pages={15--40},
  year={1993},
  publisher={Elsevier}
}

@book{aliprantis2007infinite,
  title={Infinite Dimensional Analysis: A Hitchhiker's Guide},
  author={Aliprantis, C.D. and Border, K.C.},
  isbn={9783540326960},
  lccn={01370976},
  url={https://books.google.com/books?id=4hIq6ExH7NoC},
  year={2007},
  publisher={Springer}
}

@article{andrews2023,
    author = {Andrews, Isaiah and Roth, Jonathan and Pakes, Ariel},
    title = "{Inference for Linear Conditional Moment Inequalities}",
    journal = {The Review of Economic Studies},
    volume = {90},
    number = {6},
    pages = {2763-2791},
    year = {2023},
    month = {01},
    issn = {0034-6527},
    doi = {10.1093/restud/rdad004},
    url = {https://doi.org/10.1093/restud/rdad004},
    eprint = {https://academic.oup.com/restud/article-pdf/90/6/2763/52976661/rdad004.pdf},
}

@article{andrews1999estimation,
  title={Estimation when a parameter is on a boundary},
  author={Andrews, Donald WK},
  journal={Econometrica},
  volume={67},
  number={6},
  pages={1341--1383},
  year={1999},
  publisher={Wiley Online Library}
}

@article{andrews2000,
 ISSN = {00129682, 14680262},
 URL = {http://www.jstor.org/stable/2999432},
 author = {Donald W. K. Andrews},
 journal = {Econometrica},
 number = {2},
 pages = {399--405},
 publisher = {[Wiley, Econometric Society]},
 title = {Inconsistency of the Bootstrap When a Parameter is on the Boundary of the Parameter Space},
 urldate = {2025-02-12},
 volume = {68},
 year = {2000}
}

@article{romanoshaikh2008,
title = {Inference for identifiable parameters in partially identified econometric models},
journal = {Journal of Statistical Planning and Inference},
volume = {138},
number = {9},
pages = {2786-2807},
year = {2008},
note = {Special Issue in Honor of Theodore Wilbur Anderson, Jr. on the Occasion of his 90th Birthday},
issn = {0378-3758},
doi = {https://doi.org/10.1016/j.jspi.2008.03.015},
url = {https://www.sciencedirect.com/science/article/pii/S0378375808000694},
author = {Joseph P. Romano and Azeem M. Shaikh}
}

@article{tebalditorgovitskyyang2023,
 ISSN = {00129682, 14680262},
 URL = {https://www.jstor.org/stable/48744852},
 author = {Pietro Tebaldi and Alexander Torgovitsky and Hanbin Yang},
 journal = {Econometrica},
 number = {1},
 pages = {pp. 107--146},
 publisher = {[Wiley, The Econometric Society]},
 title = {NONPARAMETRIC ESTIMATES OF DEMAND IN THE CALIFORNIA HEALTH INSURANCE EXCHANGE},
 urldate = {2026-09-21},
 volume = {91},
 year = {2023}
}

@article{kaido2017calibrated,
  title={Calibrated Projection in MATLAB: Users' Manual},
  author={Kaido, Hiroaki and Molinari, Francesca and Stoye, J{\"o}rg and Thirkettle, Matthew},
  journal={arXiv preprint arXiv:1710.09707},
  year={2017}
}

@article{chernozhukov2019inference,
  title={Inference on causal and structural parameters using many moment inequalities},
  author={Chernozhukov, Victor and Chetverikov, Denis and Kato, Kengo},
  journal={The Review of Economic Studies},
  volume={86},
  number={5},
  pages={1867--1900},
  year={2019},
  publisher={Oxford University Press}
}

@book{grunbaum1967convex,
  title={Convex polytopes},
  author={Gr{\"u}nbaum, Branko and Klee, Victor and Perles, Micha A and Shephard, Geoffrey Colin},
  volume={16},
  year={1967},
  publisher={Springer}
}

@article{romano2010inference,
  title={Inference for the identified set in partially identified econometric models},
  author={Romano, Joseph P and Shaikh, Azeem M},
  journal={Econometrica},
  volume={78},
  number={1},
  pages={169--211},
  year={2010},
  publisher={Wiley Online Library}
}

@article{manskiimbens2004,
author = {Imbens, Guido W. and Manski, Charles F.},
title = {Confidence Intervals for Partially Identified Parameters},
journal = {Econometrica},
volume = {72},
number = {6},
pages = {1845-1857},
doi = {https://doi.org/10.1111/j.1468-0262.2004.00555.x},
url = {https://onlinelibrary.wiley.com/doi/abs/10.1111/j.1468-0262.2004.00555.x},
eprint = {https://onlinelibrary.wiley.com/doi/pdf/10.1111/j.1468-0262.2004.00555.x},
year = {2004}
}

@article{syrgkanis2017inference,
  title={Inference on auctions with weak assumptions on information},
  author={Syrgkanis, Vasilis and Tamer, Elie and Ziani, Juba},
  journal={arXiv preprint arXiv:1710.03830},
  year={2021}
}

@article{klinetartari,
Author = {Kline, Patrick and Tartari, Melissa},
Title = {Bounding the Labor Supply Responses to a Randomized Welfare Experiment: A Revealed Preference Approach},
Journal = {American Economic Review},
Volume = {106},
Number = {4},
Year = {2016},
Month = {April},
Pages = {972–1014},
DOI = {10.1257/aer.20130824},
URL = {https://www.aeaweb.org/articles?id=10.1257/aer.20130824}}

@article{honoretamer2006,
  title={Bounds on parameters in panel dynamic discrete choice models},
  author={Honor{\'e}, Bo E and Tamer, Elie},
  journal={Econometrica},
  volume={74},
  number={3},
  pages={611--629},
  year={2006},
  publisher={Wiley Online Library}
}

@article{hirano,
author = {Hirano, Keisuke and Porter, Jack R.},
title = {Impossibility Results for Nondifferentiable Functionals},
journal = {Econometrica},
volume = {80},
number = {4},
pages = {1769-1790},
doi = {https://doi.org/10.3982/ECTA8681},
url = {https://onlinelibrary.wiley.com/doi/abs/10.3982/ECTA8681},
eprint = {https://onlinelibrary.wiley.com/doi/pdf/10.3982/ECTA8681},
year = {2012}
}
\end{document}